\documentclass[twoside,11pt]{article}

\usepackage{blindtext}

\usepackage[preprint]{jmlr2e}
\usepackage{amsmath}
\usepackage{algorithm,algpseudocode}
\usepackage{lscape,multirow,afterpage,longtable,tablefootnote}

\usepackage{cleveref}
\crefname{theorem}{Theorem}{Theroems}
\crefname{lemma}{Lemma}{Lemmas}
\crefname{corollary}{Corollary}{Corollaries}
\crefname{proposition}{Proposition}{Propositions}
\crefname{assumption}{Assumption}{Assumptions}
\crefname{algorithm}{Algorithm}{Algorithms}
\crefname{section}{Section}{Sections}
\crefname{subsection}{Section}{Sections}
\crefname{figure}{Figure}{Figures}
\crefname{table}{Table}{Tables}

\newtheorem{assumption}{Assumption}

\newcommand\independent{\protect\mathpalette{\protect\independenT}{\perp}}
\def\independenT#1#2{\mathrel{\rlap{$#1#2$}\mkern2mu{#1#2}}}
\DeclareMathOperator*{\argmax}{arg\,max}
\DeclareMathOperator*{\argmin}{arg\,min}

\usepackage{lastpage}
\firstpageno{1}

\begin{document}

\title{Learning Optimal Dynamic Treatment Regimens Subject to Stagewise Risk Controls}

\author{\name Mochuan Liu \email mochuan@live.unc.edu \\
       \addr Department of Biostatistics\\
       University of North Carolina at Chapel Hill\\
       Chapel Hill, NC 27599, USA
       \AND
       \name Yuanjia Wang \email yw2016@cumc.columbia.edu \\
       \addr Department of Biostatistics\\
       Columbia University\\
       New York, NY 10032, USA
       \AND 
       \name Haoda Fu \email fu\_haoda@lilly.com \\
       \addr Eli Lilly and Company\\
       Indianapolis, IN 46285, USA
       \AND 
       \name Donglin Zeng \email dzeng@umich.edu\\
       \addr Department of Biostatistics\\
       University of Michigan \\
       Ann Arbor, MI 48109, USA}

\editor{David Sontag}

\maketitle

\begin{abstract}%   <- trailing '%' for backward compatibility of .sty file
Dynamic treatment regimens (DTRs) aim at tailoring individualized sequential treatment rules that maximize cumulative beneficial outcomes by accommodating patients' heterogeneity in decision-making. For many chronic diseases including type 2 diabetes mellitus (T2D), treatments are usually multifaceted in the sense that aggressive treatments with a higher expected reward are also likely to elevate the risk of acute adverse events. In this paper, we propose a new weighted learning framework, namely benefit-risk dynamic treatment regimens (BR-DTRs), to address the benefit-risk trade-off. The new framework relies on a backward learning procedure by restricting the induced risk of the treatment rule to be no larger than a pre-specified risk constraint at each treatment stage. Computationally, the estimated treatment rule solves a weighted support vector machine problem with a modified smooth constraint. Theoretically, we show that the proposed DTRs are Fisher consistent, and we further obtain the convergence rates for both the value and risk functions. Finally, the performance of the proposed method is demonstrated via extensive simulation studies and application to a real study for T2D patients.
\end{abstract}

\begin{keywords}
  Dynamic treatment regimens, Precision medicine, Benefit-risk tradeoff, Acute adverse events, Weighted support vector machine.
\end{keywords}

\section{Introduction}
\label{sec:intro}

%%%%%%%%%%%%%%%%%% Section 1 %%%%%%%%%%%%%%%%%%%%%%%%%%%

Precision medicine aims at tailoring treatments to individual patients by taking their clinical heterogeneity into consideration \citep{hodson_precision_2016,ginsburg_precision_2018}. One important treatment strategy in precision medicine is called dynamic treatment regimens (DTRs), which sequentially assign treatments to individual patients based on their evolving health status and intermediate responses \citep{chakraborty_dynamic_2014}, with the goal of maximizing their long-term rewarding outcome. Over the past years, there has been an explosive development of statistical methods and machine learning algorithms for learning DTRs using either randomized trials \citep{murphy_experimental_2005,dawson_efficient_2012,lei_smart_2012}, or observational data \citep{rosthoj_estimation_2006,moodie_q-learning_2012}. Among them, regression-based methods, such as A-learning \citep{murphy_optimal_2003,blatt_technical_2004}, G-estimation \citep{robins_optimal_2004}, regret regression \citep{henderson_regret-regression_2010}, Q-learning \citep{qian_performance_2011,ma_learning_2022}, and doubly robust regression \citep{zhang_estimating_2012,barrett_doubly_2014}, fit regression models to estimate expected future outcome at each stage, or its varied forms such as blip functions or regrets, and obtain the optimal DTRs by comparing the model-predicted outcomes among treatments in a backward fashion. To lessen the concern of model misspecification, machine learning-based approaches have also been advocated to learn the optimal DTRs by directly optimizing the so-called value function. Examples of machine learning-based methods include outcome weighted learning (OWL) \citep{zhao_estimating_2012,zhao_new_2015} and its doubly robust extension \citep{liu_augmented_2018}, which connects the value optimization problem to a weighted classification problem that can be solved efficiently through support vector machines. 

For many chronic diseases, treatments are multifaceted: the aggressive treatment with a better reward is often accompanied by higher toxicity, leading to the elevated risk of severe and acute side effects or even fatality. For example, the Standards of Medical Care in Diabetes published by the American Diabetes Association (ADA) suggests metformin as first-line initial therapy for all general T2D patients. Intensified insulin therapy should be applied to patients when the patients' A1C level is above the target \citep{american_diabetes_association_pharmacologic_2022}. However, evidence has indicated that many patients who may eventually rely on insulin therapy to achieve ideal A1C level will be likely to experience more hypoglycemic episodes \citep{ukpds_group_intensive_1998}, and the latter can cause neurological impairments, coma, or death \citep{cryer_hypoglycemia_2003}. Thus, the benefit-risk challenge presented in chronic diseases such as T2D entails that the ideal treatment rules should also consider reducing any short-term risks during each decision stage while maximizing the long-term rewarding outcome. 

Only a limited number of existing works in DTRs have ever considered the benefit-risk balance, and most of them are restricted to the single-stage decision-making problem. Among them, most of the methods prespecified a utility function to unify the benefit and risk into one composite outcome and proposed to learn optimal decision through maximizing the utility function \citep{lee_bayesian_2015,butler_incorporating_2018}. A major disadvantage of these utility-based approaches is that the choice of the utility function is often subjective and cannot yield decision rules that strictly control the risks. More recently, \citet{wang_learning_2018} reformulated the problem into a constrained optimization problem that maximizes the reward outcome subject to a risk constraint and developed a weighted learning framework for solving the optimal rule. However, no theoretical justification was provided for the proposed method, and extending the framework to study DTRs with multiple-stage risks is nontrivial. Computationally intensive methods have also been proposed to estimate the optimal rules for either competing risks or under a single safety constraint \citep{laber_set-valued_2014,laber_identifying_2018}, but these methods cannot be easily extended to multiple constraints. 

{Learning optimal DTRs under the constraints is closely related to constrained and safe reinforcement learning (RL) \citep{garcia_comprehensive_2015,zhao_state-wise_2023} and multi-objective RL \citep{hayes_practical_2022}, which has attracted much interest in the RL field recently. Most of the safe RL algorithms consider online RL problems and the multi-objective RL aims to learn a policy to achieve the so-called Pareto optimality to balance different outcomes. Furthermore, all these methods are either designed for the problems with a finite state space \citep{van_moffaert_multi-objective_2014,chow_risk-constrained_2018,fei_risk-sensitive_2020,kalagarla_sample-efficient_2021,wei_provably-efficient_2021,bura_dope_2022,deng_reinforcement_2023}, rely on strong parametric assumptions over policy, transition model and outcome \citep{mahdavi_trading_2012,achiam_constrained_2017,yang_generalized_2019,yu_convergent_2019,bohez_value_2019,ding_provably_2021,amani_safe_2021}, or study more restrictive bandit problems \citep{amani_linear_2019,moradipari_stage-wise_2020,khezeli_safe_2020}. Other methods consider deep neural networks \citep{ray_benchmarking_2019,ma_feasible_2021} but require a large sample size. Some safe RL for the offline data, including \citet{geibel_risk-sensitive_2005,le_batch_2019,paternain_safe_2023}, are either designed for infinite horizon problems \citep{geibel_risk-sensitive_2005,le_batch_2019} or rely on strong model assumptions \citep{paternain_safe_2023} under CMDP assumption. Compared to all existing work, the DTRs problem to be studied in this work is an offline RL problem with finite horizons, and the applications usually contain only a few hundred subjects.}

To address the real-world challenge of treating chronic diseases, in this work, we consider the problem of learning the optimal DTRs in a multistage study, subject to different acute risk constraints at each stage. We develop a general framework, namely benefit-risk DTRs (BR-DTRs), using the finding that under additional acute risk assumption, the stagewise benefit-risk DTRs can be decomposed into a series of single-stage benefit-risk problem only involving the risk restriction of the current stage. Numerically, we propose a backward procedure to estimate the optimal treatment rules: at each stage, we maximize the expected value function under the risk constraint imposed at the current stage, where the solution can be obtained by solving a constrained support vector machine problem. Theoretically, we show that the resulting DTRs are Fisher consistent when some proper surrogate functions are used to replace the objective function and risk constraints. We further derive the non-asymptotic error bounds for the cumulative reward and stagewise risks associated with the estimated DTRs.

Our contributions are two-fold: first, we propose a general framework to estimate the optimal DTRs under the stagewise risk constraints. We note that the proposed framework reduces to the outcome weighted learning for DTRs in \citet{zhao_new_2015} when there is no risk constraint and reduces to the method in \citet{wang_learning_2018} when there is only one stage. When stagewise risk restrictions are imposed, we show that the backward induction technique adopted in \citet{zhao_new_2015} along with the single-stage framework proposed in \citet{wang_learning_2018} can be jointly used to solve the optimal DTRs under the stagewise risk constraints. We note that such extension is nontrivial since the treatment of each stage is entangled with unknown treatments of the previous stage through risk constraints when the backward induction technique is used. Hence, additional theoretical justification is needed to rigorously prove that the problem can be decomposed into a series of constrained optimal treatment regimen problems of the current stage under acute risk assumption. Second, our work establishes the non-asymptotic results for the estimated DTRs for both value and risk functions, and such results have never been given before. In particular, we show that support vector machines still yield Fisher consistent treatment rules under a range of risk constraints. Our theory also shows that the convergence rate of the predicted value function is in the order of the cubic root of the sample size, and the convergence rate for the risk control has the order of the square root of the sample size. 

The remaining paper is organized as follows. In Section 2, we discuss the statistical framework of BR-DTRs and give the complete BR-DTRs algorithm. In Section 3, we provide further theoretical justification for BR-DRTs. We demonstrate the performance of BR-DTRs via simulation studies in Section 4 and apply the method to a real study of T2D patients in Section 5.

%%%%%%%%%%%%%%%%%% Section 2 %%%%%%%%%%%%%%%%%%%%%%%%%%%

\section{Method}
\label{sec:method}

%% 2-1
\subsection{DTRs under stagewise risk constraints}
\label{sec:2_1}
Consider a $T$-stage DTRs problem and we use $(Y_1,...,Y_T)$ to denote the beneficial reward and $(R_1,...,R_T)$ to denote the risk outcomes at each stage. We assume that $\{(Y_t,R_t)\}_{t=1}^T$ are bounded random variables and a series of dichotomous treatment options are available at each stage, denoted by $A_t\in\{-1,+1\}$. Let $H_1\subset\cdots\subset H_T$ be the feature variables at stage $t$, which includes the baseline prognostic variables, intermediate outcomes, and any time-dependent covariates information prior to stage $t$. {In this work, we further assume that the data is collected from a sequential multiple assignment randomized trial (SMART) \citep{murphy_experimental_2005} so the treatment assignment probability $\{p(A_t|H_t)\}_{t=1}$ is known for $t=1,...,T$. Extension to observational studies with unknown treatment assignment probability is discussed in Section 3 (e.g., Remark \ref{remark7}).} DTRs are defined as a sequence of functions 
$$\mathcal{D}=\mathcal{D}_1\times\cdots\times\mathcal{D}_T:\mathcal{H}_1\times\cdots\times\mathcal{H}_T\rightarrow \{-1,+1\}^T~~\text{where}~~\mathcal{D}_t:\mathcal{H}_t	\mapsto\{-1,+1\}.$$
The goal of BR-DTRs is to find the optimal rule $\mathcal{D}^*$ that maximizes the cumulative reward at the final stage $T$, while the risk at each stage $t$ is controlled by a pre-specified risk constraint, denoted by $\tau_t$. Mathematically, we aim to solve the following optimization problem 
\begin{equation*}
  \begin{alignedat}{4}
    \max_{\mathcal{D}}&\quad &E^{\mathcal{D}}[&\sum_{t=1}^TY_t]\\
    \text{subject}&\text{ to}\quad &E^{\mathcal{D}}&[R_1]\leq \tau_1, \ldots, E^{\mathcal{D}}&[R_T]\leq \tau_T,
  \end{alignedat}
\end{equation*}
where $E[\cdot]$ {denotes the expectation taken w.r.t. the joint distribution of $\{(A_t,H_t,Y_t,R_t)\}_{t=1}^T$ and $E^{\mathcal{D}}[\cdot]$ denotes the expectation given $A_t=\mathcal{D}_t(H_t)$ for $t=1,..,T$.} 

Additional assumptions are necessary to ensure that the above problem can be solved using the observed data. To this end, we let $\bar{A}_t=(A_1,..,A_t)$ denote the observed treatment history and $\bar{a}_t=(a_1,...,a_t)\in\{-1,+1\}^t$ denote any fixed treatment history up to time $t$, and use $X(\bar a_t)$ to denote the potential outcome of variable $X$ under treatment $\bar a_t$. 

\begin{assumption}{(Stable Unit Treatment Value (SUTV))}
  \label{SUTV}
  At each stage, subjects' outcomes are not influenced by other subjects' treatment allocation, i.e., 
  $$(Y_t,R_t)=(Y_t(\bar{a}_t),R_t(\bar{a}_t))\quad\text{given}\quad \bar{A}_t=\bar{a}_t.$$
\end{assumption}

\begin{assumption}{(No Unmeasured Confounders (NUC))}
  \label{NUC}
  For any $t=1,...,T,$
  $$A_t\independent (H_{t+1}(\bar{a}_t),...,H_{T}(\bar{a}_{T-1}),Y_T(\bar{a}_T),R_T(\bar{a}_T))\big|H_t.$$
\end{assumption}

\begin{assumption}{(Positivity)}
  \label{Positivity}
  For any $t=1,...,T$, there exists universal constants $0<c_1\leq c_2<1$ such that 
  $$c_1\leq p(A_t=1|H_t)\leq c_2\quad\text{for~}H_t\text{ a.s.}$$
\end{assumption}

\begin{assumption}{(Acute Risk)}
  \label{ACUTE}
  For any $t=1,...,T$ and $\bar{a}_t\in\{-1,1\}^t$, $R_t(\bar a_t)$ only depends on $a_t$. {In other words, for potential outcome $R_t(\bar a_t)$ we have $R_t(\bar a_t)=R_t(a_t)$. }
\end{assumption}

\cref{SUTV,NUC,Positivity} are standard causal assumptions for DTRs literature and one could refer to \citet{rubin_bayesian_1978,robins_causal_1997,chakraborty_statistical_2013} for more discussions. In particular, \cref{NUC,Positivity} hold if the data are obtained from a SMART. \cref{ACUTE} captures the acute risk property of chronic diseases. That is, for the same individual, the adverse risk in each stage is caused by his/her most recent treatment. {We note that \cref{ACUTE} does not imply $R_t$ is Markovian and independent of $H_t$. In general, $R_t$ will be a function of $H_t$ for $t=1,..,T$.} As an additional note, we can further assume that $R_t$ is positive and bounded away from zero after shifting both $R_t$ and $\tau_t$ by one same constant without changing the problem of interest.

Under all four additional assumptions and suppose $\mathcal{D}_t(H_t)=\text{sign}(f_t(H_t))$ for some measurable decision function $f_t$, we note that
\begin{equation}
  \label{equation_after_assumption}
  \begin{split}
    E^{\mathcal{D}}[R_t]=E\bigg[&\frac{R_t\prod_{t=1}^T\mathbb{I}(A_tf_t(H_t)>0)}{\prod_{t=1}^Tp(A_t|H_t)}\bigg]=E\bigg[R_t(\textrm{sign}(f_1), ....,\textrm{sign}(f_t))\bigg]\\
    &=E\bigg[R_t(\textrm{sign}(f_t))\bigg]=E\bigg[\frac{R_t \mathbb{I}(A_tf_t(H_t)>0)}{p(A_t|H_t)}\bigg].
  \end{split}
\end{equation}
Then according to \citet{zhao_new_2015}, the original problem can be reformulated as
\begin{equation}
  \label{ow_prob}
  \begin{alignedat}{4}
    \underset{{(f_1,...,f_T)\in\mathcal{F}_1\times\cdots\mathcal{F}_T}}{\max}&E\bigg[\frac{(\sum_{t=1}^TY_t)\prod_{t=1}^T\mathbb{I}(A_tf_t(H_t)>0)}{\prod_{t=1}^Tp(A_t|H_t)}\bigg]&\\
    \text{subject~to~~}&E\bigg[\frac{R_t\mathbb{I}(A_tf_t(H_t)>0)}{p(A_t|H_t)}\bigg]\leq \tau_t,&t=1,...,T,\\
  \end{alignedat}
\end{equation}
where $\mathcal{F}_t$ denotes the set of all real value measurable functions from $\mathcal{H}_t\rightarrow\mathbb{R}$ {and we note that $\{f_t\}_{t=1}^T$ here is identifiable up to a positive scale.} {To solve the problem (\ref{ow_prob}), we borrow the idea of \citet{zhao_new_2015} and introduce the backward induction technique to further decompose the BR-DTRs problem into a series of single-stage single-constraint problems. Let $\{\mathcal{O}_{t}\}_{t=1}^T$ denote the feasible region of the original problem under risk constraints $(\tau_1,...,\tau_T)$ at stage $t$, i.e.,}
$$\mathcal{O}_t=\bigg\{f\in\mathcal{F}_t\bigg|E\bigg[\frac{R_t\mathbb{I}(A_tf(H_t)>0)}{p(A_t|H_t)}\bigg]\leq \tau_t\bigg\},\quad t=1,...,T,$$
and define the $U$-function as
\begin{equation*}
  \begin{split}
    U_t(h_t;g_t,g_{t+1},...,g_T):=E\bigg[\frac{(\sum_{s=t}^TY_s)\prod_{s=t}^T\mathbb{I}(A_sg_s(H_s)>0)}{\prod_{s={t}}^Tp(A_s|H_s)}\bigg|H_t=h_t\bigg],
  \end{split}
\end{equation*}
where we set $U_{T+1}=0$, then we consider the following iterative optimization problems and their associated optimal solution, denoted by $(g_1^*,...,g_T^*)$, defined via
\begin{equation}
  \label{f_zero_one}
  g_{t}^*=\argmax_{f_t\in\mathcal{O}_t}E\bigg[\frac{(Y_t+U_{t+1}(H_{t+1};g_{t+1}^*,...,g_T^*))\mathbb{I}(A_tf_t(H_t)>0)}{p(A_t|H_t)}\bigg].
\end{equation}
{When there is no risk constraint, (\ref{f_zero_one}) will reduce to the standard OWL framework which is guaranteed to yield optimal solutions for the unconstrained problem following the similar idea as the Bellman equation and Q-learning \citep{bellman_dynamic_1966,qian_performance_2011}. However, extending the backward induction technique to risk-constrained DTRs problems is nontrivial, and the backward induction usually does not yield the optimal solutions for the general problem since the estimation of the treatment of each stage is entangled with unknown treatments from previous stages via the risk constraints. As one of our major contributions, our later proof for \cref{fisher_consistency} shows that the backward algorithm (\ref{f_zero_one}) leads to the optimal solutions of the BR-DTRs problem. To the best of our knowledge, our work is the first to provide the necessary conditions for the optimality of the implementation of the backward induction for stagewise risk-constrained DTRs problems.}

\begin{remark}
 We note that the choice of decision functions $\{f_t\}_{t=1}^T$ has no restriction and can be chosen from any function class during the estimation such as tree-based models, neural networks, or functions from a reproducing kernel Hilbert space (the last is studied in our work). Also, the definition of $\mathcal{A}_t=\{-1,+1\}$ is only a generic notation, which can depend on $H_{t-1}$ and refer to different treatments in different stages. As an extreme case, $\mathcal{A}_t$ can degenerate to a single treatment in certain stages. In this case, our method can still be applied by restricting the estimation to remaining patients who can receive alternative treatments. This allows us to extend our method to more complicated applications as shown in \cref{sec:application}.
\end{remark}

%% 2-2
\subsection{Surrogate loss and Fisher consistency}
\label{sec:2_2}
One main difficulty of implementing framework (\ref{f_zero_one}) is the existence of the indicator functions in both the objective function and risk constraints, which makes solving the original problem NP-hard. Following the idea in \citet{wang_learning_2018}, we propose the following surrogate functions to replace both indicator functions: let $\phi(\cdot)$ denote the hinge loss function defined as $\phi(x)=(1-x)_{+}$ and $\psi(\cdot,\eta)$ denote the shifted ramp loss function given by 
\begin{equation}
  \notag
  \psi(x,\eta)=\begin{cases}
    1,&\text{if}~x\geq0\\
    \frac{x+\eta}{\eta},&\text{if}~x\in(-\eta,0)\\
        0,&\text{if}~x\leq-\eta,
  \end{cases}
\end{equation}
where $\eta\in(0,1]$ is a prespecified shifting parameter that can vary with stage. We then consider the following surrogate problem, namely the BR-DTRs problem,
\begin{equation}
  \label{f_surrogate}
  f_{t}^*=\argmin_{f_t\in\mathcal{A}_{t}}E\bigg[\frac{(Y_t+U_{t+1}(H_{t+1};f_{t+1}^*,...,f_{T}^*))\phi(A_tf_t(H_t))}{p(A_t|H_t)}\bigg],
\end{equation}
where
$$\mathcal{A}_{t}=\bigg\{f\in\mathcal{F}_t\bigg|E\bigg[\frac{R_t\psi(A_tf(H_t),\eta_t)}{p(A_t|H_t)}\bigg]\leq \tau_t\bigg\}\quad t=1,...,T.$$
Equivalently, we replace the 0-1 loss function in the objective function with the hinge loss and replace the indicator function in the risk constraint with the shifted ramp loss function. {The hinge loss function is a typical choice of the surrogate for 0-1 loss in classification problems such as SVM. The shifted ramp loss can be viewed as a smooth and conservative approximation of the indicator function in the risk constraint function when $\eta_t$ is small, which will converge to the true constraint as $\eta_{t}$ goes to 0.} {As a note,  applying BR-DTRs does not require the reward and risk variables to be on the same scale since the solution is the same after rescaling $Y_t$ and $R_t$ (so $\tau_t$ is rescaled too).}

{For a constrained optimization problem under the 0-1 loss, we say that a surrogate problem is Fisher consistent if the solution to the surrogate problem also solves the original problem under the 0-1 loss. This definition is consistent with the traditional Fisher consistency definition of the unconstrained problem.} Our next result shows that the new surrogate problem leads to the DTRs that are Fisher consistent. Before stating the theorem, we define a $t$-stage pseudo-outcome $Q_t$ as
$${Q}_t=Y_t+U_{t+1}(H_{t+1};g_{t+1}^*,...,g_{T}^*),$$
which is the cumulative reward from stage $t$ to $T$ assuming that all treatments have been optimized from stage $t+1$ to $T$. {Given $(Q_t,R_t,A_t,H_t)$ and $a=\pm1$, we introduce following notations:}
{
\begin{equation*}
  \begin{split}
    m_{Q_t}(h,a)=E[Q_t|H_t=h,A_t=a],&\quad\quad \delta_{Q_t}(h)=m_{Q_t}(h,1)-m_{Q_t}(h,-1), \\
    m_{R_t}(h,a)=E[R_t|H_t=h,A_t=a],&\quad\quad \delta_{R_t}(h)=m_{R_t}(h,1)-m_{R_t}(h,-1).
  \end{split}
\end{equation*}}
Let 
$$\tau_{t,\min}=E\bigg[R_t\frac{\mathbb{I}(A_t\delta_{R_t}(H_t)<0)}{p(A_t|H_t)}\bigg],$$
$$\tau_{t,\max}=E\bigg[R_t\frac{\mathbb{I}(A_t\delta_{Q_t}(H_t)>0)}{p(A_t|H_t)}\bigg].$$
In other words, $\tau_{t,\min}$ is the risk under the decision function given by $-\delta_{R_t}(H_t)$, which is the one minimizing the risk regardless of the reward outcome. Thus, $\tau_{t,\min}$ is the minimum risk that one can possibly achieve at stage $t$. While $\tau_{t,\max}$ is the risk for the decision function given by $\delta_{Q_t}(H_t)$, which is the one maximizing the reward regardless of the risk. Thus, $\tau_{t,\max}$ is the maximal risk.

\begin{theorem}
  \label[theorem]{fisher_consistency}
  For $t=1,..,T$ and any fixed $\tau_{t,\min}<\tau_t<\tau_{t,\max}$, suppose that $P(\delta_{Q_t}(H_t)\delta_{R_t}(H_t)=0)=0$ and random variable
  $\delta_{Q_t}(H_t)/\delta_{R_t}(H_t)$ has a distribution function with a continuous density function in the support of $H_t$. Then for any $\eta_t\in(0,1]$ and $t=1,...,T$, we have $\text{sign}(f_{t}^*)=\text{sign}(g^*_t)$ almost surely, and $(f_1^*, ..., f_T^*)$ solves the optimization problem in (\ref{ow_prob}).
\end{theorem}

\begin{remark}
  {\cref{ACUTE} is a key condition for obtaining \cref{fisher_consistency}, which ensures that the original multistage problem can be decomposed into a finite number of single-stage single-constraint subproblems each w.r.t. to the decision function of the current stage. Without this assumption, the solution from each stagewise problem may not necessarily control the risk and the induced risk can be either higher or lower than the risk constraint depending on the relationship between $R_t(\bar{a}_t)$ and $R_t(a_t)$.}
\end{remark}

\noindent
When $\tau_t\geq\tau_{t,\max}$, the BR-DTRs problem is reduced to a standard DTRs problem and \citet{zhao_new_2015} shows that the Fisher consistency holds without additional conditions. For $T=1$, the conditions are similar to \citet{wang_learning_2018}, but they assume $H_t$ to have a continuous distribution. \cref{fisher_consistency} basically indicates that when the risk constraints are feasible and assume that the reward difference between two treatments varies continuously with respect to the risk difference, using the surrogate loss leads to the true optimal DTRs for any shifting parameter $\eta_t\in (0, 1]$. {The proof of \cref{fisher_consistency} can be completed by first showing that the surrogate problem (\ref{f_surrogate}) yields Fisher consistent rule for $T=1$ and then proving that the backward induction algorithm (\ref{f_zero_one}) yields the optimal solution under \cref{ACUTE}. Our proof follows the same sketch where the consistency for $T=1$ is established in Section \ref{A:1} and the optimality of the backward induction is established in Section \ref{A:2}. We note that both results are nontrivial and have never been established in the existing literature.} The complete proof is presented in Section \ref{A} in Appendix A.
{
  \begin{remark}
    The key step to proving \cref{fisher_consistency} is to derive a closed-form solution to the surrogate problem for $T=1$. There are three main challenges. First, we consider the Lagrange function of the surrogate problem and obtain its closed-form solution. Second, we show that the optimal solution of the surrogate problem attains some decision boundary, and this is proved using contradiction and careful construction. Third, we show that there exists a Lagrange multiplier yielding the optimal solution to the surrogate problem. {The last step entails the continuous density assumption of $\delta_{Q_t}(H_t)/\delta_{R_t}(H_t)$, which can be implied by the continuity of $Q_t$ and $R_t$ functions.} 
  \end{remark}
}

%% 2-3
\subsection{Estimating BR-DTRs using empirical data}
\label{sec:2_3}
Given data $\{(H_{i1}, A_{i1}, Y_{i1},R_{i1},...,H_{iT},A_{iT},Y_{iT},R_{iT})\}_{i=1}^n$ from $n$ i.i.d. patients, we propose to solve the empirical version of the surrogate problem to estimate the optimal DTRs: let
\begin{equation*}
  \mathcal{A}_{t,n}=\bigg\{f\in\mathcal{G}_t\bigg|\frac{1}{n}\sum_{i=1}^n\frac{R_{it}\psi(A_{it}f(H_{it}),\eta_t)}{p(A_{it}|H_{it})}\leq \tau_t\bigg\},
\end{equation*}
then we solve
\begin{equation}
  \label{true_problem}
  \widehat{f}_{t}=\argmin_{f\in\mathcal{A}_{t,n}}\frac{1}{n}\sum_{i=1}^n\frac{(\sum_{s=t}^TY_{is})\prod_{s=t+1}^T\mathbb{I}(A_{is}\widehat{f}_{s}(H_{is})>0)}{\prod_{s=t}^Tp(A_{is}|H_{is})}\phi(A_{it}f(H_{it}))+\lambda_{n,t}\|f\|^2_{\mathcal{G}_t}
\end{equation}
for $t=T,...,1$ in turn. Here, $\|\cdot\|_{\mathcal{G}_t}$ denotes the functional norm associated with functional space $\mathcal{G}_t$. The last term $\lambda_{n,t}\|f\|^2_{\mathcal{G}_t}$ is a typical choice of penalty term which regularizes the complexity of the estimated optimal decision function to avoid overfitting. Common choices of $\mathcal{G}_t$ include Reproducing Kernel Hilbert Space (RKHS) under a linear kernel where $k(h_i,h_j)=h_i^Th_j$, or a Gaussian radial basis kernel with $k(h_i,h_j)=\exp(-\sigma^2\|h_i-h_j\|^2)$, where $\sigma$ denotes the inverse of the bandwidth. 

{A major disadvantage of implementing (\ref{true_problem}) directly is that subjects whose future stages' observed treatment do not follow the estimated optimal treatments will be assigned with zero weights, which eliminates their contributions to the estimation of early stages and leads to a considerable loss of sample size as the estimation continues. To overcome this limitation, in this work, we adopt the augmentation technique to further improve the efficiency and stability of the estimation procedure. The augmentation technique was first proposed by \citet{liu_augmented_2018} to improve the performance of OWL where the basic idea is to predict the expected reward for subjects' whose future observed treatments are not optimal. Specifically, we replace the weights in the objective function and treatment variable by}
\begin{equation}
  \label{new_y_a}
  \widehat{Y}_{it}=|Y_{it}+\widehat{Q}_{i,t+1}-\widehat{\mu}_t(H_{it})|,\quad\widehat{A}_{it}=A_{it}*\text{sign}(Y_{it}+\widehat{Q}_{i,t+1}-\widehat{\mu}_t(H_{it})).
\end{equation}
Here, $\widehat{Q}_{i,t+1}$ is the augmented $Q$-function defined as 
\begin{equation}
  \label{equation_algorithm_2}
  \begin{split}
  \widehat{Q}_{i,t+1}=&\frac{(\sum_{s=t+1}^TY_{is})\prod_{s=t+1}^T\mathbb{I}(A_{is}\widehat{f}_{s}(H_{is})>0)}{\prod_{s=t+1}^Tp(A_{is}|H_{is})}\\
  &-\sum_{j=t+1}^T\bigg\{\frac{\prod_{s=t+1}^{j-1}\mathbb{I}(A_{is}\widehat{f}_{s}(H_{is})>0)}{\prod_{s=t+1}^{j-1}p(A_{is}|H_{is})}\bigg[\frac{\mathbb{I}(A_{ij}\widehat{f}_{j}(H_{ij})>0)}{p(A_{ij}|H_{ij})}-1\bigg]\widehat{\mu}_{t+1,j}(H_{ij})\bigg\},
  \end{split}
\end{equation}
and let $\widehat{Q}_{i,T+1}=0$. {In expressions (\ref{new_y_a}) and (\ref{equation_algorithm_2}), $\{\mu_t\}$ and $\{\mu_{t,j}\}$ are fixed nuisance functions that need to be provided in advance. Intuitively, $\{\mu_{t,j}\}$ in the augmented Q-functions are contributions to the loss function for patients whose received treatments are not optimal, and $\{\mu_t\}$ in (\ref{new_y_a}) are introduced to remove the main effect which could further reduce the weight variability without affecting the treatment rule estimation. When constructing the final weight, we flip the sign for both the weight and observed treatment for patients who have negative weights to ensure that all weights are nonnegative, which will lead to the same objective function up to a constant and thus will not affect the estimation. Due to the doubly robust design in the construction of the augmentation terms as shown in \citet{liu_augmented_2018}, both $\{\mu_t\}$ and $\{\mu_{t,j}\}$ are allowed to be misspecified and the estimated DTRs will remain to be optimal asymptotically, but a more accurate prediction can potentially lead to more reliable estimation. In practice, $\{\mu_{t,j}\}$ and $\{\mu_t\}$ usually need to be estimated from observed data. For simplicity, we propose to use the simple least square estimator and minimizing $\sum_{i=1}^n(Y_{it}+\widehat{Q}_{i,t+1}-{\mu}_t(H_{it}))^2$ to estimate $\{\mu_t\}$, and estimate $\{\mu_{t,j}\}$ via solving the weighted least square}
\begin{equation}
  \label{equation_algorithm_1}
  \frac{1}{n}\sum_{i=1}^n\frac{\prod_{s=t}^T\mathbb{I}(A_{is}\widehat{f}_{s}(H_{is})>0)}{\prod_{s=t}^Tp(A_{is}|H_{is})}\frac{1-p(A_{ij}|H_{ij})}{\prod_{s=t}^jp(A_{is}|H_{is})}(\sum_{s=t}^TY_{is}-\mu_{t,j}(H_{ij}))^2
\end{equation}
{following \citet{liu_augmented_2018}. By constructing $\widehat Q_{i,t}$ and replacing the original weight by $\widehat{Y}_{it}$, the refined procedure can utilize the information from all subjects to estimate the optimal rules across all stages, which will lead to more efficient estimation for DTRs compared with (\ref{true_problem}).}

\begin{algorithm}[t]
  \footnotesize
  \caption{BR-DTRs via Backward Induction}
  \label{algorithm:1}
  \begin{algorithmic}
    \State \textbf{Input:} Given training data $(Y_{it},R_{it},A_{it},H_{it})$ and $(\lambda_{t},\mathcal{G}_{t},\tau_t,\eta_t)$ for $i=1,...,n$ and $t=1,...,T$
    \For{$t=T$ to $1$}
    \For{$j=t+1$ to $T$}
    \State obtain estimator $\widehat{\mu}_{t,j}$ via minimizing (\ref{equation_algorithm_1})
    \EndFor
    \If{$t=T$} define $\widehat{Q}_{i,T+1}=0$
    \Else~compute $\widehat{Q}_{i,t+1}$ from (\ref{equation_algorithm_2})
    \EndIf
    \State compute $\widehat{\mu}_{t}$ via least square estimator and obtain $\{(\widehat{Y}_{it},\widehat{A}_{it})\}_{i=1}^n$ via (\ref{new_y_a})
    \State obtain $\widehat{f}_{t}$ by solving 
    \begin{equation*}
      \begin{alignedat}{2}
        \underset{f\in\mathcal{G}_t}{\min}\quad&\frac{1}{n}\sum_{i=1}^n\frac{\widehat{Y}_{it}}{p(A_{it}|H_{it})}\phi(\widehat{A}_{it}f(H_{it}))+\lambda_{n,t}\|f\|^2_{\mathcal{G}_t}\\
        \text{subject~to}\quad&\frac{1}{n}\sum_{i=1}^n\frac{R_{it}}{p(A_{it}|H_{it})}\psi(A_{it}f(H_{it}),\eta_t)\leq \tau_t\\
      \end{alignedat}
    \end{equation*}
    \State using DC algorithm
    \EndFor
    \State\textbf{Output:} $(\widehat{f}_{1},...,\widehat{f}_{T})$
  \end{algorithmic}
\end{algorithm}

{Hence, we propose a backward procedure to estimate the optimal DTRs based on the refined problem. First, we solve a single-stage problem using data at stage $t=T$, and then in turn, for $t=T-1,...,1$, we solve the constrained optimization problem (\ref{true_problem}) after plugging in $(\widehat{f}_{t+1},...,\widehat{f}_T)$ into (\ref{new_y_a}) and (\ref{equation_algorithm_2}) and replacing the weight with $\widehat{Y}_{it}$. The pseudocode of our final proposed algorithm is presented in \cref{algorithm:1}. Finally, since the objective function and the risk constraint in (\ref{true_problem}) can be both written as the difference between two convex functions, for the optimization at each stage we can apply the difference of convex functions (DC) algorithm \citep{tao_convex_1997} to iteratively solve the subproblem. In each iteration, the subproblem can be further reduced to a standard quadratic programming problem. {Details of the derivation and the implementation of the DC algorithm are presented in Appendix B.}}

%%%%%%%%%%%%%%%%%% Section 3 %%%%%%%%%%%%%%%%%%%%%%%%%%%

\section{Theoretical Properties}
\label{sec:theory}
In this section, we establish the non-asymptotic error rate of the value function and stagewise risks under the estimated decision functions $(\widehat{f}_{1},...,\widehat{f}_{T})$. More specifically, for any arbitrary decision functions $(g_1,..., g_T)$, the value function of $(g_1,...,g_T)$ is defined as
$$\mathcal{V}(g_1,...,g_T)=E\bigg[\frac{(\sum_{t=1}^TY_t)\prod_{t=1}^T\mathbb{I}(A_tg_t(H_t)>0)}{\prod_{t=1}^Tp(A_t|H_t)}\bigg].$$
We aim at obtaining the non-asymptotic bound for the regret function given by 
$$
\mathcal{V}(g_{1}^*,...,g_{T}^*)-\mathcal{V}(\widehat{f}_{1},...,\widehat{f}_{T})$$ 
and the stagewise risk difference is given by
$$E\bigg[\frac{R_t\mathbb{I}(A_t\widehat{f}_{t}(H_t)>0)}{p(A_t|H_t)}\bigg]-\tau_t,$$
for $t=1,...,T$.

We assume that $\{\mathcal{G}_t\}_{t=1}^T$ are the RKHS generated by the Gaussian radial basis kernel, i.e. $\mathcal{G}_t:=\mathcal{G}(\sigma_{n,t})$, where $\mathcal{G}(\sigma)$ denotes the Gaussian RKHS associated with bandwidth $\sigma^{-1}$. 
Furthermore, for random variable $Q_t, R_t, A_t$ and $H_t$, we define for $a,b\in \{-1,1\}$,
$${H_{a,b,t,\tau}=\bigg\{h\in\mathcal{H}_t: a\delta_{Q_t}(h)>0,bf_{t,\tau}^*(h)>0\bigg\}}$$
 and
{$\Delta_{t,\tau}(h)=\sum_{a,b\in \{-1,1\}}
    \text{dist}(h,\mathcal{H}_t/H_{a,b,t,\tau})\mathbb{I}(h\in H_{a,b,t,\tau}),$}
where {$\mathcal{H}_t/H_{a,b,t,\tau}$ denotes the set difference between $\mathcal{H}_t$ and $H_{a,b,t,\tau}$, $\text{dist}(h,\cdot)$ denotes the Euclidean distance} from point $h$ to a set, and $f^*_{t,\tau}$ denotes optimal solution of (\ref{f_surrogate}) at stage $t$ but replace the risk constraint in $\mathcal{A}_t$ by $\tau$. Note that \cref{fisher_consistency} implies
$$Q_t=Y_t+U_{t+1}(H_{t+1};g_{t+1}^*,...,g_{T}^*)=Y_t+U_{t+1}(H_{t+1};f_{t+1}^*,...,f_{T}^*).$$
We assume
\begin{assumption}
  \label{GNE}
  Let $P_t$ denote the distribution of $H_t$. For given $(\tau_1,...,\tau_T)$ and any $t=1,..,T$, there exist universal positive constants $\delta_{0,t}>0$, $K_t>0$ and $\alpha_t>0$ such that for any $\tau'\in[\tau_t-2\delta_{0,t},\tau_t+2\delta_{0,t}]\subset(\tau_{t,\min},\tau_{t,\max})$ we have
  $$\int_{\mathcal{H}_t}\exp\bigg(-\frac{\Delta_{t,\tau'}(h)^2}{s}\bigg)P_t(dh)\leq K_t s^{\alpha_td_t/2}$$
  holds for any $s>0$.
\end{assumption}
\cref{GNE} is an extension of the Geometric Noise Exponent (GNE) assumption proposed by \citet{steinwart_fast_2007} to establish a fast convergence risk bound for standard SVM, and later adopted by \citet{zhao_estimating_2012} to derive the risk bound for the DTRs without risk constraints. The GNE assumption can be viewed as a regularization condition of the behavior of samples near the true optimal decision boundary. {We note that GNE assumption is implied by Tsybakov's noisy assumption \citep{audibert_fast_2007}, thus weaker than Tsybakov's noisy assumption \citep[see Theorem 2.6 of][]{steinwart_fast_2007}. } 

For a fixed $\tau_t$, $\alpha_t$ can be taken to 1 when $\Delta_{t,\tau}(h)$ has order less or equal to $O(h)$. When the optimal decision boundary is strictly separated, i.e. $\text{dist}(H_{a,b,t,\tau},H_{a',b',t,\tau})>0$ for any $a\neq a'$ and $b\neq b'$, by using the fact that $\exp(-t)\leq C_{s}t^{-s}$ one can check that \cref{GNE} holds for $\alpha_t=\infty$. When the optimal decision boundary is not strictly separated, it can be shown that \cref{GNE} can still hold for arbitrary $\alpha_t\in(0,\infty)$ when the marginal distribution of $H_t$ has light density near the optimal decision boundary (see Example 2.4 in \citet{steinwart_fast_2007}).

The following theorem gives the non-asymptotic error bound for the regret and risk difference for the estimated DTRs, assuming that $\{\mu_t\}$ and $\{\mu_{t,j}\}$ in the augmentation are known. The theorem allows stage-wise shifting parameters to vary with sample size, denoted by $(\eta_{n,1},...,\eta_{n,T})$.

\begin{theorem}
  \label[theorem]{risk_bound}
  Suppose that \cref{SUTV,NUC,Positivity,ACUTE,GNE} and conditions in \cref{fisher_consistency} hold, $H_t$ is defined on a compact set $\mathcal{H}_t\subset\mathbb{R}^{d_t}$ for $t=1,...,T$, and {assume that $\{\mu_t\}$ and $\{\mu_{t,j}\}$ are known functions.} Let $\{\nu_t\}_{t=1}^T$ and $\{\theta_t\}_{t=1}^T$ be two series of positive constants such that $0<\nu_t<2$ and $\theta_t>0$ for all $t=1,...,T$. Then for any $n\geq1$, $\delta_t>0$, $\lambda_{n,t}>0$, $\sigma_{n,t}>0$ and $0<\eta_{n,t}\leq1$, such that $\lambda_{n,t}\rightarrow0, \sigma_{n,t}\rightarrow \infty$ and that there exist constants $C_1,C_2,C_3$ satisfying
  $$C_1\sigma_{n,t}^{-\alpha_td_t}\eta^{-1}_{n,t}\leq\delta_{0,t},\ \
  C_2n^{-1}\sigma_{n,t}^{(1-\nu_t/2)(1+\theta_t)d_t}\leq1,$$
  and
  $\delta_{t}+C_1\sigma_{n,t}^{-\alpha_t,d_t}\eta_{n,t}^{-1}+C_3n^{-1/2}\sigma_{n,t}^{(1-\nu_t/2)(1+\theta_t)d_t/2}\big(\frac{M}{c_1\lambda_{n,t}}+\sigma_{n,t}^{d_t}\big)^{\nu_t/4}\eta_{n,t}^{-\nu_t/2}\leq2\delta_{0,t}$, it holds
  \begin{equation*}
    \begin{split}
      |\mathcal{V}(\widehat{f}_{1},...,\widehat{f}_{T})-\mathcal{V}(g_1^*,...,g_T^*)|\leq \sum_{t=1}^T(c_1/5)^{1-t}C_t\bigg(&n^{-1/2}\lambda_{n,t}^{-1/2}\sigma_{n,t}^{(1-\nu_t/2)(1+\theta_t)d_t/2}\\
      &+\lambda_{n,t}\sigma_{n,t}^{d_t}+\sigma_{n,t}^{-\alpha_td_t}\eta_{n,t}^{-1}+\eta_{n,t}+\delta_t\bigg)
    \end{split}
  \end{equation*}
  with probability of at least $1-\sum_{t=1}^Th_t(n,\sigma_{n,t})$, where
  $$h_t(n,\sigma_{n,t})=2\exp\bigg(-\frac{2n\delta_{0,t}^2c_1^2}{M^2}\bigg)+2\exp\bigg(-\frac{n\delta_{t}^2c_1^2}{2M^2}\bigg)+\exp\big(-\sigma_{n,t}^{(1-\nu_t/2)(1+\theta_t)d_t}\big).$$
  Moreover, with probability at least $1-h_t(n,\sigma_{n,t})$, the risk induced by $\widehat{f}_{t}$ satisfies
  \begin{equation*}
    \begin{split}
      E\bigg[\frac{R_t\mathbb{I}(A_t\widehat{f}_{t}(H_t)>0)}{p(A_t|H_t)}\bigg]\leq \tau_t+\delta_{t}+C_tn^{-1/2}\sigma_{n,t}^{(1-\nu_t/2)(1+\theta_t)d_t/2}\lambda_{n,t}^{-\nu_{t}/4}\eta_{n,t}^{-\nu/2}.
    \end{split}
  \end{equation*}
  Here, $C_t$ denotes some constant only depending on $\alpha_t, K_t$, $d_t$, $\nu_t$, $\theta_t$, $c_1$ and $M$.
\end{theorem}

\cref{risk_bound} can be established by first verifying the result for $T=1$ and then extending the result to $T\geq2$ using an analogous argument of Theorem 3.4 of \citet{zhao_new_2015}. The risk bound of the value function proved in \cref{risk_bound} indicates that the error consists of four parts. The first two terms correspond to the stochastic error and approximation error resulting from using the empirical estimator to approximate the true objective function and restricting the estimated decision functions within the Gaussian RKHS in the empirical problem. The third error term $O(\sigma_{n,t}^{-\alpha_td_t}\eta_{n,t}^{-1})$ is induced by using the empirical estimator as risk constraints in (\ref{true_problem}). The remaining error has order $O(\eta_{n,t})$ and results from the property that the regret under 0-1 loss function is upper bounded by the regret under hinge loss plus an error term of order $O(\eta)$ when we use the shifted ramp loss to approximate the indicator function in constraints. Due to the existence of the last two error terms, the choice of shifting parameter must be small but bounded away from 0 in order to minimize the regret. The proof of \cref{risk_bound} and required preliminary lemmas are provided in Section A.2 in Appendix A. 

According to \cref{risk_bound}, the risk bound of the regret is minimized by setting $\eta_{n,t}=\sigma_{n,t}^{-\alpha_td_t}\eta_{n,t}^{-1}$, $\lambda_{n,t}\sigma_{n,t}^{d_t}=\sigma_{n,t}^{-\alpha_td_t}\eta^{-1}_{n,t}$ and $\eta_{n,t}=n^{-1/2}\lambda_{n,t}^{-1/2}\sigma_{n,t}^{(1-\nu_t/2)(1+\theta_t)d_t/2}$, which gives 
$$\lambda_{n,t}=O\big(\sigma_{n,t}^{-(\alpha_t+2)d_t/2}\big),\quad \eta_t=O\big(\sigma_{n,t}^{-\alpha_td_t/2}\big)$$ 
and 
$$\sigma_{n,t}=O\big(n^{\frac{1}{\alpha_td_t+(\alpha_t+2)d_t/2+(1-\nu_{t}/2)(1+\theta_t)d_t}}\big).$$
Consequently, there exists constants $k_1,k_2>0$ independent of sample size $n$ such that
  \begin{equation*}
    |\mathcal{V}(\widehat{f}_{1},...,\widehat{f}_{T})-\mathcal{V}(g_1^*,...,g_T^*)|\leq k_1\sum_{t=1}^T(c_1/5)^{1-t}n^{-\frac{\alpha_td_t}{2\alpha_td_t+(\alpha_t+2)d_t+2(1-\nu_{t}/2)(1+\theta_{t})d_t}}
  \end{equation*}
holds with probability $1-\sum_{t=1}^T\exp\big(-k_2n^{\frac{(1-\nu_{t}/2)(1+\theta_{t})d_t}{\alpha_td_t+(\alpha_t+2)d_t/2+(1-\nu_{t}/2)(1+\theta_{t})d_t}}\big)$. When $\alpha_t$ can be selected arbitrarily large in which case the data are approximately separated near the optimal decision boundary, the convergence rate of the value function is at most of order $O(n^{-1/3})$. In terms of risks, when $\alpha_t$ can be arbitrarily large and let $\nu_t$ go to 0, the risk constraint inequality indicates that the stagewise risk under the estimated rule can always be bounded by $\tau_t$ plus an error term of order up to $O(n^{-1/2})$. In terms of stage $T$, we note that the error bound is increasing exponentially with respect to the total number of stages. This result is similar to the risk bound of value function obtained in Q-learning \citep{murphy_experimental_2005} and OWL \citep{zhao_new_2015}. {In practice, the optimal choice of tuning parameters $\{\lambda_{n,t}\}_{t=1}^T$, $\{\sigma_{n,t}\}_{t=1}^T$ and $\{\eta_{n,t}\}_{t=1}^T$ can be obtained via cross-validation. }

\begin{remark}
  {Note that the result in \cref{risk_bound} is obtained under the assumption that $\{\mu_t\}$ and $\{\mu_{t,j}\}$ are known and fixed functions. As discussed in \cref{sec:2_3}, in practice functions $\{\mu_t\}$ and $\{\mu_{t,j}\}$ usually need to be estimated from observed data. Since the value function is Lipschitz continuous in terms of the model parameters $\{\mu_t\}$ and $\{\mu_{t,j}\}$, when $\{\mu_t\}$ and $\{\mu_{t,j}\}$ are estimated by prespecified parametric models as adopted in our proposed algorithm, the estimation will only induce an additional variability of order $O(n^{-\frac{1}{2}})$, which will be dominated by the error bounds in \cref{risk_bound} and, hence, will not affect the conclusion.}
\end{remark}
\begin{remark}
  \label{remark7}
  {The result obtained in \cref{risk_bound} can also be generalized to observational study when the treatment assignment probabilities are unknown and need to be estimated from the observed data. Similar to the previous remark, when $p(A_t|H_t)$ are estimated by parametric models such as logistic regression, such estimation will only induce an additional variability of order $O(n^{-\frac{1}{2}})$, which will not affect the non-asymptotic error obtained in \cref{risk_bound}. When the treatment assignment probability is estimated at a slower rate, the additional variability can be accounted for through additional expansion of the objective function on these parameters.}
\end{remark}

%%%%%%%%%%%%%%%%%% Section 4 %%%%%%%%%%%%%%%%%%%%%%%%%%%

\section{Simulation Studies}
\label{sec:simulation}
We demonstrate the performance of BR-DTRs via simulation studies in this section. We consider two settings both of which simulate the situation when adopting preferable treatment in the early stage would immensely affect the performance of possible treatments in later stages. Specifically, in both settings, we first generate an 8-dimensional baseline prognostic variable matrix $X$ from independent uniform distribution $U[0,1]$. In the first setting, we consider a two-stage randomized trial where treatments $A_1$ and $A_2$ are randomly assigned with an equal probability of 0.5. The stage-specific rewards and risks are defined by
\begin{equation*}
  \begin{split}
    Y_1=1-X_1+A_1(-X_1-X_2+1)+\epsilon_{Y_1},\quad&R_1=2+X_1+A_1(-X_1/2+X_2+1)+\epsilon_{R_1},\\
    Y_2=1-X_1+A_2(Y_1-3X_1+A_1+1)+\epsilon_{Y_2},\quad&R_2=1+X_1+A_2(Y_2/2-X_1+A_2/2+1)+\epsilon_{R_2},
\end{split}
\end{equation*}
where $\epsilon_{Y_1}$, $\epsilon_{Y_2}$ are noises of reward outcomes generated from the independent standard normal distribution $N(0,1)$, and $\epsilon_{R_1}$, $\epsilon_{R_2}$ are noises of adverse risks generated from the independent uniform distribution $U[-0.5,0.5]$. In this setting, both $Y_1$, $Y_2$, $R_1$ and $R_2$ are the linear functions of $H_1=X$ and $H_2=(H_1,A_1,Y_1,R_1)$. In the second setting, $Y_2$ is a nonlinear function of $H_2$ and is generated according to
\begin{equation*}
  \begin{split}
    Y_1=1+A_1(-X_1-X_2/3+1.2)+\epsilon_{Y_1}, \quad &R_1=1.5+A_1(-X_1/3+1.5)+\epsilon_{R_1},\\
    Y_2=1+A_2(-X_1^2/2-X_2^2/2+3A_1/2+1.5)+\epsilon_{Y_2},\quad &R_2=1+A_2(2A_1+2)+\epsilon_{R_2},
  \end{split}
\end{equation*}
and $(A_1,A_2,\epsilon_{Y_1},\epsilon_{Y_2},\epsilon_{R_1},\epsilon_{R_2})$ are generated the same way as setting I. Note that for setting II, the optimal decision boundary in stage II is a circle w.r.t. $(X_1,X_2)$.

For each simulation setting, we implement our proposed method with training data sample size $n$ equal to 200 and 400. We let $\eta=\eta_{1}=\eta_{2}$ varying from 0.02 to 0.1 with an increment of 0.02. For the first simulation setting, we repeat the simulation for $\tau_1=\tau_2=1.4$ and $1.5$; for the second simulation setting, we repeat the simulation for $\tau_1=\tau_2=1.3$ and $1.4$. Both the linear kernel and the Gaussian kernel are employed to compare their performance. {We conduct the estimation following exactly the same description in \cref{sec:2_3}} and the tuning parameter $C_{n,t}=(2n\lambda_{n,t})^{-1}$ will be selected by a 2-fold cross-validation procedure that maximizes the Lagrange dual function from a pre-specified grid of $2^{-10}$ to $2^{10}$. To alleviate the computational burden, when using the Gaussian kernel we follow the idea of \citet{wu_robust_2010} and fix $\sigma_{n,t}^{-1}$ to be {$2*\text{median}\{\|H_{it}-H_{jt}\|:A_{it}\neq A_{jt}\}$} instead of picking $\sigma_{n,t}$ adaptively according to $n$ and other tuning parameters. In our simulations, all feature variables will be re-centered to mean 0 and rescaled into interval $[-1,1]$. When solving the optimization problem, we choose the initial values for parameters either uniformly in a bounded interval or using the estimated parameters from the unconstrained problem. We recommend the latter approach as the performance is overall better than picking the initial point randomly. All quadratic programming programs in the DC procedure will be solved by R function \textit{solve.QP()} from \textit{quadprog} package (\url{https://cran.r-project.org/web/packages/quadprog/index.html}). As a comparison, we also implement the AOWL method proposed by \citet{liu_augmented_2018} as implemented in package \textit{DTRlearn2} (\url{https://cran.r-project.org/web/packages/DTRlearn2/index.html}), which ignores the risk constraints. In addition, we also compare our method with the naive approach where in stage I, we simply use $Y_1+Y_2$ as the outcome for estimation without adjusting for any delayed treatment effects even though the risk constraints are considered. To assess the performance of each method, we calculate the stage optimal estimated reward and risk on an independent testing dataset of size $N=2\times 10^4$. We repeat the analysis with 600 replicates.

\begin{figure}[hbp]
    \centering
    \includegraphics[scale = 0.25]{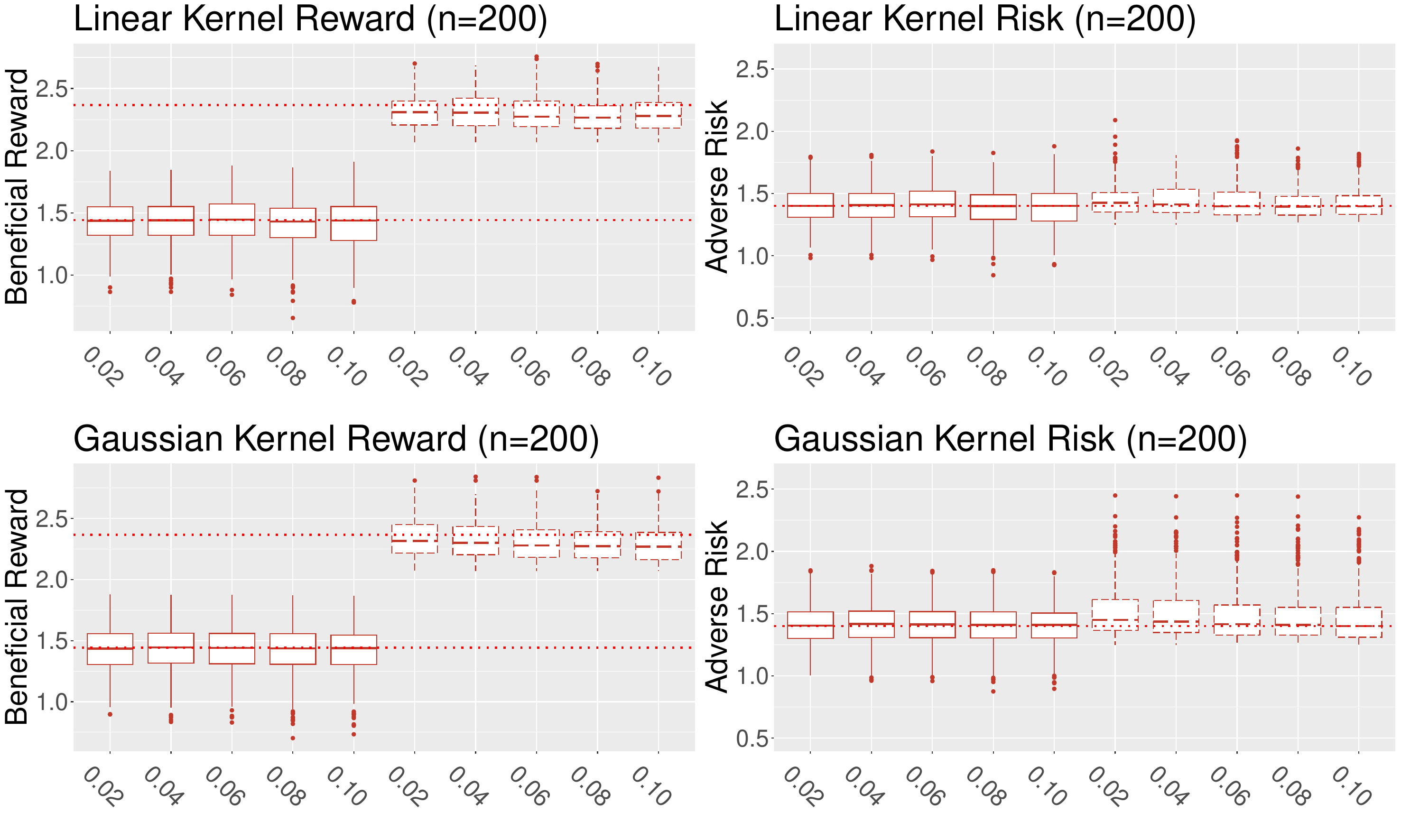}\\
    \includegraphics[scale = 0.25]{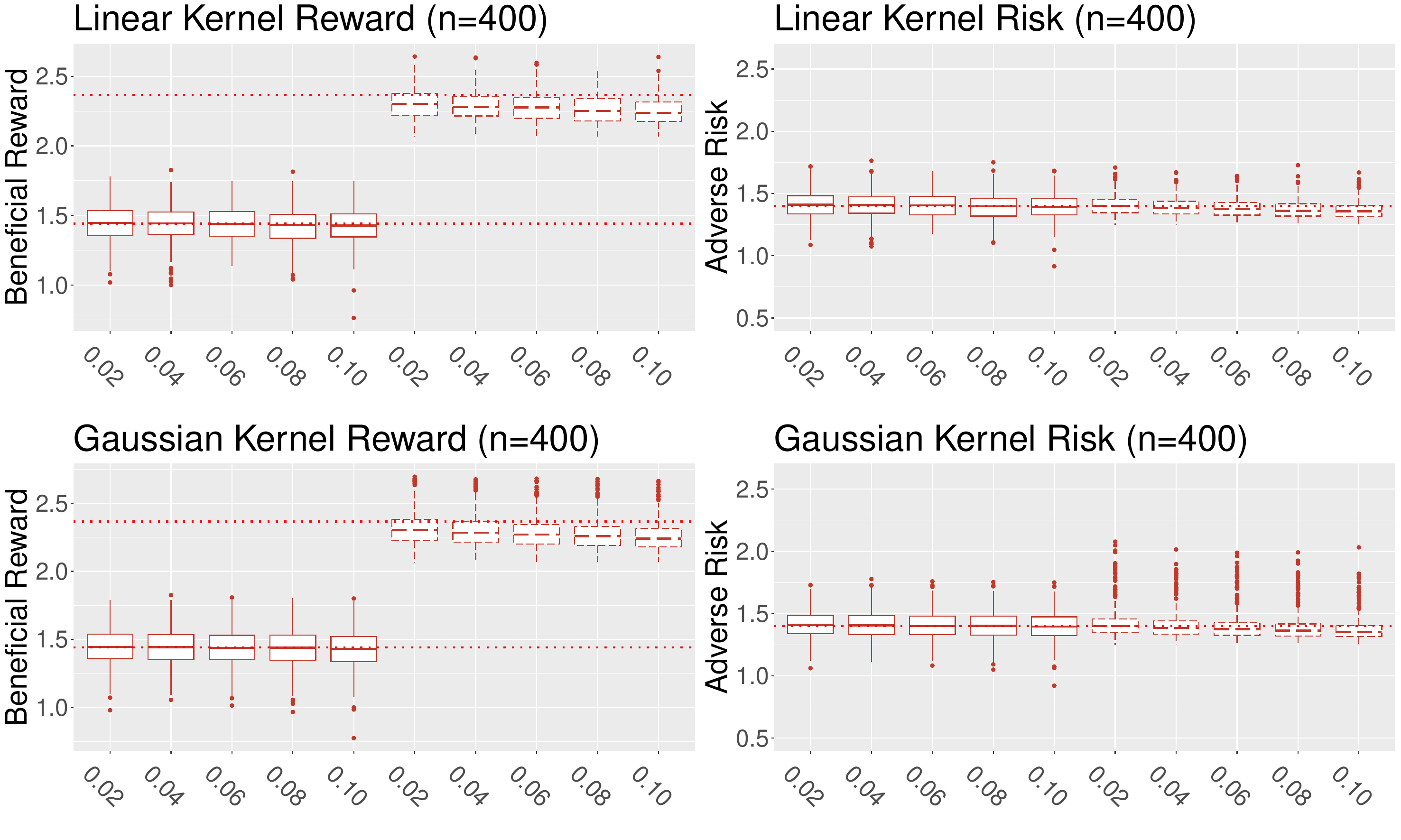}\\
    \vspace{-10pt}
    \includegraphics[scale = 0.3]{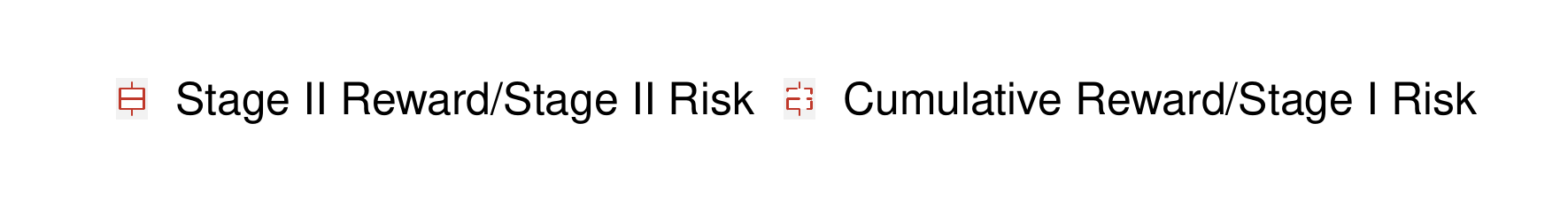}
    \vspace{-10pt}
    \caption{\footnotesize Estimated reward/risk on independent testing data set for simulation setting I, training sample size $n=\{200,400\}$ and $\eta=\{0.02,0.04,...,0.1\}$ (x-axis) under linear kernel or Gaussian kernel. The dashed line in reward plots refers to the theoretical optimal reward under given constraints. The dashed line in risk plots represents the risk constraint $\tau=1.4$.}
    \label{fig:1}
\end{figure}

\begin{figure}[hbp]
  \centering
  \includegraphics[scale = 0.25]{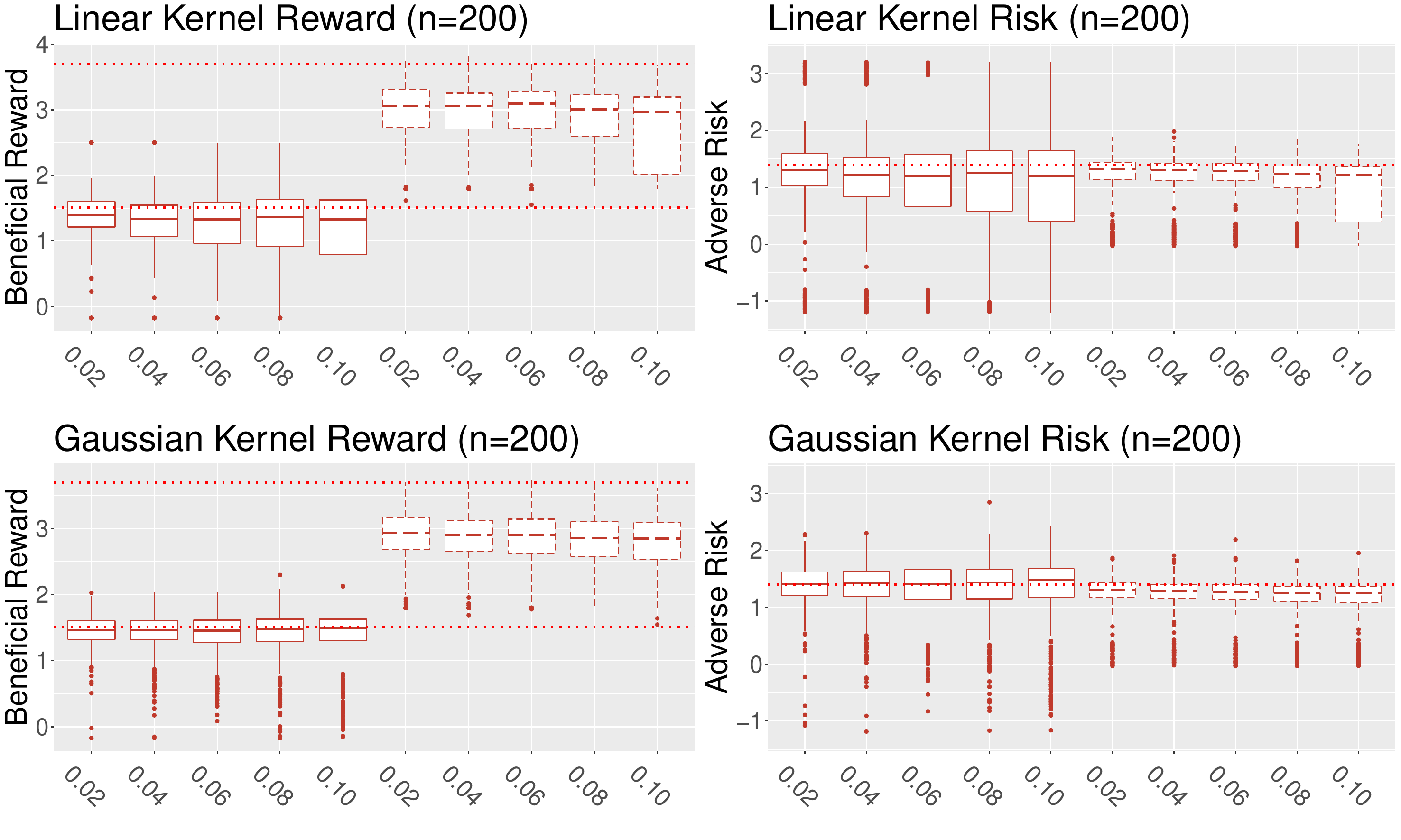}\\
  \includegraphics[scale = 0.25]{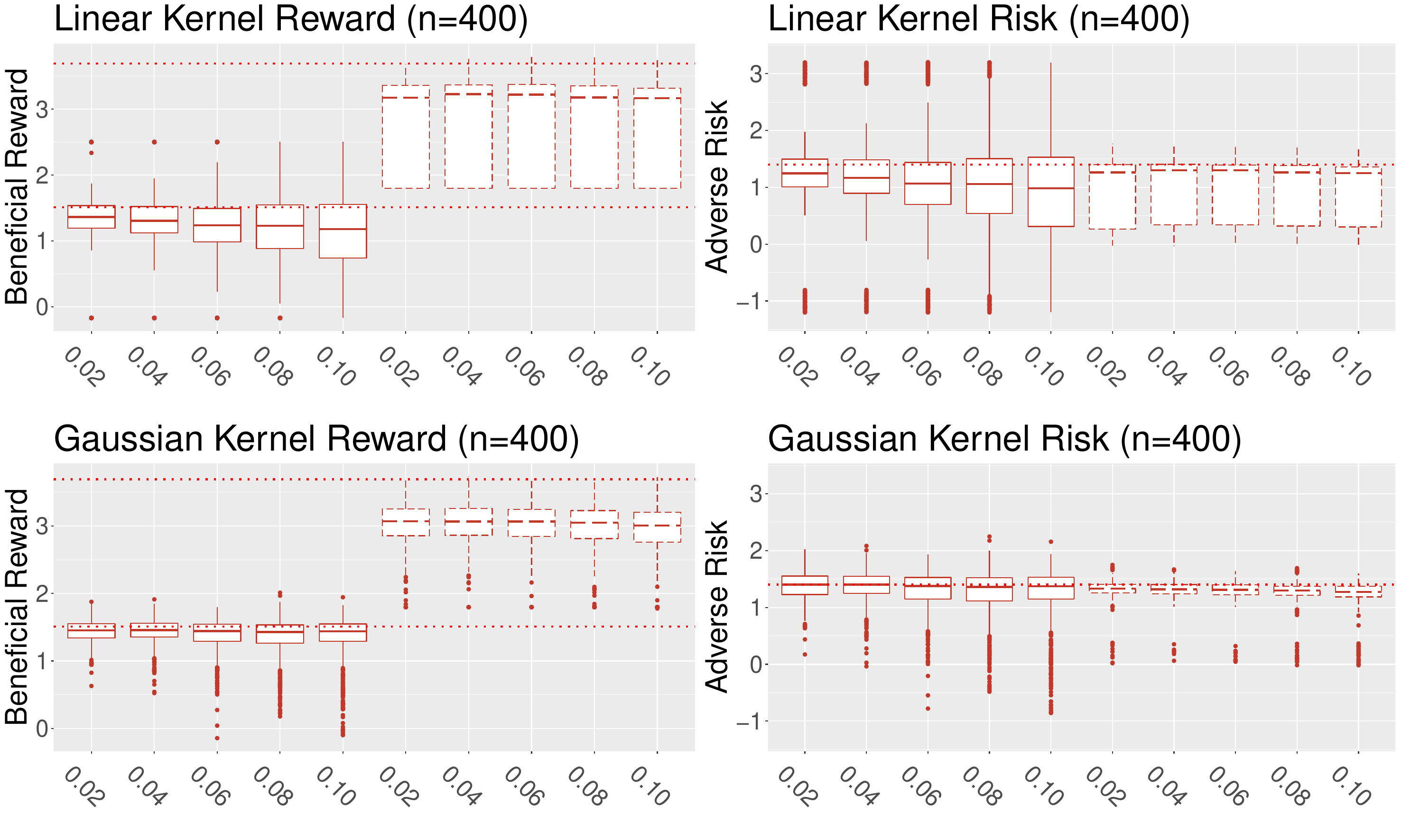}\\
  \vspace{-10pt}
  \includegraphics[scale = 0.3]{legend.pdf}
  \vspace{-10pt}
  \caption{\footnotesize Estimated reward/risk on independent testing data set for simulation setting II, training sample size $n=\{200,400\}$ and $\eta=\{0.02,0.04,...,0.1\}$ (x-axis) under linear kernel or Gaussian kernel. The dashed line in reward plots refers to the theoretical optimal reward under given constraints. The dashed line in risk plots represents the risk constraint $\tau=1.4$.}
  \label{fig:2}
\end{figure}

\cref{fig:1} displays the estimated reward and risk on the independent testing data for the first simulation setting under the different choices of training sample size, kernel basis, and shifting parameter $\eta$ for $\tau_1=\tau_2=1.4$. From the plot, we notice that for the simple linear setting, under both linear and Gaussian kernel the median values of estimated reward/risk will be close to the theoretical reward/pre-specified risk constraints. This indicates that the proposed method can successfully maximize the reward while controlling the risks across both stages. In this setting, compared with the linear kernel, using the Gaussian kernel will significantly underestimate the risk on training data, leading to somewhat exceeding risk on the testing data. Also as expected, in this setting increasing sample size would improve the performance under both kernel choices. In terms of the shifting parameter $\eta$, in setting I there is no obvious preference for choosing a small value to a large value. The result from the second nonlinear simulation setting under $\tau_1=\tau_2=1.4$ is presented in \cref{fig:2}. Under this more complicated setting and when both two stages' optimal decision boundaries are nonlinear, we notice that our method still yields a value close to the truth and the risks are reasonably controlled in both stages. The Gaussian kernel outperformed the linear kernel in both stages since using the linear kernel will misspecify the true model. When the sample size increased, the performance for the Gaussian kernel improved but it was not necessary for the linear kernel, likely due to the misspecification. We also observe that under the second simulation setting and when the Gaussian kernel is used, choosing a small shifting parameter $\eta$ will achieve better performance on the testing data with much smaller variability. {The results for $\tau=1.5$ for setting I and $\tau=1.3$ for setting II are similar to $\tau=1.4$ already discussed. An additional simulation study is conducted to investigate the performance of the proposed method in an observational study by assuming that the treatment assignment probability is unknown and estimated from data. Similar conclusions can be made in this setting. All additional results are presented in Appendix C.}

Finally, the results in \cref{table1} compare the performance of BR-DTRs to AOWL, which ignores the risk constraints, and the naive method, which considers the risk constraints but uses the immediate outcomes as the reward. Clearly, even though AOWL always gives a higher reward than BR-DTRs, the corresponding risks of applying the estimated treatment rules are much larger than the ones from BR-DTRs. In contrast, BR-DTRs can always give valid decision rules with risks close to pre-specified threshold values. When compared with the naive method, due to the nature of DTRs, the reward of the BR-DTRs method is always higher than the naive method. {In terms of the algorithm complexity, as a benchmark, the median running times for completing one estimation with fixed tuning parameters $C_{n,1}=C_{n,2}=1$ and fixed shifting parameter $\eta=0.02$ under setting I ($n=200$ and $\tau_1=\tau_2=1.4$) are 4.79 and 3.13 minutes for linear or Gaussian kernel, respectively. For setting II, the running times for $n=200$ and $\tau_1=\tau_2=1.4$ are 3.37 and 3.54 minutes, respectively. When the sample size is increased to $n=400$, the median running time will increase to 18.71 and 16.41 minutes for setting I, and 8.46 and 10.27 minutes for setting II. When the sample size is large, the running time can be reduced by using stochastic gradient-based methods to speed up the quadratic optimization in each DC iteration or implementing gradient-based methods to solve the constrained non-convex optimization directly.} 

\afterpage{
\begin{landscape}
  \renewcommand{\arraystretch}{1.25}
  \setlength{\tabcolsep}{4pt}
  \begin{table}[t]
    \centering
    \scriptsize
    \caption{\scriptsize Estimated reward/risk on independent testing data for $\tau_1=\tau_2=1.4$ and $n=400$ under 3 different methods using linear/Gaussian kernel. 
    \label{table1}}
    \begin{tabular}{lcccccccccc}
      \hline
             &      &        & \multicolumn{4}{c}{Linear Kernel}                               & \multicolumn{4}{c}{Gaussian Kernel}                             \\\cline{4-11}
    Setting    & $\eta$  & Method & Reward - II & Risk - II & Cumulative Reward & Risk - I & Reward - II & Risk - II & Cumulative Reward & Risk - I \\
    \hline\\
    Setting I  & 0.02 & BR-DTRs & 1.449(0.086)\tablefootnote{\tiny The estimated results are reported in \textit{median(dev)} format. \textit{median} denotes the median of expected risk/reward estimated via normalized estimator of 600 repeated analyses. \textit{dev} denotes the median value of the absolute difference between estimated risk/reward and \textit{median}}     & 1.410(0.072)  & 2.306(0.077)   & 1.400(0.053)                    & 1.438(0.093)    & 1.402(0.076)  & 2.301(0.078)   & 1.399(0.054) \\
             & 0.02 & Naive  &       ---       &      ---      & 2.224(0.072)   & 1.377(0.062)  &      ---        &      ---      & 2.201(0.093)   & 1.363(0.077) \\
             & 0.04 & BR-DTRs & 1.441(0.083)    & 1.404(0.067)  & 2.279(0.071)   & 1.384(0.051) & 1.437(0.094)    & 1.402(0.077)  & 2.281(0.077)   & 1.383(0.053) \\
             & 0.04 & Naive  &      ---        &      ---      & 2.207(0.086)   & 1.359(0.064)  &      ---        &      ---      & 2.196(0.085)   & 1.355(0.071) \\
             & 0.06 & BR-DTRs & 1.442(0.089)    & 1.405(0.074)  & 2.276(0.071)   & 1.377(0.050) & 1.435(0.092)    & 1.400(0.075)  & 2.268(0.070)   & 1.376(0.050) \\
             & 0.06 & Naive  &      ---        &      ---      & 2.185(0.086)   & 1.348(0.063)  &      ---        &      ---      & 2.181(0.082)   & 1.347(0.064) \\
             & 0.08 & BR-DTRs & 1.431(0.086)    & 1.393(0.070)  & 2.249(0.078)   & 1.358(0.048) & 1.437(0.093)    & 1.401(0.076)  & 2.257(0.068)   & 1.363(0.045) \\
             & 0.08 & Naive  &      ---        &      ---      & 2.164(0.088)   & 1.322(0.066)  &      ---        &      ---      & 2.168(0.082)   & 1.335(0.062) \\
             & 0.1  & BR-DTRs & 1.428(0.081)    & 1.394(0.066)  & 2.237(0.066)   & 1.357(0.045) & 1.430(0.094)    & 1.396(0.077)  & 2.239(0.065)   & 1.350(0.044) \\
             & 0.1  & Naive  &      ---        &      ---      & 2.168(0.074)   & 1.321(0.051)  &      ---        &      ---      & 2.161(0.073)   & 1.327(0.057) \\
             &      & AOWL   & 1.983(0.010)    & 2.149(0.044)  & 3.257(0.018)   & 2.678(0.096)  & 1.914(0.030)    & 2.099(0.083)  & 3.212(0.036)   & 2.584(0.218) \\
             \hline\\
    Setting II & 0.02 & BR-DTRs & 1.362(0.173)    & 1.246(0.247)  & 3.174(0.283)   & 1.262(0.198) & 1.456(0.106)    & 1.403(0.157)  & 3.069(0.192)   & 1.329(0.076) \\
             & 0.02 & Naive  &      ---        &      ---      & 1.797(0.000)   & 0.166(0.000)    &      ---        &      ---      & 1.816(0.019)   & 0.184(0.018) \\
             & 0.04 & BR-DTRs & 1.306(0.202)    & 1.166(0.288)  & 3.228(0.188)   & 1.299(0.130)   & 1.459(0.102)    & 1.402(0.153)  & 3.066(0.196)   & 1.319(0.080) \\
             & 0.04 & Naive  &      ---        &      ---      & 1.797(0.000)   & 0.166(0.000)    &      ---        &      ---      & 1.797(0.000)   & 0.166(0.000) \\
             & 0.06 & BR-DTRs & 1.238(0.252)    & 1.067(0.371)  & 3.221(0.197)   & 1.297(0.126)   & 1.444(0.123)    & 1.377(0.189)  & 3.068(0.208)   & 1.311(0.086) \\
             & 0.06 & Naive  &      ---        &      ---      & 1.797(0.000)   & 0.166(0.000)    &      ---        &      ---      & 1.797(0.000)   & 0.166(0.000) \\
             & 0.08 & BR-DTRs & 1.228(0.329)    & 1.059(0.479)  & 3.178(0.229)   & 1.260(0.149)   & 1.430(0.129)    & 1.360(0.195)  & 3.049(0.204)   & 1.297(0.078) \\
             & 0.08 & Naive  &      ---        &      ---      & 1.797(0.000)   & 0.166(0.000)    &      ---        &      ---      & 1.797(0.000)   & 0.166(0.000) \\
             & 0.1  & BR-DTRs & 1.177(0.404)    & 0.980(0.576)  & 3.169(0.239)   & 1.247(0.152)   & 1.438(0.123)    & 1.371(0.179)  & 3.009(0.206)   & 1.271(0.094) \\
             & 0.1  & Naive  &      ---        &      ---      & 1.797(0.000)   & 0.166(0.000)    &      ---        &      ---      & 1.797(0.001)   & 0.167(0.001) \\
             &      & AOWL   & 2.440(0.064)    & 3.017(0.002)  & 5.188(0.000)   & 2.839(0.000)    & 2.424(0.080)    & 3.018(0.002)  & 5.188(0.000)   & 2.839(0.000)  \\
             \hline
    \end{tabular}
  \end{table}
\end{landscape}}

%%%%%%%%%%%%%%%%%% Section 5 %%%%%%%%%%%%%%%%%%%%%%%%%%%

\section{Real Data Application}
\label{sec:application}
We apply BR-DTRs to analyze the data from the DURABLE study \citep{fahrbach_durable_2008}. The DURABLE study is a two-phase trial designed to compare the safety and efficacy of insulin glargine versus insulin lispro mix in addition to oral antihyperglycemic agents in T2D patients. During the first phase trial, patients were randomly assigned to the daily insulin glargine group or twice daily insulin lispro mix 75/25 (LMx2) group for 24 weeks. By the end of 24 weeks, patients who failed to reach an HbA1c level lower than 7.0\% would enter the second phase intensification study and be randomly reassigned with either basal-bolus therapy (BBT) or LMx2 for insulin glargine group or basal-bolus therapy (BBT) or three times daily insulin lispro mix 50/50 (MMx3) therapy for LMx2 group. Any other patients who reached HbA1c 7.0\% or lower would enter the maintenance study and keep the initial therapy for another 2 years. {A flowchart of the study design of the DURABLE trial is provided in Appendix D for reference.}

In the DURABLE study, the major objective is lowering patients' endpoint blood glucose level measured in HbA1c level, and in this analysis, {we use the reduction of HbA1c level at 24 weeks since baseline as the reward outcome for the first stage and use the reduction of HbA1c level at 48 weeks since 24 weeks as the reward outcome for the second stage.} The risk outcome is set to be hypoglycemia frequency encountered by patients, which reflects the potential risk induced by adopting assigned treatment. {Patients who achieved an HbA1c level lower than 7\% and entered the maintenance study would not be re-randomized with new treatments during the second stage. To accommodate these patients in our proposed framework, we make the additional assumption that for patients in the maintenance study, their first-stage treatment is already optimal and should not be adjusted. This assumption is consistent with the general guidance of treating T2D patients suggested by ADA where the patient's treatment should be unchanged if the patient's HbA1c level can be maintained lower than 7\% \citep{american_diabetes_association_pharmacologic_2022}. Under this assumption, in the second stage, patients in the maintenance study are already receiving optimal treatment so it is not necessary to estimate their optimal decision rules. Consequently, the second stage analysis will only involve patients who entered the intensification study, and only in the first stage will all patients be included in the analysis. In the first stage estimation, for patients in the maintenance study, their future reward outcome (reduction of HbA1c) is assumed to be maintained.} That is, in Stage I, the reward outcome becomes
$$Y'=\begin{cases}
  Y,& \text{if subject is from the maintenance study}\\
  Y\frac{\mathbb{I}(A_{2}\widehat{f}_2(H_{2})>0)}{0.5},& \text{if subject is from the intensification study}.\\
\end{cases}$$
Finally, the second stage risk outcome is the total frequency of hyperglycemia events during the intensification study (from 24 weeks to 48 weeks) and the first stage risk outcome is defined to be the total hypoglycemia events from week 0-24 for patients who entered intensification study, and the total hypoglycemia events from week 0-48 rescaled to 24 weeks for the remaining patients who entered maintenance study. In the analysis, we eventually apply the logarithm transformation to these counts to handle some extremely large counts in the data.

We consider 20 relevant covariates as the baseline predictors $H_1$, including HbA1c testing result, heart rate, systolic/diastolic blood pressures, body weight, body height, BMI, and 7 points self-monitored blood glucose measured at baseline (week 0) along with patient's age, gender, duration of T2D and 3 indicator variables indicating whether patients were taking metformin, thiazolidinedione, or sulfonylureas. The second stage predictors $H_2$ include all predictors in $H_1$, patient's treatment assignment, the cumulative number of hyperglycemia events during the first stage, along with heart rate, systolic/diastolic blood pressures, HbA1c and same 7 points self-monitored blood glucose measured at the initial time of the second stage (24 weeks). All covariates are centered at mean 0 and rescaled to be within $[-1,1]$.

The final study cohort includes 579 patients from the intensification study and another 781 from the maintenance study. To compare the performance, we randomly sample 50\% patients from the intensification study as the training sample for stage II and an additional 50\% patients from the maintenance study as the training sample for stage I. The remaining patients will be treated as the testing data to assess the performance of the estimated rules. We consider different risk constraints $\tau_2=(0.334, \infty)$ and $\tau_1=(0.893,0.948,1.005)$ where we rescale the risk to hypoglycemia events per 4 weeks. We note that $0.334$ and $0.948$ are the mean risks of stage II and stage I, respectively, and $1.005$ is close to the median estimated risk on testing data under the unconstrained case. We repeat the analysis 100 times for random splitting of the training and testing data. For our method, we also conduct the estimation following the description in \cref{sec:2_3} and use the Gaussian kernel and choose $\eta=0.02$, while tuning parameter $\{C_{n,t}\}_{t=1}^2$ for each stage will be selected by two-fold cross-validation similar to the simulation studies. The bandwidth of the Gaussian kernel is also selected similar to the simulation studies.

\begin{table}[t]
  \renewcommand{\arraystretch}{1.25}
  \footnotesize
  \caption{Estimated reward/risk under different risk constraints for DURABLE study analysis \label{table2}. Results are reported in the same format as \cref{table1}.}
  \begin{tabular}{cccccccc}
    \hline
    \multicolumn{2}{c}{Risk Constraint} & \multicolumn{3}{c}{BR-DTRs} & \multicolumn{3}{c}{Naive}\\
    $\tau_2$ & $\tau_1$ & Reward       & Stage II Risk & Stage I Risk & Reward       & Stage II Risk & Stage I Risk \\
    \hline
    0.334 & 0.893 & 1.471(0.072) & 0.311(0.033) & 0.844(0.044) & 1.460(0.087) & 0.311(0.033) & 0.842(0.049) \\
          & 0.948 & 1.520(0.078) & 0.311(0.033) & 0.874(0.067) & 1.499(0.091) & 0.311(0.033) & 0.868(0.066) \\
          & 1.005 & 1.547(0.089) & 0.311(0.033) & 0.929(0.102) & 1.527(0.098) & 0.311(0.033) & 0.923(0.111) \\
 $\infty$ & 0.893 & 1.598(0.043) & 0.347(0.028) & 0.832(0.039) & 1.604(0.048) & 0.347(0.028) & 0.840(0.040) \\
          & 0.948 & 1.605(0.053) & 0.347(0.028) & 0.832(0.040) & 1.607(0.056) & 0.347(0.028) & 0.850(0.056) \\
          & 1.005 & 1.620(0.068) & 0.347(0.028) & 0.922(0.107) & 1.625(0.062) & 0.347(0.028) & 0.888(0.103) \\
          & $\infty$ & 1.713(0.052) & 0.347(0.025) & 1.040(0.047) &        -     &        -     &       -     \\
    \hline   
  \end{tabular}
\end{table}

All real data analysis results are displayed in \cref{table2}. From \cref{table2} we first notice that in each stage, the median estimated risk on testing data is tightly controlled by the prespecified risk constraints. This demonstrates that BR-DTRs can also successfully control adverse risks in real applications. Under each risk constraint, the cumulative reward estimated by BR-DTRs is only slightly better or closed against the estimated reward using the naive method. One reason is that the majority of the patients in stage I would not enter the intensification study and, hence, have no delayed treatment effect at all. 

Among all 7 constraint settings, the uncontrolled setting, as expected, produces the estimated rules with both the highest reward and risks, and the estimated reward decreases as the risk constraint of either stage decreases. Under the unconstrained estimated optimal rules, all patients are recommended to receive LMx2 in the first stage and later switch to MMx3 after 24 weeks if patients' HbA1c level is greater than 7.0\% by the end of the first phase. As a comparison, when the risk constraint is imposed in stage II, the optimal rules will instead recommend all patients to receive BBT when patients fail to reach HbA1c lower than 7.0\% in the second stage at a price of significantly lower reduction in HbA1c by the end of 48 months. Similar treatment preference change happens in stage I as the optimal estimated rule becomes less favorable to LMx2 against insulin glargine when $\tau_1$ decreases. 

Comparing the reward and risks under different choices of risk constraint, $\tau_1=1.005$ and $\tau_2=\infty$ produce the second highest reward with moderate risk in the second stage and 10\% lower risk in the first stage compared to the unconstrained setting. Under this suboptimal setting, the estimated rules recommend only 50.7\% of patients start with LMx2 therapy and later switch to MMx3 therapy if patients fail to reach an HbA1c level of less than 7.0\% by the end of the first phase of treatment. By checking the baseline covariates between the patients who received different treatment recommendations, under this estimated rule for the patients whose baseline HbA1c falls in the range $[7,8)$, $[8,9)$ and $[9,10)$, the proportion of the patients who are recommended with LMx2 therapy drops from 62.7\% to 56.3\% and 46.3\%; similarly, for the patients whose baseline BMI falls in the range $[28,32)$, $[32,34)$ to $[34, 36)$, the proportion of patients recommended with LMx2 also drops from 59.3\% to 53.8\% and 51.3\%. The negative correlation between the increment of baseline HbA1c/BMI against the proportion of patients recommended with LMx2 as the first phase treatment indicates that the patients with a worse initial health condition are less likely to be recommended with LMx2 therapy as the initial treatment when the risk impact is considered. This is consistent with the fact that LMx2 is a more intense therapy compared with insulin glargine therapy and would cause more hypoglycemia events among unhealthier T2D patients. In particular, the suboptimal rules obtained from BR-DTRs meet the ADA guidance which suggests that intensive insulin therapy should be prescribed to patients according to patients' health condition to reduce potential hypoglycemia events. In conclusion, the real data application demonstrates that, by evaluating the impact of adverse risks along with beneficial reward, BR-DTRs can produce better personalized, more practically implementable treatment recommendations compared with standard OWL which only takes beneficial reward into consideration.

%%%%%%%%%%%%%%%%%% Section 6 %%%%%%%%%%%%%%%%%%%%%%%%%%%

\section{Discussion}
\label{sec:discussion}
{In this work, we introduce a new statistical framework BR-DTRs to estimate the optimal dynamic treatment rules under the stagewise risk constraints. Sufficient conditions are provided to guarantee the Fisher consistency of using backward induction to learn the optimal decision rules of DTRs problems under stagewise risk constraints. The backward induction technique provides an algorithm to solve BR-DTRs efficiently through iteratively solving a series of single-stage, single-constraint sub-problems. In addition, we establish the non-asymptotic risk bound for the value and stagewise risks under the estimated decision functions. Our theoretical contributions include providing sufficient conditions for implementing backward induction for the constrained decision-making problem and non-asymptotic performance guarantee under the estimated rules.}

{To tackle the numerical challenge due to the 0-1 loss, we introduced the hinge loss and shifted ramp loss as the surrogate losses in this work. We note that although the shifted ramp loss could also be used as the surrogate function for the objective function, it does not reduce to a standard SVM problem when $\tau_t$ is infinity and involves an additional tuning parameter.
More numerical comparisons with the alternative choices of the surrogate functions are necessary.}

{It is worth noting that even though we focus on handling DTRs problems, the proposed method is applicable and can be generalized to other sequential decision-making problems beyond biomedical research. One example is the promotion recommendation in E-commerce,  where the goal is to learn a personalized strategy that maximizes customers' buying willingness at a tolerable loss of revenue \citep{goldenberg_optimization_2021,wang_multi-stage_2023}. In this application, multiple waves of promotions are scheduled to be delivered to customers in a cycle \citep{chen_bcrlsp_2022} and BR-DTRs can be applied to learn the optimal strategies at each stage. In Appendix C.3, an additional simulation study mimicking such promotion recommendation problem has been conducted for $T=4$ and the results indicate that the BR-DTRs method still performs well.} Moreover, even though we assumed treatments to be dichotomous and only one risk constraint is imposed at each stage in BR-DTRs, our method can also be extended to problems with more treatment options and risk constraints at each stage. One can achieve this by imposing multiple smooth risk constraints to multicategory learning algorithms, such as angle-based learning methods \citet{qi_multi-armed_2020,ma_learning_2023}. However, verifying the Fisher consistency of generalized problems is not trivial and is beyond the scope of this work. {In addition, for many real world applications, finding the most influential feature variables that drive the optimal decisions is of equal importance as obtaining the explicit rules that maximize the beneficial reward under the constraints. Thus, BR-DTRs can also be extended to incorporate feature selection during the estimation. For example, when the RKHS is generated by the linear kernel, the optimal decision boundary is linear, and one can introduce an additional penalty term with a group structure to impose sparsity over feature variables. } 

{There are several limitations of the proposed method. One limitation is that the proposed method may not perform well for a very large number of horizons. For example, the uncertainty for the objective maximization is accumulated over stages in the backward algorithm, so it will increase for large  $T$. In contrast, as shown in \cref{risk_bound}, the uncertainty for the risk control at each stage will remain independent of $T$. Consequently, the risk constraint will mainly drive the decision rules for large $T$, which may not be the ideal solution in practice. Possible extensions can be to impose appropriate parametric assumptions on the DTRs, or less strict control on the risk function.} {Another limitation is the acute risk assumption, which requires the stagewise risk to be solely determined by the most recent action. However, this assumption may be violated in some applications when risks are expected to be affected by earlier actions. {For example, the stagewise risks can be defined as the total number of the most toxic treatments received since the beginning of the treatments.} Therefore, further extensions are necessary when the delayed risks exist.}

%%%%%%%%%%%%%%%%%% Acknowledgements %%%%%%%%%%%%%%%%%%%%%%%%%%%

% \acks{All acknowledgements go at the end of the paper before appendices and references.
% Moreover, you are required to declare funding (financial activities supporting the
% submitted work) and competing interests (related financial activities outside the submitted work).
% \acks{We gratefully thank the action editor and 4 reviewers for their careful review and constructive comments. 
% This work is partially supported by the NIH grants R01 GM124104, R01 MH12348 and R01 NS073671. }
% More information about this disclosure can be found on the JMLR website.}


% \newpage

\appendix

% Define new section title format for appendix
\renewcommand{\thesection}{A.\arabic{section}}
\renewcommand{\thesubsection}{A.\arabic{section}.\arabic{subsection}} 
\renewcommand{\thesubsubsection}{A.\arabic{section}.\arabic{subsection}.\arabic{subsubsection}} 
\renewcommand{\thefigure}{S.\arabic{figure}}
\renewcommand{\thetable}{S.\arabic{table}}

%%%%%%%%%%%%%%%%%% Appendix A %%%%%%%%%%%%%%%%%%%%%%%%%%%

\section*{Appendix A: Proof of \cref{fisher_consistency} and \cref{risk_bound}}
\label{appendix:A}

We summarize the additional notations used in the proofs below:
\begin{longtable}{p{0.16\textwidth}p{0.84\textwidth}}
$f_{t}^*$ & the true optimal decision function solving the BR-DTRs (\ref{f_surrogate}),\\
& and we use $f_{t,\tau}^*$ for $f_t^*$ whenever $\tau$ is necessary in the context;\\
$\mathcal{V}_{t,\phi}(s,h)$ & $-\left\{\phi(s)E[Q_t|H_t=h,A_t=1]+\phi(-s)E[Q_t|H_t=h,A_t=-1]\right\}$;\\
$\mathcal{R}_{t, \psi}(s,\eta,h)$ & $ \psi(s,\eta)E[R_t|H_t=h,A_t=1]+\psi(-s,\eta)E[R_t|H_t=h, A_t=-1]$.
\end{longtable}
\noindent When $T=1$, we omit subscript $t$ from all these notations.

%%%%%%%%%%%%%%%%%% Section A.1 %%%%%%%%%%%%%%%%%%%%%%%%%%%

\section{Proof of \cref{fisher_consistency}}
\label{A}
\subsection{{Verification of equation (\ref{equation_after_assumption})}}
\label{A:0}
In this section, we provide detailed intermediate steps for establishing (\ref{equation_after_assumption}). In other word, we show that equations $(i)-(iv)$ hold
\begin{equation*}
  \begin{split}
    E^{\mathcal{D}}[R_t]\overset{(i)}{=}E\bigg[&\frac{R_t\prod_{t=1}^T\mathbb{I}(A_tf_t(H_t)>0)}{\prod_{t=1}^Tp(A_t|H_t)}\bigg]\overset{(ii)}{=}E\bigg[R_t(\textrm{sign}(f_1), ....,\textrm{sign}(f_t))\bigg]\\
    &\overset{(iii)}{=}E\bigg[R_t(\textrm{sign}(f_t))\bigg]\overset{(iv)}{=}E\bigg[\frac{R_t \mathbb{I}(A_tf_t(H_t)>0)}{p(A_t|H_t)}\bigg].
  \end{split}
\end{equation*}
The proof uses the fact from \citet{qian_performance_2011} and \citet{zhao_new_2015}, which shows that under \cref{SUTV,NUC,Positivity}, we have 
\begin{equation}
  \label{derivatives_1}
  \frac{d P_\mathcal{D}}{d P}=\frac{\prod_{s=1}^T\mathbb{I}(A_sf_s(H_s)>0)}{\prod_{s=1}^Tp(A_s|H_s)},
\end{equation}
and 
\begin{equation}
  \label{derivatives_2}
  \frac{d P_{(\mathcal{D}_1,...,\mathcal{D}_t)}}{d P}=\frac{\prod_{s=1}^t\mathbb{I}(A_sf_s(H_s)>0)}{\prod_{s=1}^tp(A_s|H_s)},\quad \frac{d P_{\mathcal{D}_t}}{d P}=\frac{\mathbb{I}(A_tf_t(H_t)>0)}{p(A_t|H_t)}.
\end{equation}
Equality (i) immediately follows by (\ref{derivatives_1}). To show (ii), we first note that by conditioning on $(H_T,A_T)$, we have that for any $t<T$
\begin{align*}
  &E\bigg[\frac{R_t\prod_{t=1}^T\mathbb{I}(A_tf_t(H_t)>0)}{\prod_{t=1}^Tp(A_t|H_t)}\bigg]\\
  =&E\bigg[E\bigg[\frac{R_t\prod_{t=1}^T\mathbb{I}(A_tf_t(H_t)>0)}{\prod_{t=1}^Tp(A_t|H_t)}\bigg|H_T,A_T\bigg]\bigg]\\
  =&E\bigg[\frac{R_t\prod_{t=1}^{T-1}\mathbb{I}(A_tf_t(H_t)>0)}{\prod_{t=1}^{T-1}p(A_t|H_t)}E\bigg[\frac{\mathbb{I}(A_Tf_T(H_T)>0)}{p(A_T|H_T)}\bigg|H_T,A_T\bigg]\bigg]\\
  =&E\bigg[\frac{R_t\prod_{t=1}^{T-1}\mathbb{I}(A_tf_t(H_t)>0)}{\prod_{t=1}^{T-1}p(A_t|H_t)}\bigg].
\end{align*}
Repeating the same argument from stage $T-1$ to $t$, we have 
$$E^{\mathcal{D}}[R_t]=E\bigg[\frac{R_t\prod_{s=1}^{t}\mathbb{I}(A_sf_s(H_s)>0)}{\prod_{s=1}^{t}p(A_s|H_s)}\bigg].$$
Then, equations $(ii)$ and $(iv)$ can be obtained by applying (\ref{derivatives_2}) and recalling that by definition $E^{(\mathcal{D}_1,...,\mathcal{D}_t)}[R_t]=E[R_t(\text{sign}(f_1),...,\text{sign}(f_t))]$ and $E^{\mathcal{D}_t}[R_t]=E[R_t(\text{sign}(f_t))]$. Lastly, equation (iii) is followed by \cref{ACUTE} which implies that $R_t(\bar{a}_t)=R_t(a_t)$ for any $\bar{a}_t\in\{-1,+1\}^t$ and, consequently, $R_t(\text{sign}(f_1),...,\text{sign}(f_t))=R_t(\text{sign}(f_t))$.

\subsection{Proof of \cref{fisher_consistency} for $T=1$}
\label{A:1}

We consider $T=1$. After dropping the stage subscript, both (\ref{ow_prob}) and (\ref{f_zero_one}) are equivalent to solving
\begin{equation}
  \label{single_stage_original}
  \min_{f\in\mathcal{F}}~E\bigg[ \frac{Y\mathbb{I}(Af(H)<0)}{p(A|H)}\bigg],\quad\text{subject to }E\bigg[\frac{R \mathbb{I}(Af(H)>0)}{p(A|H)}\bigg]\le\tau,
\end{equation}
and its resulting decision is given by $\textrm{sign}(g^*)$. Without loss of generality, we assume that $Y$ is nonnegative; otherwise, we can change $Y$ to $|Y|$ and $A$ to $A*\text{sign}(Y)$, which will not change the optimal solution since the objective functions are equivalent up to a constant due to
\begin{equation*}
  \begin{split}
    E\bigg[\frac{Y\mathbb{I}(Af(H)<0)}{p(A|H)}\bigg]=&E\bigg[\frac{Y^+\mathbb{I}(Af(H)<0)}{p(A|H)}\bigg]-E\bigg[\frac{Y^-\mathbb{I}(Af(H)<0)}{p(A|H)}\bigg]\\
    =&E\bigg[\frac{Y^+\mathbb{I}(Af(H)<0)}{p(A|H)}\bigg]+E\bigg[\frac{Y^-\mathbb{I}(Af(H)>0)}{p(A|H)}\bigg]-E\bigg[\frac{Y^-}{p(A|H)}\bigg]\\
    =&E\bigg[\frac{|Y|\mathbb{I}(A*\text{sign}(Y)f(H)<0)}{p(A|H)}\bigg]-E\bigg[\frac{Y^-}{p(A|H)}\bigg],
  \end{split}
\end{equation*}
and note that $E[Y^-/p(A|H)]$ is a term that is independent of $f$. In addition, {following the notation in \cref{sec:2_2}, given random variables $(Y,R,A,H)$ and for $a=\pm1$, we define
\begin{equation*}
  \begin{split}
    m_Y(h,a)=E[Y|H=h,A=a],&\quad\quad \delta_Y(h)=m_Y(h,1)-m_Y(h,-1), \\
    m_R(h,a)=E[R|H=h,A=a],&\quad\quad \delta_R(h)=m_R(h,1)-m_R(h,-1),
  \end{split}
\end{equation*} 
and let
$$\tau_{\min}=E\bigg[R\frac{\mathbb{I}(A\delta_{R}(H)<0)}{p(A|H)}\bigg],$$
$$\tau_{\max}=E\bigg[R\frac{\mathbb{I}(A\delta_{Y}(H)>0)}{p(A|H)}\bigg].$$}

We define ${\cal M}=\left\{h: \delta_Y(h)\delta_R(h)<0\right\},$ i.e., the set of subjects where the beneficial treatment also reduces risk.
Then according to Theorem 1 in \citet{wang_learning_2018}, for any $\tau \in (\tau_{\min}, \tau_{\max})$, the optimal $g^*$ can be chosen as
  \begin{equation*}
    g^*(h)=\begin{cases}
      \textrm{sign}(\delta_Y(h)),&\text{if~}h\in\mathcal{M}\\
      1,&\text{if~}h\in\{\delta_Y(h)/\delta_R(h)>\lambda^*,\delta_Y(h)>0\}\cap\mathcal{M}^c\\
      -1,&\text{if~}h\in\{\delta_Y(h)/\delta_R(h)<\lambda^*,\delta_Y(h)>0\}\cap\mathcal{M}^c\\
      -1,&\text{if~}h\in\{\delta_Y(h)/\delta_R(h)>\lambda^*,\delta_Y(h)<0\}\cap\mathcal{M}^c\\
      1,&\text{if~}h\in\{\delta_Y(h)/\delta_R(h)<\lambda^*,\delta_Y(h)<0\}\cap\mathcal{M}^c,
    \end{cases}
  \end{equation*}
where $\lambda^*$ satisfies $E[R\mathbb{I}(Ag^*(H)>0)/p(A|H)]=\tau$. Our surrogate problem to be solved is (\ref{f_surrogate}), which is
\begin{equation}
  \label{single_stage_surrogate}
  \min_{f\in\mathcal{F}}~E\bigg[ \frac{Y\phi(Af(H))}{p(A|H)}\bigg],\quad\text{subject to }E\bigg[\frac{R\psi(Af(H),\eta)}{p(A|H)}\bigg]\le\tau.
\end{equation}
We let $f^*$ denote the solution. Our following theorem (the same version for \cref{fisher_consistency} for $T=1$) gives an explicit expression for $f^*$ so that the solution for the surrogate problem has 
the same sign as $g^*$.

\begin{theorem}
  \label[theorem]{thm:s1}
  For any fixed $\tau_{\min}<\tau<\tau_{\max}$, suppose that $P(\delta_Y(H)\delta_R(H)=0)=0$ and random variable $\delta_{Y}(H)/\delta_{R}(H)$ has distribution function with a continuous density function in the support of $H$. Then for any $\eta\in(0,1]$, $f^*(h)$ can be taken as
  \begin{equation}
    \label{f_expression}
    \begin{split}
      f^*(h)=\begin{cases}
       \textrm{sign}(\delta_Y(h)),&\text{if~}h\in\mathcal{M}\\
        1,&\text{if~}h\in\{\delta_Y(h)/\delta_R(h)>\lambda^*,\delta_Y(h)>0\}\cap\mathcal{M}^c\\
        -\eta,&\text{if~}h\in\{\delta_Y(h)/\delta_R(h)<\lambda^*,\delta_Y(h)>0\}\cap\mathcal{M}^c\\
        -1,&\text{if~}h\in\{\delta_Y(h)/\delta_R(h)>\lambda^*,\delta_Y(h)<0\}\cap\mathcal{M}^c\\
        \eta,&\text{if~}h\in\{\delta_Y(h)/\delta_R(h)<\lambda^*,\delta_Y(h)<0\}\cap\mathcal{M}^c,
             \end{cases}
    \end{split}
  \end{equation}
where $\lambda^*$ is the same one in the definition of $g^*$.
\end{theorem}

By comparing the expressions for $g^*$ and $f^*$, we immediately conclude that they have the same signs so solving (\ref{single_stage_surrogate}) leads to a Fisher consistent solution to the original problem in (\ref{single_stage_original}). The proof consists of several steps. For any decision function $f$, we say that $f$ is feasible meaning that $f$ satisfies the risk constraint in the surrogate problem (\ref{single_stage_surrogate}), and for any two feasible functions, $f_1$ and $f_2$, ``$f_1$ is non-inferior to $f_2$'' means that the objective function in (\ref{single_stage_surrogate}) is less than or equal to the one for $f_2$, and ``$f_1$ is superior to $f_2$'' if the objective function is strictly less than.

From now on, we assume $\eta\in(0,1]$ and $\tau\in (\tau_{\min},\tau_{\max})$.
By the definitions of $\mathcal{V}_{\phi}$ and $\mathcal{R}_{\psi}$, we note 
$$E\bigg[\frac{Y\phi(Af(H))}{p(A|H)}\bigg]=-E[\mathcal{V}_{\phi}(f,H)],$$
$$E\bigg[\frac{R\psi(Af(H),\eta)}{p(A|H)}\bigg]=E[\mathcal{R}_{\psi}(f, \eta, H)].$$

\noindent{\textbf{Proof of \cref{thm:s1}:}}\\
Step 1. 
We show that the value for the optimal solution, $f^*$, can be restricted within $[-1,1]$. That is, the following lemma holds.
\begin{lemma}
  \label[lemma]{lemma:1}
  For any feasible decision function $f(h)$, define $\widetilde f(h)=\min(\max(f(h),-1),1)$ as the truncated $f$ at -1 and 1. Then $\widetilde f$ is non-inferior to $f$.
\end{lemma}

\begin{proof}
  Note that $\psi(h,\eta)=\psi(1,\eta)$ for any $h>1$ and $\psi(h,\eta)=\psi(-1,\eta)$ for any $h<-1$. Thus, it follows from $\eta\le 1$ that $E[\mathcal{R}_{\psi}(\widetilde f,\eta,H)]=E[\mathcal{R}_{\psi}(f,\eta,H)]\le\tau$, so $\widetilde f$ is feasible. Moreover, it is easy to see that if $f(h)>1$, then $\widetilde f(h)=1$ so
  $$E\left[\frac{Y\phi(Af(H))}{p(A|X)}\bigg|H=h\right]=E[Y|A=-1,H=h](1+f(h))$$
  $$\ge 2E[Y|A=-1,H=h]=E\left[\frac{Y\phi(A\widetilde f(H))}{p(A|X)}\bigg|H=h\right].$$
 Similarly, if $f(h)<-1$, 
  $$E\left[\frac{Y\phi(Af(H))}{p(A|X)}\bigg|H=h\right]=E[Y|A=1,H=h](1-f(h))$$
  $$\ge 2E[Y|A=-1,H=h]=E\left[\frac{Y\phi(A\widetilde f(H))}{p(A|X)}\bigg|H=h\right].$$
 Since $f(h)=\widetilde f(h)$ when $|f(h)|\le 1$, we conclude
 $$E\left[\frac{Y\phi(A f(H))}{p(A|X)}\right]\ge E\left[\frac{Y\phi(A\widetilde f(H))}{p(A|X)}\right].$$
 Thus, \cref{lemma:1} holds.
\end{proof}

Step 2. We characterize the expression of $f^*(h)$ for $h\in \mathcal{M}$, which is the region where the beneficial treatment also reduces the risk.
\begin{lemma}
  \label[lemma]{lemma:2}
  For any feasible function $f$ with $|f|\le 1$, we define 
  $$\widetilde f(h)=f(h) \mathbb{I}(h\in {\cal M}^c) + \textrm{sign}(\delta_Y(h))\mathbb{I}(h\in \mathcal{M}).$$
 Then $\widetilde f$ is non-inferior to $f$.
\end{lemma}

\begin{proof}
For any $h\in\mathcal{M}$ with $\delta_Y(h)>0$ and $\delta_R(h)<0$, 
$\mathcal{R}_{\psi}(s,\eta,h)$ is minimized when $s\in [\eta, 1]$, while
$\mathcal{V}_{\phi}(s,h)$
is maximized at $s=1$.
Since $\widetilde f(h)=1$ for any $h\in\mathcal{M}$, we have $\mathcal{R}_{\psi}(\widetilde f(h),\eta,h)\le \mathcal{R}_{\psi}(f(h), \eta, h)$ and $\mathcal{V}_{\phi}(\widetilde f(h),h)\ge \mathcal{V}_{\phi}(f(h),h).$
The same inequalities hold for $h$ with $\delta_Y(h)<0$ and $\delta_R(h)>0$. In other words, they hold for any $h\in \mathcal{M}$.

Since $\widetilde f(h)=f(h)$ for $h\in \mathcal{M}^c$, 
 \begin{equation*}
    \begin{split}
      &E[\mathcal{R}_{\psi}(f,\eta,H)]-E[\mathcal{R}_{\psi}(\widetilde f,\eta,H)]
      =E[(\mathcal{R}_{\psi}(f,\eta,H)-\mathcal{R}_{\psi}(\widetilde f,\eta,H))\mathbb{I}(H\in\mathcal{M})]
      \geq0,
    \end{split}
  \end{equation*}
and
similarly, $E[\mathcal{V}_{\phi}(f,H)]-E[\mathcal{V}_{\phi}(\widetilde f,H)]\le 0$.
We conclude that $\widetilde f$ is non-inferior to $f$.
\end{proof}

Step 3. From steps 1 and 2, we can restrict $f$ to satisfy $|f|\le 1$ and $f(h)=\textrm{sign}(\delta_Y(h))$ for $h\in \mathcal{M}$. Furthermore, since $\tau_{\max}$ is the risk under decision rule $\textrm{sign}(\delta_{Y}(h))$, $\tau<\tau_{\max}$ implies that 
$$P(f(H)\neq \textrm{sign}(\delta_Y(H)),H\in\mathcal{M}^c)>0.$$
In this step, we wish to show that the optimal solution should attain the risk bound, i.e., $E[\mathcal{R}_{\psi}(f,\eta,H)]=\tau$. Otherwise, assume for some feasible solution $f$ such that $E[\mathcal{R}_{\psi}(f,\eta,H)]=\tau_0<\tau.$
Consider two sets
  $$\mathcal{D}^+=\{h\in\mathcal{H}:f(h)< 1,\delta_{Y}(h)>0\}\cap\mathcal{M}^c$$
  $$\mathcal{D}^-=\{h\in\mathcal{H}:f(h)> -1,\delta_{Y}(h)<0\}\cap\mathcal{M}^c,$$
then $P(\mathcal{D}^+)+P(\mathcal{D}^{-})>0$. Without loss of generality, we assume that $P(\mathcal{D}^+)>0$. We construct 
  $$\widetilde f(h)=
  \begin{cases}
    f(h),&\text{if}~h\notin \mathcal{D}^+\\
    \min\bigg(f(h)+\frac{\eta(\tau-\tau_0)}{MP(\mathcal{D}^+)},1\bigg),&\text{if}~h\in\mathcal{D}^+,
  \end{cases}
  $$
where $M$ is the bound for $R$.

For $h \in \mathcal{D}^+$, $\mathcal{V}_{\phi}(\widetilde f(h), h)>\mathcal{V}_{\phi}(f(h),h)$ since $1\geq\widetilde{f}(h)>f(h)$ and $\mathcal{V}_{\phi}(s, h)$ is an strictly increasing function of $s\in[-1,1]$ due to $\delta_Y(h)>0$. We immediately conclude
$E[\mathcal{V}_{\phi}(\widetilde f, H)]> E[\mathcal{V}_{\phi}(f, H)].$
On the other hand, $\mathcal{R}_{\psi}(s, \eta, h)$ is 
a piecewise linear function of $s$ with absolute value of slopes no larger than 
$$\frac{\max(E[R|H=h, A=1],E[R|H=h, A=-1])}{\eta}\le\frac{M}{\eta}.$$
Hence, it follows that 
  \begin{align*}
    E[\mathcal{R}_{\psi}(\widetilde f,\eta,H)]=&E[\mathcal{R}_{\psi}(\widetilde f,\eta,H)]-E[\mathcal{R}_{\psi}(f,\eta,H)]+E[\mathcal{R}_{\psi}(f,\eta,H)]\\
    \le&E[(\mathcal{R}_{\psi}(\widetilde f,\eta,H)-\mathcal{R}_{\psi}(f,\eta,H))\mathbb{I}(H\in\mathcal{D}^+)]+\tau_0\\
    \le&\frac{M}{\eta}\frac{\eta(\tau-\tau_0)}{MP(\mathcal{D}^+)}P(\mathcal{D}^+)+\tau_0=\tau.
  \end{align*}
As a result, $\widetilde f$ is superior to $f$ with a strictly larger objective function, a contradiction. In other words, the expected risk for the optimal solution should attain the bound.

With steps 1-3, we can restrict within the class
$${\cal W}=\left\{f: |f|\le 1, f(h)=\textrm{sign}(\delta_Y(h)) \textrm{ for } h\in \mathcal{M}, E[R_{\psi}(f,\eta, H)]=\tau\right\}$$
to find the optimal decision function.

Step 4. We derive the expression of the optimal function for $f$ by considering solving a Lagrange multiplier for the problem (\ref{single_stage_surrogate}):
\begin{equation}
  \label{langrange_function}
  \begin{split}
    \underset{f\in\mathcal{W}}{\max} -E\bigg[\frac{Y\phi(Af(H))}{p(A|H)}\bigg]-\nu \bigg(E\bigg[\frac{R\psi(Af(H))}{p(A|H)}\bigg]-\tau\bigg),
  \end{split}
\end{equation}
where $\nu$ is a constant to be determined by the constraint $\tau$ in ${\cal W}$. We maximize the above function by maximizing the conditional mean of the term in the expectation given $H=h$ for every $h$, which is given by
$$G(f)\equiv \mathcal{V}_{\phi}(f, h)-\nu {R}_{\psi}(f, \eta, h).$$
Note that $G(f)$ is now a function w.r.t. the value of $f$ given fixed $h$. Since $f\in [-1,1]$ and $f$ in ${\cal W}$ is already given for $h\in \mathcal{M}$, it suffices to examine that for $h\in \mathcal{M}^c$. In addition, $G(f)$ is a piecewise linear function for $f\in [-1,-\eta], (-\eta, 0]$, $(0, \eta]$ and $(\eta, 1]$. Thus,
the maximizer can only be achieved at points $-1,-\eta, 0, \eta$ and 1. 
Note that $R$ is assumed to be positive,
$G'(0)=-\nu/\eta (E[R|H=h,A=1]+E[R|H=h,A=-1])<0$ if $\nu>0$, or $>0$ if $\nu<0$. For $\nu=0$,
$G(0)=-E[Y|H=h,A=1]-E[Y|H=h,A=-1]= (G(1)+G(-1))/2$. Thus, the maximum for $G(f)$ can always be attained at $f$ which is not zero. In other words, we only need to compare the values at $f\in\{-1,-\eta,\eta,1\}$.
 
Simple calculation gives 
$$G(-1)=-2 E[Y|H=h,A=1]-\nu E[R|H=h, A=-1], $$
$$ G(-\eta)=-(1+\eta)E[Y|H=h,A=1]-(1-\eta)E[Y|H=h,A=-1]
-\nu E[R|H=h, A=-1], 
$$
$$G(\eta)=-(1-\eta)E[Y|H=h,A=1]-(1+\eta)E[Y|H=h,A=-1]-\nu E[R|H=h, A=1], $$
and
$$ G(1)=-2E[Y|H=h,A=-1]-\nu E[R|H=h, A=1].$$
When $\delta_Y(h)>0$ so $\delta_R(h)$ is also positive, it is straightforward to check $G(1)>G(\eta)$ and $G(-\eta)>G(-1)$. Note $G(1)-G(-\eta)=(1+\eta)\delta_Y(h)-\lambda \delta_R(h)$ so we immediately conclude
that the optimal value for $f$ should be 1, if $\delta_Y(h)>\lambda$, where $\lambda=\nu/(1+\eta)$, and it is $-\eta$ otherwise.
When $\delta_Y(h)\le 0$, we use the same arguments to obtain that the optimal value for $f$ should be -1 if $\delta_Y(h)>\lambda$, and it is $\eta$ otherwise. 
Therefore, the optimal function maximizing the Lagrange multiplier for any fixed $\nu$ (equivalently, $\lambda)$ has the same expression as (\ref{f_expression}). 

Next, we show that there is some positive $\lambda^*=\nu^*/(1+\eta)$ such that
$$E[R\mathbb{I}(Ag^*(H)>0)/p(A|H)]=E[R\mathbb{I}(Af^*(H)>0)/p(A|H)]=E[\mathcal{R}_{\psi}(f^*,\eta,H)]=\tau.$$
The first equality follows from the fact that $\text{sign}(g^*)=\text{sign}(f^*)$, and the second equality follows from that $R_{\psi}(s, \eta, h)$ is constant for any $s\in[-1,-\eta]$ and $s\in[\eta,1]$. To prove the existence of $\lambda^*$, we notice
  \begin{equation}
    \label{gamma_expression}
    \begin{split}
      \Gamma(\lambda)\equiv &E[R\mathbb{I}(Af^*(H)>0)/p(A|H)]\\
      =&E[E[R|H,A=1]\mathbb{I}(H\in\{\delta_Y(h)>0\}\cap\mathcal{M})]\\
      &+E[E[R|H, A=-1]\mathbb{I}(H\in\{\delta_Y(h)<0\}\cap\mathcal{M})]\\
      &+E[E[R|H,A=1]\mathbb{I}(H\in\{\delta_Y(h)/\delta_R(h)>\lambda,\delta_Y(h)>0\}\cap\mathcal{M}^c)]\\
      &+E[E[R|H,A=-1]\mathbb{I}(H\in\{\delta_Y(h)/\delta_R(h)<\lambda,\delta_Y(h)>0\}\cap\mathcal{M}^c)]\\
      &+E[E[R|H,A=-1]\mathbb{I}(H\in\{\delta_Y(h)/\delta_R(h)>\lambda,\delta_Y(h)<0\}\cap\mathcal{M}^c)]\\
      &+E[E[R|H,A=1]\mathbb{I}(H\in\{\delta_Y(h)/\delta_R(h)<\lambda,\delta_Y(h)<0\}\cap\mathcal{M}^c)]
    \end{split}
  \end{equation}
is a continuous function of $\lambda$ since $\delta_Y(H)/\delta_R(H)$ has a continuous density function. Furthermore, 
$\Gamma(\infty)=\tau_{\min}$, $\Gamma(0)=\tau_{\max}.$ Thus, there exists some $\lambda^*>0$ such that $\Gamma(\lambda^*)=\tau$.

Finally, for any $f$, based on steps 1-3, we have
$$-E\bigg[\frac{Y\phi(Af(H))}{p(A|H)}\bigg]\le \max_{f\in {\cal W}}\left\{
-E\bigg[\frac{Y\phi(Af(H))}{p(A|H)}\bigg]\right\}.$$
On the other hand, for any $f\in {\cal W}$, we have $E[R\psi(Af(H))/p(A|H)]=\tau$ and
$$-E\bigg[\frac{Y\phi(Af(H))}{p(A|H)}\bigg]-\nu^* \bigg(E\bigg[\frac{R\psi(Af(H))}{p(A|H)}\bigg]-\tau\bigg)$$
$$\le 
-E\bigg[\frac{Y\phi(Af^*(H))}{p(A|H)}\bigg]-\nu^* \bigg(E\bigg[\frac{R\psi(Af^*(H))}{p(A|H)}\bigg]-\tau\bigg).$$
The inequality above holds since $f^*$ maximizes the Lagrange function (\ref{langrange_function}) under multiplier $\nu^*$ for any $f\in\mathcal{W}$. Therefore, for any $f$ we have
$$E\bigg[\frac{Y\phi(Af(H))}{p(A|H)}\bigg]\ge 
E\bigg[\frac{Y\phi(Af^*(H))}{p(A|H)}\bigg].$$
In other words, $f^*$ given by (\ref{f_expression}) is the optimal solution to the problem (\ref{single_stage_surrogate}). We thus complete the proof of \cref{thm:s1}.

\subsection{Proof of \cref{fisher_consistency} for $T\geq2$}
\label{A:2}
 
Start from stage $T$. For any given $f_1, ..., f_{T-1}$, we consider $f_T$ maximizing
$$E\bigg[\frac{(\sum_{t=1}^T Y_t)\mathbb{I}(A_Tf_T(H_T)>0)}{p(A_T|H_T)}\frac{\prod_{t=1}^{T-1} \mathbb{I}(A_tf_t(H_t)>0)}{\prod_{t=1}^{T-1}
p(A_t|H_t)}\bigg]$$
subject to constraint
$$E\left[\frac{R_T\mathbb{I}(A_Tf_T(H_T)>0)}{p(A_T|H_T)}\frac{\prod_{t=1}^{T-1} \mathbb{I}(A_tf_t(H_t)>0)}{\prod_{t=1}^{T-1}
p(A_t|H_t)}\right]\le \tau_T.$$
Based on Theorem 1 in \citet{wang_learning_2018}, the optimal solution can be chosen as
\begin{equation*}
    \widetilde g_T^*(h)=\begin{cases}
      \textrm{sign}(\delta_{\widetilde Y}(h)),&\text{if~}h\in\widetilde {\mathcal{M}}\\
      1,&\text{if~}h\in\{\delta_{\widetilde Y}(h)/\delta_{\widetilde R}(h)>\widetilde\lambda,\delta_{\widetilde Y}(h)>0\}\cap\widetilde{\mathcal{M}}^c\\
      -1,&\text{if~}h\in\{\delta_{\widetilde Y}(h)/\delta_{\widetilde R}(h)<\widetilde\lambda,\delta_{\widetilde Y}(h)>0\}\cap\widetilde{\mathcal{M}}^c\\
      -1,&\text{if~}h\in\{\delta_{\widetilde Y}(h)/\delta_{\widetilde R}(h)>\widetilde\lambda,\delta_{\widetilde Y}(h)<0\}\cap\widetilde{\mathcal{M}}^c\\
      1,&\text{if~}h\in\{\delta_{\widetilde Y}(h)/\delta_{\widetilde R}(h)<\widetilde\lambda,\delta_{\widetilde Y}(h)<0\}\cap\widetilde{\mathcal{M}}^c,
    \end{cases}
\end{equation*}
where 
$$\delta_{\widetilde Y}(h)=(E[Q_T|H_T=h, A_T=1]-E[Q_T|H_T=h, A_T=-1])\frac{\prod_{t=1}^{T-1} \mathbb{I}(A_tf_t(H_t)>0)}{\prod_{t=1}^{T-1}
p(A_t|H_t)},$$
$$\delta_{\widetilde R}(h)=(E[R_T|H_T=h, A_T=1]-E[R_T|H_T=h, A_T=-1])\frac{\prod_{t=1}^{T-1} \mathbb{I}(A_tf_t(H_t)>0)}{\prod_{t=1}^{T-1}
p(A_t|H_t)},
$$
$\widetilde {\cal M}=\left\{h: \delta_{\widetilde Y}(h) \delta_{\widetilde R}(h)<0\right\}, $ and 
and $\widetilde\lambda$ satisfies 
$$E\left[\frac{R_T\mathbb{I}(A_T\widetilde g_T^*(H_T)>0)}{p(A_T|H_T)}\frac{\prod_{t=1}^{T-1} \mathbb{I}(A_tf_t(H_t)>0)}{\prod_{t=1}^{T-1}p(A_t|H_t)}\right]=\tau_T.$$
Note that for $h$ in the support of $H_t$ where $A_tf_t(H_t)\le 0$ for any $t=1,...,T-1$, $\widetilde g_T^*(h)$ can be any arbitrary value since it does not affect the value and risk expectations. On the other hand, recall that $g^*_T(h)$ is the function maximizing
$$E\bigg[\frac{(\sum_{t=1}^T Y_t)\mathbb{I}(A_Tf_T(H_T)>0)}{p(A_T|H_T)}\bigg]$$
subject to constraint
$$E\left[\frac{R_T\mathbb{I}(A_Tf_T(H_T)>0)}{p(A_T|H_T)}\right]\le \tau_T.$$
Based on Theorem 1 in \citet{wang_learning_2018}, $g_T^*$ is given as
\begin{equation*}
  g_T^*(h)=\begin{cases}
      \textrm{sign}(\delta_{Q_T}(h)),&\text{if~}h\in \mathcal{M}\\
      1,&\text{if~}h\in\{\delta_{Q_T}(h)/\delta_{R_T}(h)>\lambda^*,\delta_{Q_T}(h)>0\}\cap\mathcal{M}^c\\
      -1,&\text{if~}h\in\{\delta_{Q_T}(h)/\delta_{R_T}(h)<\lambda^*,\delta_{Q_T}(h)>0\}\cap\mathcal{M}^c\\
      -1,&\text{if~}h\in\{\delta_{Q_T}(h)/\delta_{R_T}(h)>\lambda^*,\delta_{Q_T}(h)<0\}\cap\mathcal{M}^c\\
      1,&\text{if~}h\in\{\delta_{Q_T}(h)/\delta_{R_T}(h)<\lambda^*,\delta_{Q_T}(h)<0\}\cap\mathcal{M}^c,
    \end{cases}
\end{equation*}
where ${\cal M}=\left\{h: \delta_{Q_T}(h) \delta_{R_T}(h)<0\right\}$, and $\lambda^*$ satisfies 
$$E\left[\frac{R_T\mathbb{I}(A_T\widetilde g_T^*(H_T)>0)}{p(A_T|H_T)}\frac{\prod_{t=1}^{T-1} \mathbb{I}(A_tf_t(H_t)>0)}{\prod_{t=1}^{T-1}
p(A_t|H_t)}\right]=\tau_T.$$

From the above two expressions, it is clear that on the set when $A_tf_t(H_t)>0$ for all $t=1,..., T-1$, $\widetilde g_T^*(h)$ takes the same form as the solution as $ g^*_T(h)$. Furthermore, due to the conclusion in Section \ref{A:0}, we have 
$$E\left[\frac{R_T\mathbb{I}(A_Tf_T(H_T)>0)}{p(A_T|H_T)}\frac{\prod_{t=1}^{T-1} \mathbb{I}(A_tf_t(H_t)>0)}{\prod_{t=1}^{T-1}
p(A_t|H_t)}\right]=E\left[\frac{R_T\mathbb{I}(A_Tf_T(H_T)>0)}{p(A_T|H_T)}\right].$$
Thus, we conclude that $\widetilde\lambda$ can be chosen to be the same as $\lambda^*$ so $\widetilde g_T^*(h)$ can be chosen to be exactly the same as $g_T^*(h)$. In other words,
$$\mathcal{V}(f_1,...,f_{T-1}, g_T^*)\ge \mathcal{V}(f_1,...,f_T)$$
and $g_T^*$ satisfies
$$E\left[\frac{R_T\mathbb{I}(A_Tg_T^*(H_T)>0)}{p(A_T|H_T)}\frac{\prod_{t=1}^{T-1} \mathbb{I}(A_tf_t(H_t)>0)}{\prod_{t=1}^{T-1}
p(A_t|H_t)}\right]=\tau_T.$$
By \cref{thm:s1}, both $g_T^*$ and $f_{T}^*$ have the same signs. Therefore, 
$$\mathcal{V}(f_1,...,f_{T-1}, f_{T}^*)\ge \mathcal{V}(f_1,...,f_T)$$
and $f_{T}^*$ satisfies
$$E\left[\frac{R_T\mathbb{I}(A_Tf_{T}^*(H_T)>0)}{p(A_T|H_T)}\right]=\tau_T.$$
Once $f_{T}^*$ is determined, we consider the $T-1$ stage. Now the original problem (\ref{ow_prob}) becomes
\begin{eqnarray*}
 \max  {\cal V}(f_1,...,f_{T-1}, f_{T}^*)=&E\bigg[\frac{(\sum_{t=1}^TY_t)\mathbb{I}(A_Tf_{T}^*(H_T)>0)}{p(A_T|H_T)}\frac{\prod_{t=1}^{T-1}\mathbb{I}(A_tf_t(H_t)>0)}{\prod_{t=1}^{T-1}p(A_t|H_t)}\bigg]&\\
 \text{subject~to~~}&E\bigg[\frac{R_t\mathbb{I}(A_tf_t(H_t)>0)}{p(A_t|H_t)}\bigg]\le \tau_t,&t=1,...,T-1.
\end{eqnarray*}
We repeat the same arguments as for stage $T$ as before, to conclude
$$\mathcal{V}(f_1,...,f_{T-1}^*, f_{T}^*)\ge \mathcal{V}(f_1,...,f_{T-1}, f_T)$$
and $f_{T-1}^*$ satisfies
$$E\left[\frac{R_{T-1}\mathbb{I}(A_{T-1}f_{T-1}^*(H_{T-1})>0)}{p(A_{T-1}|H_{T-1})}\right]=\tau_{T-1}.$$

We continue this proof till $t=1$ so conclude that $(f_{1}^*, ..., f_{T}^*)$ maximizes the multistage value and satisfies the constraints over all the stages. The above arguments also show that $f_t^*$ has the same sign as $g_{t}^*$. \cref{fisher_consistency} thus holds.

%%%%%%%%%%%%%%%%%% Section A.2 %%%%%%%%%%%%%%%%%%%%%%%%%%%

\section{Proof of \cref{risk_bound}}

Instead, we prove a more general version of \cref{risk_bound}. 

\begin{theorem}
  \label[theorem]{thm:s2}
  In addition to the conditions in \cref{fisher_consistency}, suppose that \cref{GNE} holds and $H_t$ takes value from a compact subset of $\mathbb{R}^{d_t}$ for $t=1,...,T$. Let $(\tau_1,...,\tau_T)$ and $(\delta_{0,1},...,\delta_{0,T})$ denote the constraints and corresponding constants in \cref{GNE}. Let $\delta_t>0$, $1\le x_t$, $0<\theta_{1,t}$, $0<\theta_{2,t}$, $0<\nu_{1,t}<2$, $0<\nu_{2,t}\le2$ for $t=1,...,T$. Give positive parameter $\lambda_{n,t}\rightarrow0$ and $\sigma_{n,t}\rightarrow\infty$, and let
  \begin{align*}
    \xi_{n,t}^{(1)}=c\bigg(&\frac{2M}{c_1}\sqrt{\frac{M}{c_1\lambda_{n,t}}+\sigma_{n,t}^{d_t}}+\\
    &\lambda_{n,t}\bigg(\frac{M}{c_1\lambda_{n,t}}+\sigma_{n,t}^{d_t}\bigg)\bigg)n^{-1/2}(\sigma_{n,t}^{(1-\nu_{1,t}/2)(1+\theta_{1,t})d_t/2}+2\sqrt{2x_t}+2x_t/\sqrt{n}),
  \end{align*}
  $\xi_{n,t}^{(2)}=c(\lambda_{n,t}\sigma_{n,t}^{d_t}+\sigma_{n,t}^{-\alpha_td_t})$ and $\xi_{n,t}=\xi_{n,t}^{(1)}+\xi_{n,t}^{(2)}$. In addition, let
  \begin{equation*}
    \epsilon_{n,t}'=\delta_{t}+C_{1,t}\sigma_{n,t}^{-\alpha_td_t}\eta_{n,t}^{-1}+C_{3,t}n^{-1/2}\sigma_{n,t}^{(1-\nu_{2,t}/2)(1+\theta_{2,t})d_t/2}\bigg(\frac{M}{c_1\lambda_{n,t}}+\sigma_{n,t}^{d_t}\bigg)^{\nu_{2,t}/4}\eta_{n,t}^{-\nu_{2,t}/2}
  \end{equation*}
  and
  $$h_t(n,x_t)=2\exp\bigg(-\frac{2n\delta_{0,t}^2c_1^2}{M^2}\bigg)+2\exp\bigg(-\frac{n\delta_{t}^2c_1^2}{2M^2}\bigg)+\exp(-x_t).$$
  Then for any $n\geq1$ and $(\lambda_{n,t},\sigma_{n,t},\eta_{n,t})$ such that $$C_{1,t}\sigma_{n,t}^{-\alpha_td_t}\eta_{n,t}^{-1}\le \delta_{0,t},$$
  $$C_{2,t}\sigma_{n,t}^{(1-\nu_{1,t}/2)(1+\theta_{1,t})d_t}\le1,$$ 
  $\epsilon_{n,t}'<2\delta_{0,t}$, and $x_t\geq1$, with probability at least $1-\sum_{t=1}^Th_t(n,x_t)$, we have
  \begin{equation}
    \label{value_error_bound}
    |\mathcal{V}(\widehat{f}_{1},...,\widehat{f}_{T})-\mathcal{V}(g_1^*,...,g_T^*)|\le \sum_{t=1}^T(c_1/5)^{1-t}(\xi_{n,t}+(T-t+1)M\eta_{n,t}+2c\epsilon'_{n,t}).
  \end{equation}
  Moreover, with probability at least $1-h_t(n,x_t)$ the risk induced by $\widehat{f}_{t}$ satisfies
  \begin{equation}
    \begin{split}
      \label{risk_error_bound}
      E\bigg[&\frac{R_t\mathbb{I}(A_t\widehat{f}_{t}(H_t)>0)}{p(A_t|H_t)}\bigg]\le \tau_t+\delta_{t}+C_{3,t}\sigma_{n,t}^{(1-\nu_{2,t}/2)(1+\theta_{2,t})d_t/2}\bigg(\frac{M}{c_1\lambda_{n,t}}+\sigma_{n,t}^{d_t}\bigg)^{\nu_{2,t}/4}\eta_{n,t}^{-\nu_{2,t}/2}.
    \end{split}
  \end{equation}
  Here, $c$ in front of $\xi_{n,t}^{(1)}$ is a positive constant only depends on $(\nu_{1,t},\theta_{1,t},d_t)$, $c$ in front of $\xi_{n,t}^{(2)}$ is a positive constant only depends on $(\alpha_t,d_t,K_t,M)$ and $c$ of $\epsilon_{n,t}'$ is a positive constant only depends on $(\tau_t,\delta_{0,t})$. $C_{1,t}$ is a positive constant depend on $(\alpha_t,d_t,K_t,M)$, $C_{2,t}$ is a positive constant depends on $(\nu_{1,t},\theta_{1,t},d_t)$, $C_{3,t}$ a positive constant depends on $(\nu_{2,t},\theta_{2,t},d_t,c_1,M)$.
\end{theorem}

\cref{risk_bound} can be obtained from \cref{thm:s2} by setting $\theta_t=\theta_{1,t}=\theta_{2,t}$, $\nu_t=\nu_{1,t}=\nu_{2,t}$, $x_t=\sigma_{n,t}^{(1-\nu_{t}/2)(1+\theta_{t})d_t}$, and $C_i=\sup_t C_{i,t}$ for $i=1,2,3$.
We first prove \cref{thm:s2} by for $T=1$ and then extend the result to $T\ge 2$. 

\subsection{Proof of \cref{thm:s2} for $T=1$}

Since $T=1$, we omit the subscript for the stage in this subsection, so all the notations are the same as in Section A.1.2. Since $\tau$ is necessary for the proof, we use $f^*_{\tau}$ to refer to $f^*$ that solves (\ref{single_stage_surrogate}) corresponding to $\tau$ and shifting parameter $\eta_n$.

\subsubsection{Excessive risk}
\label{sec1:1}
In this section, we prove some preliminary lemmas. \cref{lemma:3} shows that the regret from the optimal decision function solving the original problem (\ref{single_stage_original}) is bounded by the regret from the one solving the surrogate problem (\ref{single_stage_surrogate}), plus an additional biased term of order $O(\eta_n)$. \cref{lemma:4} shows that the optimal value using the surrogate loss is Lipschitz continuous with respect to $\tau$. 

\begin{lemma}
  \label[lemma]{lemma:3}
  Under the condition of \cref{thm:s1}, for any $f:\mathcal{H}\rightarrow\mathbb{R}$ and any $\eta_n\in(0,1]$, we have
  $$\mathcal{V}(f_{\tau}^*)-\mathcal{V}(f)\le E[\mathcal{V}_{\phi}(f_{\tau}^*, H)]-E[\mathcal{V}_{\phi}(f, H)]+M\eta_n.$$
\end{lemma}

\begin{proof}
  \cref{thm:s1} shows that $f_{\tau}^*$ must have expression (\ref{f_expression}) almost surely. 
  Let $\widetilde {\cal V}(f,h)=\mathbb{I}(f(h)>0)E[Y|H=h, A=1]+\mathbb{I}(f(h)\le 0)E[Y|H=h, A=-1]$.
For any $h\in\{\delta_Y(h)>0\}$, we consider the following 6 scenarios:
  \begin{enumerate}
    \item When $h\in\mathcal{M}$, $f_{\tau}^*(h)=1$ and $f(h)>0$, we have $\widetilde{\mathcal{V}}(f_{\tau}^*,h)-\widetilde{\mathcal{V}}(f,h)=0$ and 
    $$\mathcal{V}_{\phi}(f_{\tau}^*,h)-\mathcal{V}_{\phi}(f,h)=
    \begin{cases}
      (1-f(h))\delta_Y(h),&f(h)\le1\\
      (f(h)-1)m_Y(h,-1),&f(h)>1,
    \end{cases}  
    $$
    which implies $\mathcal{V}_{\phi}(f_{\tau}^*,h)-\mathcal{V}_{\phi}(f,h)\geq\widetilde{\mathcal{V}}(f_{\tau}^*,h)-\widetilde{\mathcal{V}}(f,h)$.
    \item When $h\in\mathcal{M}$, $f_{\tau}^*(h)=1$ and $f(h)\le 0$, we have $\widetilde{\mathcal{V}}(f_{\tau}^*,h)-\widetilde{\mathcal{V}}(f,h)=\delta_Y(h)$ and 
    $$\mathcal{V}_{\phi}(f_{\tau}^*,h)-\mathcal{V}_{\phi}(f,h)=
    \begin{cases}
      (1-f(h))\delta_Y(h),&f(h)\geq-1\\
      2\delta_Y(h)+(-f(h)-1)m_Y(h,1),&f(h)<-1,
    \end{cases}  
    $$
    which implies $\mathcal{V}_{\phi}(f_{\tau}^*,h)-\mathcal{V}_{\phi}(f,h)\geq\widetilde{\mathcal{V}}(f_{\tau}^*,h)-\widetilde{\mathcal{V}}(f,h)$.
    \item When $h\in\mathcal{M}^c$, $f_{\tau}^*(h)=1$ and $f(h)>0$, we have $\widetilde{\mathcal{V}}(f_{\tau}^*,h)-\widetilde{\mathcal{V}}(f,h)=0$ and
    $$\mathcal{V}_{\phi}(f_{\tau}^*,h)-\mathcal{V}_{\phi}(f,h)=
    \begin{cases}
      (1-f(h))\delta_Y(h),&f(h)\le1\\
      (f(h)-1)m_Y(h,-1),&f(h)>1,
    \end{cases}  
    $$
    in which case $\mathcal{V}_{\phi}(f_{\tau}^*,h)-\mathcal{V}_{\phi}(f,h)\geq\widetilde{\mathcal{V}}(f_{\tau}^*,h)-\widetilde{\mathcal{V}}(f,h)$.
    \item When $h\in\mathcal{M}^c$, $f_{\tau}^*(h)=1$ and $f(h)\le 0$, we have $\widetilde{\mathcal{V}}(f_{\tau}^*,h)-\widetilde{\mathcal{V}}(f,h)=\delta_Y(h)$ and
    $$\mathcal{V}_{\phi}(f_{\tau}^*,h)-\mathcal{V}_{\phi}(f,h)=
    \begin{cases}
      (1-f(h))\delta_Y(h),&f(h)\geq-1\\
      2\delta_Y(h)+(-f(h)-1)m_Y(h,1),&f(h)<-1,
    \end{cases}  
    $$
    in which case $\mathcal{V}_{\phi}(f_{\tau}^*,h)-\mathcal{V}_{\phi}(f,h)\geq\widetilde{\mathcal{V}}(f_{\tau}^*,h)-\widetilde{\mathcal{V}}(f,h)$.
    \item When $h\in\mathcal{M}^c$, $f_{\tau}^*(h)=-\eta_n$ and $f(h)>0$, we have $\widetilde{\mathcal{V}}(f_{\tau}^*,h)-\widetilde{\mathcal{V}}(f,h)=-\delta_Y(h)$ and
    $$\mathcal{V}_{\phi}(f_{\tau}^*,h)-\mathcal{V}_{\phi}(f,h)=
    \begin{cases}
      -f(h)\delta_Y(h)-\eta_n\delta_Y(h),&f(h)\le1\\
      -\delta_Y(h)-\eta_n\delta_Y(h)+(f(h)-1)m_Y(h,-1),&f(h)>1.
    \end{cases}  
    $$
    Thus, $\mathcal{V}_{\phi}(f_{\tau}^*,h)-\mathcal{V}_{\phi}(f,h)\geq-\delta_Y(h)-\eta_n\delta_Y(h)=\widetilde{\mathcal{V}}(f_{\tau}^*,h)-\widetilde{\mathcal{V}}(f,h)-\eta_n\delta_Y(h)$.
    \item When $h\in\mathcal{M}^c$, $f_{\tau}^*(h)=-\eta_n$ and $f(h)\le 0$, we have $\widetilde{\mathcal{V}}(f_{\tau}^*,h)-\widetilde{\mathcal{V}}(f,h)=0$ and
    $$\mathcal{V}_{\phi}(f_{\tau}^*,h)-\mathcal{V}_{\phi}(f,h)=
    \begin{cases}
     -f(h)\delta_Y(h)-\eta_n\delta_Y(h),&f(h)\geq-1\\
      (f(h)-1)m_Y(h,1)+(1-\eta_n)\delta_Y(h),&f(h)<-1.
    \end{cases} 
    $$
    Thus, $\mathcal{V}_{\phi}(f_{\tau}^*,h)-\mathcal{V}_{\phi}(f,h)\geq-\eta_n\delta_Y(h)=\widetilde{\mathcal{V}}(f_{\tau}^*,h)-\widetilde{\mathcal{V}}(f,h)-\eta_n\delta_Y(h)$.
  \end{enumerate}
  \vspace{-5pt}
  Hence, by combining all these cases, we conclude that
  $$\widetilde{\mathcal{V}}(f_{\tau}^*,h)-\widetilde{\mathcal{V}}(f,h)\le \mathcal{V}_{\phi}(f_{\tau}^*,h)-\mathcal{V}_{\phi}(f,h)+M\eta_n$$
  for any $\eta_n\in(0,1]$ and any decision function $f$. The same argument holds for any $h$ such that $\delta_Y(h)<0$. Consequently, since ${\cal V}(f)=E[\widetilde{\mathcal{V}}(f,H)]$,we have 
  $$\mathcal{V}(f_{\tau}^*)-\mathcal{V}(f)\le E[\mathcal{V}_{\phi}(f_{\tau}^*, H)]-E[\mathcal{V}_{\phi}(f, H)]+M\eta_n.$$
\end{proof}

\begin{lemma}
  \label[lemma]{lemma:4}
  For any $\delta>0$ and $\tau$ such that $[\tau-2\delta,\tau+2\delta]\subseteq(\tau_{\min},\tau_{\max})$, 
  $E[\mathcal{V}_{\phi}(f_{\tau}^*, H)]$, as a function of $\tau$, is Lipschitz continuous at $\tau$. 
\end{lemma}

\begin{proof}
Let $\tau_1=\tau$ and $\tau_2$ be any number in $[\tau-2\delta, \tau+2\delta]$. Without loss of generality, we assume $\tau_2<\tau_1$. We also let $f_1^*$ and $f_2^*$ be the optimal decision functions solving (\ref{single_stage_surrogate}) for $\tau_1$ and $\tau_2$, respectively, and their corresponding $\lambda^*$'s values are denoted as $\lambda_1$ and $\lambda_2$. According to (\ref{f_expression}), it is easy to verify that 
  \begin{equation*}
    \begin{split}
      &E[\mathcal{V}_{\phi}(f_1^*, H)]-E[\mathcal{V}_{\phi}(f_2^*, H)]\\
      =&E\bigg[(1+\eta_n)\delta_Y(H)\mathbb{I}\bigg(\lambda_1\le\frac{\delta_Y(H)}{\delta_R(H)}\le \lambda_2 \bigg)\mathbb{I}(\delta_Y(H)>0)\mathbb{I}(H\in\mathcal{M}^c)\bigg]\\
      &-E\bigg[(1+\eta_n)\delta_Y(H)\mathbb{I}\bigg(\lambda_1\le\frac{\delta_Y(H)}{\delta_R(H)}\le \lambda_2 \bigg)\mathbb{I}(\delta_Y(H)<0)\mathbb{I}(H\in\mathcal{M}^c)\bigg]\\
      =&(1+\eta_n)E\bigg[|\delta_Y(H)|\mathbb{I}\bigg(\lambda_1\le\frac{\delta_Y(H)}{\delta_R(H)}\le \lambda_2\bigg)\mathbb{I}(H\in\mathcal{M}^c)\bigg].
    \end{split}
  \end{equation*}
  On the other hand 
  \begin{equation*}
    \begin{split}
      \tau_1-\tau_2
      =&E[\mathcal{R}_{\psi}(f_1^*,\eta_n,H)\mathbb{I}(H\in\mathcal{M}^c)]-E[\mathcal{R}_{\psi}(f_2^*,\eta_n,H)\mathbb{I}(H\in\mathcal{M}^c)]\\
      =&E\bigg[\delta_R(H)\mathbb{I}\bigg(\lambda_1\le\frac{\delta_Y(H)}{\delta_R(H)}\le \lambda_2 \bigg)\mathbb{I}(\delta_Y(H)>0)\mathbb{I}(H\in\mathcal{M}^c)\bigg]\\
      &-E\bigg[\delta_R(H)\mathbb{I}\bigg(\lambda_1\le\frac{\delta_Y(H)}{\delta_R(H)}\le \lambda_2 \bigg)\mathbb{I}(\delta_Y(H)<0)\mathbb{I}(H\in\mathcal{M}^c)\bigg]\\
      =&E\bigg[|\delta_R(H)|\mathbb{I}\bigg(\lambda_1\le\frac{\delta_Y(H)}{\delta_R(H)}\le \lambda_2\bigg)\mathbb{I}(H\in\mathcal{M}^c)\bigg].
    \end{split}
  \end{equation*}
The above two equations imply that 
  \begin{equation*}
    \begin{split}
      &E[\mathcal{V}_{\phi}(f_1^*, H)]-E[\mathcal{V}_{\phi}(f_2^*, H)]\\
      =&(1+\eta_n)E\bigg[\frac{|\delta_Y(H)|}{|\delta_R(H)|}|\delta_R(H)|\mathbb{I}\bigg(\lambda_1\le\frac{\delta_Y(H)}{\delta_R(H)}\le \lambda_2\bigg)\mathbb{I}(H\in\mathcal{M}^c)\bigg]\\
      \le&2\lambda_2E\bigg[|\delta_R(H)|\mathbb{I}\bigg(\lambda_1\le\frac{\delta_Y(H)}{\delta_R(H)}\le \lambda_2\bigg)\mathbb{I}(H\in\mathcal{M}^c)\bigg]\\
      \le&2\lambda_2(\tau_1-\tau_2).
    \end{split}
  \end{equation*}
  The lemma holds since $\lambda_2$ is no larger than $\lambda^*$-value corresponding to $\tau-2\delta$.
\end{proof}

\subsubsection{Approximation error in RKHS}
\label{sec1:2}
In this section, we prove a series of lemmas to quantify the approximation error of $\widehat{f}$, where $\widehat{f}$ denotes the solution of the single-stage empirical problem 
\begin{equation}
  \label{single_stage_empirical}
  \begin{split}
    \argmin_{f\in\mathcal{G}}~&\frac{1}{n}\sum_{i=1}^nY_i\frac{\phi(A_if(H_i))}{p(A_i|H_i)}+\lambda_n\|f\|_{\mathcal{G}}^2\\
    \text{subject to }&\frac{1}{n}\sum_{i=1}^nR_i\frac{\psi(A_if(H_i),\eta_n)}{p(A_i|H_i)}\le\tau,
  \end{split}
\end{equation}
resulted from restricting $\widehat{f}$ to be a function in Gaussian RKHS $\mathcal{G}$. 

The section is organized as follows: \cref{lemma:5} provides an approximation of $f^*_{\tau}$ using functions in Gaussian RKHS; in \cref{lemma:6}, we quantify the difference of risk under shifted ramp loss between $f^*_{\tau}$ and its approximation in $\mathcal{G}$; in \cref{lemma:7}, we show that $\|\widehat{f}\|_{\mathcal{G}}$ is bounded with high probability; in \cref{lemma:8}, we show that $\mathcal{A}_n(\tau)$ defined later changes continuously w.r.t. $\tau$ and the approximation error is given \cref{lemma:9}.

For convenience, we define 
$$\mathfrak{L}_{\phi}(f)=Y\frac{\phi(Af(H))}{p(A|H)},\quad \mathbb{P}_n[\mathfrak{L}_{\phi}(f)]=\frac{1}{n}\sum_{i=1}^nY_i\frac{\phi(A_if(H_i))}{p(A_i|H_i)},$$
$$\mathfrak{R}_{\psi}(f,\eta_n)=R\frac{\psi(Af(H),\eta_n)}{p(A|H)},\quad \mathbb{P}_n[\mathfrak{R}_{\psi}(f,\eta_n)]=\frac{1}{n}\sum_{i=1}^nR_i\frac{\psi(A_if(H_i),\eta_n)}{p(A_i|H_i)},$$
where $\mathbb{P}_n$ denotes the empirical distribution. Recall $\mathcal{G}=\mathcal{G}(\sigma_n)$ denote the Gaussian Reproducing Kernel Hilbert Space (RKHS) with bandwidth $\sigma_n^{-1}$, we define
$$\mathcal{A}(\tau)=\bigg\{f\in\mathcal{G}\bigg|E[\mathfrak{R}_{\psi}(f,\eta_n)]\le\tau\bigg\},$$
$$\mathcal{A}_{n}(\tau)=\bigg\{f\in\mathcal{G}\bigg|\mathbb{P}_n[\mathfrak{R}_{\psi}(f,\eta_n)]\le\tau\bigg\},$$
where $\mathcal{A}_{n}(\tau)$ is equivalent to the definition of the feasible region of the empirical problem with $T=1$.
We also define $\bar{\mathcal{H}}=3\mathcal{H}$,
$$
\bar{\delta}_Y(h)=
\begin{cases}
    \delta_Y(h),&\text{if}~|h|\le1\\
    \delta_Y(h/|h|),&\text{if}~|h|>1,
\end{cases}\quad
\bar{\delta}_R(h)=
\begin{cases}
    \delta_R(h),&\text{if}~|h|\le1\\
    \delta_R(h/|h|),&\text{if}~|h|>1.
\end{cases}
$$
Following the notation in \cref{GNE} and omitting the index $t$ when $T=1$, we define
$$\bar{H}_{a,b,\tau}=\bigg\{h\in\bar{\mathcal{H}}:a\bar{\delta}_Y(h)>0,b(\bar{\delta}_Y(h)-\lambda^*\bar{\delta}_R(h))>0\bigg\}$$
so $\Delta_{\tau}(h)=\sum_{a,b\in \{-1,1\}}\text{dist}(h,\bar{\mathcal{H}}/\bar{H}_{a,b,\tau})\mathbb{I}(h\in \bar{H}_{a,b,\tau}),$
where $a,b\in\{-1,1\}$ and $\lambda^*$ is the multiplier associated witn $f^*_{\tau}$ so function $\Delta$ depends on $\tau$, and 
$$
\bar{f}_{\tau}(h)=
\begin{cases}
    1,&\text{if}~h\in\bar{H}_{1,1}\\
    \eta_n,&\text{if}~h\in\bar{H}_{-1,1}\\
    -1,&\text{if}~h\in\bar{H}_{-1,-1}\\
    -\eta_n,&\text{if}~h\in\bar{H}_{1,-1}\\
    0,&\text{otherwise}.
\end{cases}
$$
Thus, $\bar{f}_{\tau}$ can be viewed as an extension of $f_{\tau}^*$ from 
of ${\cal H}$ to $\bar{\mathcal{H}}$. Our first lemma is to determine the pointwise approximation error of $f_{\tau}^*$ using the RKHS. Note that we assumed $\mathcal{H}$ is a compact subset of $\mathcal{G}$, without loss of generality, from now on we assume that $\mathcal{H}\subseteq\mathcal{B}_{\mathcal{G}}$ where $\mathcal{B}_{\mathcal{G}}$ denotes the unit ball in $\mathcal{G}$. We use $d$ to denote the dimension of $\mathcal{H}$.

\begin{lemma}
  \label[lemma]{lemma:5}
  Let $\check{f}_{\tau}=(\sigma_n^2/\pi)^{d/4}\bar{f}_{\tau}$ and define linear operator 
  $$V_{\sigma}\check{f}(x)=\frac{(2\sigma)^{d/2}}{\pi^{d/4}}\int_{\mathbb{R}^d}e^{-2\sigma^2\|x-y\|_2^2}\check{f}(y)dy.$$
  Then, we have 
  \begin{equation}
    \label{boundness_vg}
    \|V_{\sigma_n}\check{f}_{\tau}\|_{\mathcal{G}}^2\le c\sigma_n^d,
  \end{equation}
  and
  \begin{equation}
    \label{projection_error}
    |V_{\sigma_n}\check{f}_{\tau}(h)-f_{\tau}^*(h)|\le 8e^{-\sigma_n^2\Delta_{\tau}(h)^2/2d}.
  \end{equation}
  holds for all $h\in\mathcal{H}$, where $c$ is a constant depending on dimension $d$.
\end{lemma}
\begin{remark}
  Note that $V_{\sigma}\check{f}_{\tau}$ is an approximation of $f_{\tau}^*$ in $\mathcal{G}$. Thus, \cref{lemma:5} quantifies the distance between the true optimal decision function and its approximation at each point $h$.
\end{remark}
\begin{proof}
  Since $\mathcal{H}\subset\mathcal{B}_{\mathcal{G}}$ and $\check{f}_{\tau}=(\sigma_n^2/\pi)^{d/4}\bar{f}_{\tau}$, we can easily obtain that the $L_2$ norm of $\check{f}_{\tau}$ satisfies
  $$\|\check{f}_{\tau}\|^2_{2}\le \text{Vol}(d)^2\bigg(\frac{81}{\pi}\bigg)^{d/2}\sigma_n^{d}=c\sigma_n^{d},$$
  where $\text{Vol}(d)$ is the volume of $\mathcal{B}_{\mathcal{G}}$ (see equation (25) from \citet{steinwart_fast_2007}) so $c$ is a positive constant depends only on $d$. Moreover, it has been shown in \citet{steinwart_explicit_2006} that $V_{\sigma}:L^2(\mathbb{R}^d)\rightarrow\mathcal{G}(\sigma)$ is an isometric isomorphism and the inequality above implies 
  $$\|V_{\sigma_n}\check{f}_{\tau}\|_{\mathcal{G}}^2=\|\check{f}_{\tau}\|_{2}^2\le c\sigma_n^d.$$
  \par 
  We now start proving (\ref{projection_error}). By the construction of $\bar{f}_{\tau}$, it is straightforward to see that $\bar{f}_{\tau}(h)=f_{\tau}^*(h)$ for all $h\in\mathcal{H}$. Note for any $h\in H_{1,1}$ we have 
  {\footnotesize
  \begin{equation*}
    \begin{split}
      V_{\sigma_n}\check{f}_{\tau}(h)=&\bigg(\frac{2\sigma_n^2}{\pi}\bigg)^{d/2}\int_{\mathbb{R}^d}e^{-2\sigma_n^2\|h-y\|^2}\bar{f}_{\tau}(y)dy\\
      =&\bigg(\frac{2\sigma_n^2}{\pi}\bigg)^{d/2}\bigg[\int_{B(h,\Delta_{\tau}(h))}e^{-2\sigma_n^2\|h-y\|^2}\bar{f}_{\tau}(y)dy+\int_{\mathbb{R}^d/B(h,\Delta_{\tau}(h))}e^{-2\sigma_n^2\|h-y\|^2}\bar{f}_{\tau}(y)dy\bigg],
    \end{split}
  \end{equation*}
  }where $B(h,r)$ is the ball of radius $r$ centering at $h$ under Euclidean norm. By Lemma 4.1 in \citet{steinwart_fast_2007}, the construction of $\bar{f}_{\tau}$ guarantees that $B(h,\Delta_{\tau}(h))\subseteq \bar{H}_{1,1}$ for all $h\in H_{1,1}$. It then follows that for any $h\in H_{1,1}$
  \begin{align*}
      |V_{\sigma_n}\check{f}_{\tau}-f_{\tau}^*(h)|
      =&|V_{\sigma_n}\check{f}_{\tau}(h)-\bar{f}_{\tau}(h)|\\
      =&\bigg|V_{\sigma_n}\check{f}_{\tau}(h)-\bigg(\frac{2\sigma_n^2}{\pi}\bigg)^{d/2}\int_{\mathbb{R}^d}e^{-2\sigma_n^2\|h-y\|^2}dy\bigg|\\
      =&\bigg|\bigg(\frac{2\sigma_n^2}{\pi}\bigg)^{d/2}\int_{\mathbb{R}^d/B(h,\Delta_{\tau}(h))}e^{-2\sigma_n^2\|h-y\|^2}[\bar{f}_{\tau}(y)-1]dy\bigg|\\
      {\le}&2P(|U|\geq\Delta_{\tau}(h)),
  \end{align*}
  where the last step uses the fact that $|\bar{f}_{\tau}-1|_{\infty}\le2$, and $U$ follows the spherical Gaussian distribution on $\mathbb{R}^d$ with parameter $\sigma_n$. Following inequality (3.5) from \citet{ledoux_probability_1991}, we have 
  $$P(|U|\geq\Delta_{\tau}(h))\le 4e^{-\sigma_n^2\Delta_{\tau}^2(h)/2d}.$$
  Similarly, we can obtain the same bound for $h\in \bar H_{-1,1}, \bar H_{1,-1}$ and $H_{-1,-1}$. As a conclusion, we have
  $$|V_{\sigma_n}\check{f}_{\tau}(h)-f_{\tau}^*(h)|\le 8e^{-\sigma_n^2\Delta_{\tau}^2(h)/2d}$$
  for any $h\in\mathcal{H}$.
\end{proof}

In the next lemma, we show that under \cref{GNE}, the difference of the risk under shifted ramp loss between $f^*_{\tau'}$ and its approximation $V_{\sigma_n}\check{f}_{\tau'}$ is uniformly bounded by $O(\sigma_n^{-\alpha d}\eta_n^{-1})$ for any $\tau'\in[\tau-2\delta_0,\tau+2\delta_0]$; moreover when $n$ is sufficiently large, $V_{\sigma_n}\check{f}_{\tau-2\delta_0}$ will belong to the empirical feasible region $\mathcal{A}_n(\tau)$ with high probability when $c\sigma_n^{-\alpha d}\eta_n^{-1}\le \delta_0$. 

\begin{lemma}
  \label[lemma]{lemma:6}
  For any $\tau'\in[\tau-2\delta_0,\tau+2\delta_0]$, 
  \begin{equation}
    \label{find_f_1}
    |E[\mathfrak{R}_{\psi}(V_{\sigma_n}\check{f}_{\tau'},\eta_n)]-E[\mathfrak{R}_{\psi}(f_{\tau'}^*,\eta_n)]|\le c\sigma_n^{-\alpha d}\eta_{n}^{-1},
  \end{equation}
  where $c$ is a constant depending on $(\alpha,d,K,M)$. Moreover, for any $\sigma_n$ and $\eta_n$ such that $c\sigma_n^{-\alpha d}\eta_n^{-1}\le \delta_0$, with probability $1-2\exp\big(\frac{-2n\delta^2_0c_1^2}{M^2}\big)$, we have $V_{\sigma_n}\check{f}_{\tau-2\delta_0}\in\mathcal{A}_{n}(\tau)$.
\end{lemma}
\begin{proof}
  First note that for any measurable function $f_1,f_2:\mathcal{H}\rightarrow\mathbb{R}$, we always have 
  \begin{align*}
    &E[\mathfrak{R}_{\psi}(f_1,\eta_n)]-E[\mathfrak{R}_{\psi}(f_2,\eta_n)]\\
    =&E\bigg[E[R|H,A=1][\psi(f_1(H),\eta_n)-\psi(f_2(H),\eta_n)]\\
    &+E[R|H,A=-1][\psi(-f_1(H),\eta_n)-\psi(-f_2(H),\eta_n)]\bigg]\\
    \le&2M\eta_n^{-1}E[|f_1(H)-f_2(H)|].
  \end{align*}
  Using result (\ref{projection_error}) in \cref{lemma:5} and \cref{GNE}, we can obtain
  \begin{align*}
    |E[\mathfrak{R}_{\psi}(V_{\sigma_n}\check{f}_{\tau'},\eta_n)]-E[\mathfrak{R}_{\psi}(f_{\tau'}^*,\eta_n)]|\le&\eta_n^{-1}16ME[e^{-\sigma_n^2\Delta_{\tau'}(H)^2/2d}]\\
    \le&16MK(2d)^{\alpha d/2}\sigma_n^{-\alpha d}\eta_n^{-1}\\
    =&c\sigma_n^{-\alpha d}\eta_n^{-1}.
  \end{align*}
  To prove the remaining part of the lemma, we now let $\tau'=\tau-2\delta_0$. We note that $\mathfrak{R}_{\psi}(V_{\sigma_n}\check{f}_{\tau'},\eta_n)$ is bounded by $M/c_1$. Based on Hoeffding's inequality, we can obtain 
  \begin{equation}
    \label{find_f_2}
    P\bigg[|\mathbb{P}_n[\mathfrak{R}_{\psi}(V_{\sigma_n}\check{f}_{\tau'}, \eta_n)]-E[\mathfrak{R}_{\psi}(V_{\sigma_n}\check{f}_{\tau'}, \eta_n)]|\geq\delta_0\bigg]\le 2\exp\bigg(\frac{-2n\delta_0^2c_1^2}{M^2}\bigg).
  \end{equation}
  According to (\ref{find_f_1}) and under the choice of $(\sigma_n,\eta_n)$, we have
  \begin{equation}
    \label{find_f_3}
    \begin{split}
      &|E[\mathfrak{R}_{\psi}(V_{\sigma_n}\check{f}_{\tau'},\eta_n)]-(\tau-2\delta_0)|\\
      =&|E[\mathfrak{R}_{\psi}(V_{\sigma_n}\check{f}_{\tau'},\eta_n)]-E[\mathfrak{R}_{\psi}(f^*_{\tau},\eta_n)]|\\
      \le&c\sigma_n^{-\alpha d}\eta_n^{-1}\le\delta_0.
    \end{split}
  \end{equation}
  Combining (\ref{find_f_2}) and (\ref{find_f_3}), we obtain
  $$P\bigg[\mathbb{P}_n[\mathfrak{R}_{\psi}(V_{\sigma_n}\check{f}_{\tau'}, \eta_n)]\geq \tau\bigg]\le 2\exp\bigg(\frac{-2n\delta_0^2c_1^2}{M^2}\bigg),$$
  which implies that $V_{\sigma_n}\check{f}_{\tau'}=V_{\sigma_n}\check{f}_{\tau-2\delta_0}\in\mathcal{A}_{n}(\tau)$ with probability at least $1-2\exp\big(\frac{-2n\delta_0^2c_1^2}{M^2}\big)$.
\end{proof}

In \cref{lemma:7}, we show that $\|\widehat{f}\|_{\mathcal{G}}$ is bounded with high probability.

\begin{lemma}
  \label[lemma]{lemma:7}
  $\widehat{f}_{\tau}$ satisfies
  \begin{equation}
    \label{bound_f_hat}
    P\bigg(\|\widehat{f}_{\tau}\|_{\mathcal{G}}^2\le c\bigg(\frac{M}{c_1\lambda_n}+\sigma_n^d\bigg)\bigg)\geq1-2\exp\bigg(\frac{-2n\delta_0^2c_1^2}{M^2}\bigg),
  \end{equation}
  for any choice of $c\sigma_n^{-\alpha d}\eta_{n}^{-1}\le \delta_0$. Here, the constant $c$ in front of $\sigma_n^d$ only depends on dimension $d$ and the constant in front of $\sigma_n^{-\alpha d}\eta_n^{-1}$ is equal to the constant of the same term in \cref{lemma:6}.
\end{lemma}
\begin{proof}
  From the last claim of \cref{lemma:6}, we have $V_{\sigma_n}\check{f}_{\tau-2\delta_0}\in \mathcal{A}_{n}(\tau)$ holds with probability at least $1-2\exp\big(\frac{-2n\delta_0^2c_1^2}{M^2}\big)$. Using and (\ref{boundness_vg}) of \cref{lemma:5}, under the choice of $(\sigma_n,\eta_n)$ we have
  \begin{align*}
      \lambda_n\|\widehat{f}\|_{\mathcal{G}}^2&\le \mathbb{P}_n[\mathfrak{L}_{\phi}(\widehat{f})]+\lambda_n\|\widehat{f}\|_{\mathcal{G}}^2\le \mathbb{P}_n[\mathfrak{L}_{\phi}(V_{\sigma_n}\check{f}_{\tau-2\delta_0})]+\lambda_n\|V_{\sigma_n}\check{f}_{\tau-2\delta_0}\|_{\mathcal{G}}^2
      \le c\bigg(\frac{M}{c_1}+\lambda_n\sigma_n^d\bigg),
  \end{align*}
  which gives (\ref{bound_f_hat}). 
\end{proof}

\cref{lemma:7} implies that, instead of $\mathcal{A}(\tau)$ and $\mathcal{A}_{n}(\tau)$, we can concentrate on the sets given by
$$\mathcal{A}(\tau,\mathcal{C}_n)=\bigg\{f\in\mathcal{G}\bigg|\|f\|_{\mathcal{G}}\le \mathcal{C}_n,E[\mathfrak{R}_{\psi}(f,\eta_{n})]\le\tau\bigg\},$$
$$\mathcal{A}_{n}(\tau,\mathcal{C}_n)=\bigg\{f\in\mathcal{G}\bigg|\|f\|_{\mathcal{G}}\le \mathcal{C}_n,\mathbb{P}_n[\mathfrak{R}_{\psi}(f,\eta_{n})]\le\tau\bigg\}, $$
where $\mathcal{C}_n=c\sqrt{\frac{M}{c_1\lambda_n}+\sigma_n^d}$. This is because $\widehat f$ belongs to them with a high probability. 

We further study the relationships among $\mathcal{A}(\tau,\mathcal{C}_n)$ and $\mathcal{A}_{n}(\tau,\mathcal{C}_n)$. The proof will use a general covering number property for Gaussian RKHS from \citet{steinwart_fast_2007}, which is stated as \cref{prop:1} in Section A.2.1.4.

\begin{lemma}
  \label[lemma]{lemma:8}
  For any $\delta>0$ with probability at least $1-\exp\big(-\frac{n\delta^2c_1^2}{2M^2}\big)$, we have 
  \begin{equation}
    \label{A_inclusion}
    \mathcal{A}(\tau-\epsilon_n,\mathcal{C}_n)\subset \mathcal{A}_{n}(\tau,\mathcal{C}_n)\subset\mathcal{A}(\tau+\epsilon_n,\mathcal{C}_n),
  \end{equation}
  where 
  $$\epsilon_n=c\sigma_n^{(1-\nu_2/2)(1+\theta_2)d/2}\bigg(\frac{M}{c_1\lambda_n}+\sigma_n^d\bigg)^{\nu_2/4}\eta_{n}^{-\nu_2/2}+\delta$$
  for $0<\nu_2\le2$ and $\theta_2>0$. Moreover, let 
  $$\epsilon'_n=\epsilon_n+c\sigma_n^{-\alpha d}\eta_{n}^{-1},$$
  then for any $\lambda_n$ and $\sigma_n$ such that $\epsilon'_n\le2\delta_0$, we have $V_{\sigma_n}\check{f}_{\tau-\epsilon_n'}\in \mathcal{A}(\tau-\epsilon_n,\mathcal{C}_n)$ and
  $$|E[\mathfrak{L}_{\phi}(V_{\sigma_n}\check{f}_{\tau-\epsilon_n'})]-E[\mathfrak{L}_{\phi}(f^*_{\tau-\epsilon_n'})]|\le c\sigma_n^{-\alpha d}.$$
  Here, the constants in front of $\sigma_n^{-ad}$ and $\sigma_n^{-ad}\eta_n^{-1}$ are equal to the constants in \cref{lemma:7}. $c$ in front of $\epsilon_n$ is a constant only dependent on $(M,c_1,\nu_2,\theta_2,d)$.
\end{lemma}
\begin{proof}
  To prove (\ref{A_inclusion}), it is sufficient to show that with probability $1-\exp\big(-\frac{n\delta^2c_1^2}{2M^2}\big)$ we have
  \begin{equation}
    \label{concentration}
    \underset{f\in\mathfrak{R}_{\psi,\eta_{n}}\circ\mathcal{B}_{\mathcal{G}}(\mathcal{C}_n)}{\sup}|\mathbb{P}_n[f]-E[f]|\le \epsilon_n,
  \end{equation}
  where $\mathfrak{R}_{\psi,\eta_{n}}\circ\mathcal{B}_{\mathcal{G}}(\mathcal{C}_n)=\{\mathfrak{R}_{\psi}(f,\eta_{n})|f\in\mathcal{B}_{\mathcal{G}}(\mathcal{C}_n)\}$ and $\mathcal{B}_{\mathcal{G}}(\mathcal{C}_n)$ denotes the closed ball in $\mathcal{G}$ with radius $\mathcal{C}_n$. By Theorem 4.10 from \citet{wainwright_high-dimensional_2019}, we have that
  $$\underset{f\in\mathfrak{R}_{\psi,\eta_{n}}\circ\mathcal{B}_{\mathcal{G}}(\mathcal{C}_n)}{\sup}|\mathbb{P}_n[f]-E[f]|\le 2\text{Rad}_n(\mathfrak{R}_{\psi,\eta_{n}}\circ\mathcal{B}_{\mathcal{G}}(\mathcal{C}_n))+\delta$$
  holds with probability $1-\exp\big(-\frac{n\delta^2c_1^2}{2M^2}\big)$, where $\text{Rad}_n(\mathcal{F})$ is the Rademacher complexity of some functional set $\mathcal{F}$ defined as
  $$\text{Rad}_n(\mathcal{F})=E_{X}\bigg[E_{\epsilon}\bigg[\underset{f\in\mathcal{F}}{\sup}~\bigg|\frac{1}{n}\sum_{i=1}^n\epsilon_if(X_i)\bigg|\bigg]\bigg],\quad\epsilon_i\sim i.i.d.~P(\epsilon_i=\pm1)=0.5.$$
  Following the proof in Example 5.24 from \citet{wainwright_high-dimensional_2019}, by Dudley's entropy integral the Rademacher complexity of $\mathfrak{R}_{\psi,\eta_{n}}\circ\mathcal{B}_{\mathcal{G}}(\mathcal{C}_n)$ is upper bound by 
  \begin{equation}
    \label{rademacher}
    \begin{split}
      \text{Rad}_n(\mathfrak{R}_{\psi,\eta_{n}}\circ\mathcal{B}_{\mathcal{G}}(\mathcal{C}_n))\le &E\bigg[\frac{24}{\sqrt{n}}\int_{0}^{2\frac{M}{c_1}}\sqrt{\log \mathcal{N}(\mathfrak{R}_{\psi,\eta_{n}}\circ\mathcal{B}_{\mathcal{G}}(\mathcal{C}_n),\epsilon, L_2(\mathbb{P}_n))}d\epsilon\bigg]\\
      \overset{(i)}{\le}&E\bigg[\frac{24}{\sqrt{n}}\int_{0}^{2\frac{M}{c_1}}\sqrt{\log \mathcal{N}\bigg(\mathcal{B}_{\mathcal{G}},\frac{\eta_{n} c_1}{M\mathcal{C}_n}\epsilon, L_2(\mathbb{P}_n)\bigg)}d\epsilon\bigg]\\
      \overset{(ii)}{\le}&c\sigma_n^{(1-\nu_2/2)(1+\theta_2)d/2}\bigg(\frac{M}{c_1\lambda_n}+c_d^2\sigma_n^d\bigg)^{\nu_2/4}\eta_{n}^{-\nu_2/2},
    \end{split}
  \end{equation}
  where to obtain $(i)$ we have used the fact that $\mathfrak{R}_{\psi,\eta_{n}}$ is a Lipschitz function of $f$ with Lipschitz constant $\frac{M}{c_1\eta_{n}}$, and in $(ii)$ we used the covering number property of $\mathcal{B}_{\mathcal{G}}$ stated in \cref{prop:1}.
  \par 
  For the second part of the lemma, $V_{\sigma_n}\check{f}_{\tau-\epsilon_n'}\in \mathcal{A}(\tau-\epsilon_n,\mathcal{C}_n)$ is a direct conclusion of (\ref{find_f_1}) from \cref{lemma:6} since 
  \begin{equation*}
    \begin{split}
      &E[\mathfrak{R}_{\psi}(V_{\sigma_n}\check{f}_{\tau-\epsilon_n'},\eta_n)]\\
      \le &|E[\mathfrak{R}_{\psi}(V_{\sigma_n}\check{f}_{\tau-\epsilon_n'},\eta_n)]-E[\mathfrak{R}_{\psi}(f^*_{\tau-\epsilon_n'},\eta_n)]|+|E[\mathfrak{R}_{\psi}(f^*_{\tau-\epsilon_n'},\eta_n)]|\\
      \le&\tau-\epsilon_n'+c\sigma_n^{-\alpha d}\eta_n^{-1}=\tau-\epsilon_n.
    \end{split}
  \end{equation*}
  Note that for any $f_1,f_2:\mathcal{H}\rightarrow[-1,1]$ we always have
  \begin{align*}
      &E[\mathfrak{L}_{\phi}(f_1)]-E[\mathfrak{L}_{\phi}(f_2)]\\
      =&E\bigg[E[Y|H,A=1][\phi(f_1(H))-\phi(f_2(H))]+E[Y|H,A=-1][\phi(-f_1(H))-\phi(-f_2(H))]\bigg]\\
      \le&E[|\delta_Y(H)||f_1(H)-f_2(H)|].
  \end{align*}
  Hence, using (\ref{projection_error}) in \cref{lemma:5} and \cref{GNE} we have
  \begin{align*}
  |E[\mathfrak{L}_{\phi}(V_{\sigma_n}\check{f}_{\tau-\epsilon_n'})]-&E[\mathfrak{L}_{\phi}(f^*_{\tau-\epsilon_n'})]|\le\\
  &8E\big[|\delta_Y(H)|e^{-\sigma_n^2\Delta_{\tau-\epsilon_n'}^2(H)/2d}\big]\le 8MK(2d)^{\alpha d/2}\sigma_n^{-\alpha d}=c\sigma_n^{-\alpha d}.
  \end{align*}
  This completes the proof for the lemma.
\end{proof}

As a corollary of \cref{lemma:8}, we can establish the risk error bound (\ref{risk_error_bound}) stated in \cref{thm:s2} for $T=1$. We state this result as \cref{corollary:1} below 

\begin{corollary}
  \label[corollary]{corollary:1}
  Suppose $(\sigma_n,\eta_n)$ satisfy the requirement in \cref{lemma:7}, then for any $0<\nu_2\le2$, $\theta_2>0$ and $\delta>0$ with probability at least $1-2\exp\big(\frac{-2n\delta_0^2c_1^2}{M^2}\big)-2\exp\big(-\frac{n\delta^2c_1^2}{2M^2}\big)$ we have
  $$E\bigg[\frac{R\mathbb{I}(A\widehat{f}(H_t)>0)}{p(A|H)}\bigg]\le\tau+\delta+cn^{-1/2}\sigma_n^{(1-\nu_2/2)(1+\theta_2)d/2}\bigg(\frac{M}{c_1\lambda_n}+\sigma_n^{d_t}\bigg)^{\nu_2/4}\eta_n^{-\nu2/2}.$$
  Here, $c$ is a constant only depends on $(M,c_1,\nu_2,\theta_2,d)$.
\end{corollary}

\begin{proof}
  \cref{lemma:7} implies that $\widehat{f}$ is bounded by $\mathcal{C}_n$ with probability at least $1-2\exp\big(\frac{-2n\delta_0^2c_1^2}{M^2}\big)$. Moreover, the concentration inequality (\ref{concentration}) of \cref{lemma:8} implies that 
  $$E[\mathfrak{R}_{\psi}(f,\eta_n)]-\mathbb{P}_n[\mathfrak{R}_{\psi}(f,\eta_n)]\le\delta+cn^{-1/2}\sigma_n^{(1-\nu_2/2)(1+\theta_2)d/2}\bigg(\frac{M}{c_1\lambda_n}+\sigma_n^{d_t}\bigg)^{\nu_2/4}\eta_n^{-\nu2/2}$$
  holds with probability at least $1-2\exp\big(-\frac{n\delta^2c_1^2}{2M^2}\big)$ for any $\delta>0$ and $f\in\mathcal{B}_{\mathcal{G}}(\mathcal{C}_n)$. The result holds since $\mathbb{P}_n[\mathfrak{R}_{\psi}(\widehat{f},\eta_n)]\le\tau$ the choice of $\widehat{F}$ and note that 
  $$E\bigg[\frac{R\mathbb{I}(A\widehat{f}(H)>0)}{p(A|H)}\bigg]\le E[\mathfrak{R}_{\psi}(\widehat{f},\eta_n)].$$
\end{proof}

\cref{lemma:8} indicates that 
$$V_{\sigma_n}\check{f}_{\tau-\epsilon_n'}\in\mathcal{A}(\tau-\epsilon_n,\mathcal{C}_n)\subseteq\mathcal{A}_n(\tau,\mathcal{C}_n)$$
with high probability. In \cref{lemma:9}, we will show that $V_{\sigma_n}\check{f}_{\tau-\epsilon_n'}$ can be used to quantify the approximation error caused by RKHS.

\begin{lemma}
  \label[lemma]{lemma:9}
  Under the condition of \cref{lemma:8}, we have 
  $$\underset{f\in\mathcal{A}(\tau-\epsilon_n,\mathcal{C}_n)}{\inf}(E[\mathfrak{L}_{\phi}(f)] +\lambda_n\|f\|_{\mathcal{G}}^2-E[\mathfrak{L}_{\phi}(f^*_{\tau-\epsilon_n'})])\le\xi_{n}^{(2)}.$$
\end{lemma}

\begin{proof}
  Let $\check{f}_{\tau-\epsilon_n'}=(\sigma_n^2/\pi)^{d/4}\bar{f}_{\tau-\epsilon_n'}$, then from \cref{lemma:8} we have 
  $$|E[\mathfrak{L}_{\phi}(V_{\sigma_n}\check{f}_{\tau-\epsilon_n'})]-E[\mathfrak{L}_{\phi}(f^*_{\tau-\epsilon_n'})]|\le c\sigma_n^{-\alpha d},$$
  and $V_{\sigma_n}\check{f}_{\tau-\epsilon_n'}\in\mathcal{A}(\tau-\epsilon_n,\mathcal{C}_n)$. Moreover, (\ref{boundness_vg}) from \cref{lemma:5} gives that $\|V_{\sigma_n}\check{f}_{\tau-\epsilon_n'}\|^2_{\mathcal{G}}\le c\sigma_n^d$. Hence, we have 
  \begin{align*}
      &\underset{f\in\mathcal{A}(\tau-\epsilon_n,\mathcal{C}_n)}{\inf}\bigg[E[\mathfrak{L}_{\phi}(f)] +\lambda_n\|f\|^2_{\mathcal{G}}-E[\mathfrak{L}_{\phi}(f^*_{\tau-\epsilon_n'})]\bigg]\\
      =&\underset{f\in\mathcal{A}(\tau-\epsilon_n,\mathcal{C}_n)}{\inf}\bigg[E[\mathfrak{L}_{\phi}(f)]+\lambda_n\|f\|^2_{\mathcal{G}}-E[\mathfrak{L}_{\phi}(V_{\sigma_n}\check{f}_{\tau-\epsilon_n'})]-\lambda_n\|V_{\sigma_n}\check{f}_{\tau-\epsilon_n'}\|_{\mathcal{G}}^2\bigg]+\lambda_n\|V_{\sigma_n}\check{f}_{\tau-\epsilon_n'}\|^2_{\mathcal{G}}\\
      &+E[\mathfrak{L}_{\phi}(V_{\sigma_n}\check{f}_{\tau-\epsilon_n'})]-E[\mathfrak{L}_{\phi}(f^*_{\tau-\epsilon_n'})]\\
      \le&c(\lambda_n\sigma_n^d+\sigma_n^{-\alpha d})\equiv \xi_{n}^{(2)}.
  \end{align*}
\end{proof}

\subsubsection{Completing the proof to \cref{thm:s2} for T=1}
\label{sec1:3}
\noindent
We first establish the error bound for the regret (\ref{value_error_bound}). Since the Fisher consistency of \cref{thm:s1} indicates $\mathcal{V}(g^*)=\mathcal{V}(f^*_{\tau})$ and using the excessive risk inequality in \cref{lemma:3} we have 
\begin{equation*}
  \label{excessive_risk}
  \mathcal{V}(f^*_{\tau})-\mathcal{V}(\widehat{f})\le E[\mathcal{V}_{\phi}(f^*_{\tau}, H)]-E[\mathcal{V}_{\phi}(\widehat{f}, H)]+M\eta_{n},
\end{equation*}
the proof is completed if we can show
\begin{equation}
  \label{need_to_proof}
  E[\mathcal{V}_{\phi}(f^*_{\tau}, H)]-E[\mathcal{V}_{\phi}(\widehat{f}, H)]=E[\mathfrak{L}_{\phi}(\widehat{f})]-E[\mathfrak{L}_{\phi}(f^*_{\tau})]\le \xi_{n}+2\lambda_0\epsilon'_{n}=\xi_{n}+c\epsilon'_{n}
\end{equation}
holds with probability at least $1-h(n,x)$, where $\lambda_0$ denotes the $\lambda^*$-value for $(\tau-2\delta_0)$ which is a constant only depends on $(\tau,\delta_0)$.

According to \cref{lemma:7}, we have shown that $\|\widehat{f}\|_{\mathcal{G}}$ is bounded by $\mathcal{C}_n=c\sqrt{\frac{M}{c_1\lambda_n}+\sigma_n^d}$ with probability at least $1-2\exp\big(\frac{-2n\delta_{0}^2c_1^2}{M^2}\big)$. Hence, similar to proof of \cref{corollary:1}, we can restrict to set $\mathcal{B}_{\mathcal{G}}(\mathcal{C}_n)$, and replace $\mathcal{A}(\tau)$ and $\mathcal{A}_n(\tau)$ by $\mathcal{A}(\tau,\mathcal{C}_n)$ and $\mathcal{A}_{n}(\tau,\mathcal{C}_n)$ correspondingly with high probability.

To prove (\ref{need_to_proof}), we note that the left-hand side of the inequality can be composed as
\begin{equation}
  \label{decomposition}
  \begin{split}
      &E[\mathfrak{L}_{\phi}(\widehat{f})]-E[\mathfrak{L}_{\phi}(f^*_{\tau})]\\
      \le&E[\mathfrak{L}_{\phi}(\widehat{f})] + \lambda_n\|\widehat{f}\|^2_{\mathcal{G}}-E[\mathfrak{L}_{\phi}(f^*_{\tau})]\\
      \le&E[\mathfrak{L}_{\phi}(\widehat{f})] + \lambda_n\|\widehat{f}\|^2_{\mathcal{G}} - \underset{f\in\mathcal{A}_{n}(\tau,\mathcal{C}_n)}{\inf}(E[\mathfrak{L}_{\phi}(f)]+\lambda_n\|f\|^2_{\mathcal{G}}) \\
      &+\underset{f\in\mathcal{A}_{n}(\tau,\mathcal{C}_n)}{\inf}(E[\mathfrak{L}_{\phi}(f)] + \lambda_n\|f\|^2_{\mathcal{G}})-E[\mathfrak{L}_{\phi}(f^*_{\tau})]\\
      \le&\underbrace{E[\mathfrak{L}_{\phi}(\widehat{f})] + \lambda_n\|\widehat{f}\|^2_{\mathcal{G}} - \underset{f\in\mathcal{A}_{n}(\tau,\mathcal{C}_n)}{\inf}(E[\mathfrak{L}_{\phi}(f)] + \lambda_n\|f\|^2_{\mathcal{G}})}_{(I)}\\
      &+\underbrace{\underset{f\in\mathcal{A}_{n}(\tau,\mathcal{C}_n)}{\inf}(E[\mathfrak{L}_{\phi}(f)] +\lambda_n\|f\|_{\mathcal{G}}^2)-\underset{f\in\mathcal{A}(\tau-\epsilon_n,\mathcal{C}_n)}{\inf}(E[\mathfrak{L}_{\phi}(f)] +\lambda_n\|f\|^2_{\mathcal{G}})}_{(II)}\\
      &+\underbrace{\underset{f\in\mathcal{A}(\tau-\epsilon_n,\mathcal{C}_n)}{\inf}(E[\mathfrak{L}_{\phi}(f)] +\lambda_n\|f\|^2_{\mathcal{G}})-E[\mathcal{L}_{\phi}(f^*_{\tau-\epsilon_n'})]}_{(III)}+\underbrace{E[\mathfrak{L}_{\phi}(f^*_{\tau-\epsilon_n'})]-E[\mathfrak{L}_{\phi}(f^*_{\tau})]}_{(IV)}.
  \end{split}
\end{equation}

Using the inclusion result from \cref{lemma:8}, we have 
$$\mathcal{A}(\tau-\epsilon_n,\mathcal{C}_n)\subseteq\mathcal{A}_{n}(\tau,\mathcal{C}_n)\subseteq\mathcal{A}(\tau+\epsilon_n,\mathcal{C}_n)$$
holds with probability no more than $2\exp\big(-\frac{n\delta_0^2c_1^2}{2M^2}\big)$, and $\mathcal{A}(\tau-\epsilon_n,\mathcal{C}_n)\subseteq\mathcal{A}_{n}(\tau,\mathcal{C}_n)$ implies $(II)<0$. Using the approximation error bound obtained in \cref{lemma:9}, $(III)$ is bounded by $\xi_{n}^{(2)}$. For term $(IV)$, using the Lipschitz continuity property of the value function obtained in \cref{lemma:4} we have
$$|E[\mathfrak{L}_{\phi}(f^*_{\tau-\epsilon_n'})]-E[\mathfrak{L}_{\phi}(f^*_{\tau})]|\le 2\lambda_0\epsilon_n'.$$
In addition, $\mathcal{A}_{n}(\tau,\mathcal{C}_n)\subseteq\mathcal{A}(\tau+\epsilon_n,\mathcal{C}_n)$ and \cref{lemma:4} implies that
\begin{align*}
  &E[\mathfrak{L}_{\phi}(\widehat{f})]-E[\mathfrak{L}_{\phi}(f^*_{\tau})]\\
  \geq&E[\mathfrak{L}_{\phi}(\widehat{f})]-E[\mathfrak{L}_{\phi}(f^*_{\tau+\epsilon_n'})]+E[\mathfrak{L}_{\phi}(f^*_{\tau+\epsilon_n'})]-E[\mathfrak{L}_{\phi}(f^*_{\tau})]\\
  \geq&E[\mathfrak{L}_{\phi}(f^*_{\tau+\epsilon_n'})]-E[\mathfrak{L}_{\phi}(f^*_{\tau})]=O(\epsilon_n').
\end{align*}
which provides a lower bound for the difference of the surrogate reward between $\widehat{f}$ and $f^*_{\tau}$. 

Hence, it remains to derive an upper bound for $(I)$. Let 
$$\mathcal{W}=\{\mathfrak{L}_{\phi}(f_1)+\lambda_n\|f_1\|^2_{\mathcal{G}}-\mathfrak{L}_{\phi}(f_2)-\lambda_n\|f_2\|^2_{\mathcal{G}}:f_1\in\mathcal{B}_{\mathcal{G}}(\mathcal{C}_n), f_2\in\mathcal{B}_{\mathcal{G}}(\mathcal{C}_n)\},$$
and note that $\mathcal{A}_n(\tau,\mathcal{C}_n)\subset\mathcal{B}_{\mathcal{G}}(\mathcal{C}_n)$, then following a similar argument as the proof of Theorem 5.3 of \citet{steinwart_fast_2007}, it is sufficient to control the probability of 
$$P(\exists w\in\mathcal{W}\text{~with~}\mathbb{P}_n(w)\leq0\text{~and~}E(w)\geq\xi_n^{1}).$$
To achieve this, we aim at applying \cref{prop:3} to $\mathcal{W}$, which requiring that for any $w\in\mathcal{W}$, one has $\|w\|_{\infty}\leq B$ for some constant $B>0$ and the $\epsilon$-covering number of $\mathcal{N}(B^{-1}\mathcal{W},\epsilon,L_2(\mathbb{P}_n))$ satisfies that 
$$\underset{\mathbb{P}_n}{\sup}\log\mathcal{N}(B^{-1}\mathcal{W},\epsilon,L_2(\mathbb{P}_n))\leq l\epsilon^{-p}$$
for some $l>0$ and $p>0$.

To verify the first condition, we first note that by the definition of $\mathcal{W}$, it is sufficient to verify the condition for
$$\mathcal{W}'=\{\mathfrak{L}_{\phi}(f)+\lambda_n\|f\|^2_{\mathcal{G}}:f\in\mathcal{B}_{\mathcal{G}}\}(\mathcal{C}_n),$$
which will then yield the expected conditions for $\mathcal{W}$ up to a constant. For any $f\in\mathcal{B}_{\mathcal{G}}(\mathcal{C}_n)$, we first note that 
\begin{align*}
  \|\mathfrak{L}_{\phi}(f)+\lambda_n\|f\|^2_{\mathcal{G}}\|_{\infty}&\le\frac{M}{c_1}\|f\|_{\infty}+\lambda_n\|\widetilde{f}\|^2_{\mathcal{G}}\\&\le \frac{cM}{c_1}\sqrt{\frac{M}{c_1\lambda_n}+\sigma_n^d}+c^2\lambda_n\bigg(\frac{M}{c_1\lambda_n}+\sigma_n^d\bigg)=B,
\end{align*}
since $\|f\|_{\infty}\le\|f\|_{\mathcal{G}}$ for the Gaussian RKHS, which indicates that $\|w\|_{\infty}\leq B$ for any $w\in\mathcal{W}'$. Furthermore, by the sub-additivity of the entropy we have
\begin{align*}
    \log\mathcal{N}(B^{-1}\mathcal{W'},2\epsilon,L_2(\mathbb{P}_n))\le&\underbrace{\log\mathcal{N}(B^{-1}\{\mathfrak{L}_{\phi}(f):f\in\mathcal{B}_{\mathcal{G}}(\mathcal{C}_n)\},\epsilon,L_2(\mathbb{P}_n))}_{(V)}\\
    &+\underbrace{\log\mathcal{N}(B^{-1}\{\lambda_n\|f\|_{\mathcal{G}}^2:f\in\mathcal{B}_{\mathcal{G}}(\mathcal{C}_n)\},\epsilon,L_2(\mathbb{P}_n))}_{(VI)}.
\end{align*}
For $(V)$ we have 
\begin{align*}
(V)\le&\log\mathcal{N}\bigg(\mathcal{B}_{\mathcal{G}}(\mathcal{C}_n),\frac{c_1B\epsilon}{M},L_2(\mathbb{P}_n)\bigg)\\
    =&\log\mathcal{N}\bigg(\mathcal{B}_{\mathcal{G}},\frac{c_1B\epsilon}{M}\bigg(c\sqrt{\frac{M}{c_1\lambda_n}+\sigma_n^d}\bigg)^{-1},L_2(\mathbb{P}_n)\big)\\
    \le&\log\mathcal{N}(\mathcal{B}_{\mathcal{G}},c\epsilon,L_2(\mathbb{P}_n)),
\end{align*}
since $\mathfrak{L}_{\phi}$ is a 1-Lipschitz function of $f$ and $\frac{c_1B}{M}\bigg(c\sqrt{\frac{M}{c_1\lambda_n}+\sigma_n^d}\bigg)^{-1}$ converges to a finite positive constant for sufficient small $\lambda_n$ and large $\sigma_n$.
Moreover, for $(VI)$ we have that for $\epsilon>0$, 
$$(VI)\le\log\bigg(c\sqrt{\frac{M}{c_1\lambda_n}+\sigma_n^d}\bigg/(B\epsilon)\bigg)=\log\bigg(c\sqrt{\frac{M}{c_1\lambda_n}+\sigma_n^d}\bigg/B\bigg)-\log\epsilon\le c(-\log\epsilon).$$
Combining the upper bound of $(V)$ and $(VI)$ and using the covering number property for $\log\mathcal{N}(\mathcal{B}_{\mathcal{G}},\epsilon,L_2(\mathbb{P}_n))$ given in the \cref{prop:1},
we have 
\begin{align*}
    \underset{\mathbb{P}_n}{\sup}\log\mathcal{N}(B^{-1}\mathcal{W},\epsilon,L_2(\mathbb{P}_n))\le&\underset{\mathbb{P}_n}{\sup}\log\mathcal{N}(\mathcal{B}_{\mathcal{G}},c\epsilon,L_2(\mathbb{P}_n))-\log\epsilon\\
    \le&c\sigma_n^{(1-\nu_1/2)(1+\theta_1)d}\epsilon^{-\nu_1},
\end{align*}
for any $0<\nu_1<2$, $\theta_1>0$, and some positive constant $c$ which only depends on $(\nu_1,\theta_1,d,M,c_1)$. Therefore, let $B=\frac{2cM}{c_1}\sqrt{\frac{M}{c_1\lambda_n}+\sigma_n^d}+2c^2\lambda_n\bigg(\frac{M}{c_1\lambda_n}+\sigma_n^d\bigg)$, $l=c\sigma_n^{(1-\nu_1/2)(1+\theta_1)d}$ and $p=\nu_1$, the \cref{prop:3} implies that (I) is lower bounded by
\begin{align*}
  &P^*([E[\mathfrak{L}_{\phi}(\widehat{f})] + \lambda_n\|\widehat{f}\|^2_{\mathcal{G}} - \underset{f\in\mathcal{A}_{n}(\tau,\mathcal{C}_n)}{\inf}(E[\mathfrak{L}_{\phi}(f)] + \lambda_n\|f\|^2_{\mathcal{G}})]\geq\xi_{n}^{(1)})\\
  \leq&P(\exists w\in\mathcal{W}\text{~with~}\mathbb{P}_n(w)\leq0\text{~and~}E(w)\geq\xi_n^{1}).\\
  \le&e^{-x},
\end{align*}
where 
$$\xi_{n}^{(1)}=c\bigg(\frac{2M}{c_1}\sqrt{\frac{M}{c_1\lambda_n}+\sigma_n^d}+2\lambda_n\bigg(\frac{M}{c_1\lambda_n}+\sigma_n^d\bigg)\bigg)n^{-1/2}(\sigma_n^{(1-\nu_1/2)(1+\theta_1)d/2}+2\sqrt{2x}+2x/\sqrt{n})$$
for some positive $c$ only depends on $(\nu_1,\theta_1,d)$, which completes the proof of (\ref{value_error_bound}). The risk inequality (\ref{risk_error_bound}) is guaranteed by \cref{corollary:1} and this completes the proof for $T=1$.

\subsubsection{Statement of Propositions}

In this section, we give the statements of all propositions used for establishing \cref{thm:s2}. The first proposition states that the $\epsilon$-covering number of $\mathcal{B}_{\mathcal{G}}$ under $L_2(\mathbb{P}_n)$ is uniformly with polynomial order in terms of $\sigma_n$ and $\epsilon$. This result was first established as Theorem 2.1 in \citet{steinwart_explicit_2006}.

\begin{proposition}{(\citet[Theorem 2.1]{steinwart_fast_2007})}
  \label[proposition]{prop:1}
  For any $\epsilon>0$, we have 
  $$\underset{\mathbb{P}_n}{\sup}\log\mathcal{N}(\mathcal{B}_{\mathcal{G}},\epsilon,L_2(\mathbb{P}_n))\le c\sigma_n^{(1-\nu/2)(1+\theta)d}\epsilon^{-\nu}$$
  for any $0<\nu\le2$ and $\theta>0$. Here, $\mathcal{B}_{\mathcal{G}}$ is the closed unit ball in $\mathcal{G}$ w.r.t. $\|\cdot\|_{\mathcal{G}}$ and $\mathcal{N}(\cdot,\epsilon,L_2(\mathbb{P}_n))$ is the covering number of $\epsilon$-ball w.r.t. empirical $L_2(\mathbb{P}_n)$ norm 
  $$\|f\|_{L_2(\mathbb{P}_n)}=\bigg(\frac{1}{n}\sum_{i=1}^nf(X_i)^2\bigg)^{1/2}.$$
  $c$ is a constant only depends on $(\nu,\theta,d)$.
\end{proposition}

\cref{prop:3} quantifies the stochastic error of $\widehat{f}$, which is a modification of Theorem 5.1 of \citet{steinwart_fast_2007}. Two preliminary propositions used to establish \cref{prop:3} are stated as \cref{prop:4} and \cref{prop:5} at the end of this section. 

\begin{proposition}
  \label[proposition]{prop:3}
  Let P be a probability measure on $\mathcal{Z}$ and $\mathcal{W}$ be a set of bounded measurable functions from $\mathcal{Z}$ to $\mathbb{R}$. Suppose that $\mathcal{W}$ is separable w.r.t. $\|\cdot\|_{\infty}$ and $\|w\|_{\infty}\le B<\infty$ for all $w\in\mathcal{W}$. Let $\mathcal{S}=\{E(w)-w:w\in\mathcal{W}\}$. Then for all $n\geq1$, $h\geq1$ and 
  $$\zeta_n=3E[\underset{s\in\mathcal{S}}{\sup}~\mathbb{P}_n(s)]+2\sqrt{2}B\sqrt{\frac{h}{n}}+2B\frac{h}{n},$$ 
  we have 
  $$P^*(\exists w\in\mathcal{W}\text{~with~}\mathbb{P}_n(w)\le0\text{~and~}E(w)\geq\zeta_n)\leq e^{-h}.$$
  Moreover, when 
  $$\underset{\mathbb{P}_n}{\sup}\log\mathcal{N}(B^{-1}\mathcal{W},\epsilon,L_2(\mathbb{P}_n))\leq l\epsilon^{-p}$$ 
  we have 
  $$E[\underset{s\in\mathcal{S}}{\sup}~\mathbb{P}_n(s)]\geq cB\bigg(\frac{l}{n}\bigg)^{1/2}.$$
  Here, $c$ is a positive constant that only depends on $p$.
\end{proposition}
\begin{proof}
  By the assumption of $\mathcal{W}$ and the definition of $\mathcal{S}$, it is obvious that $E(s)=0$, $\|s\|_{\infty}\le2B$ and $E(s^2)\le 4B^2$ for all $s\in\mathcal{S}$. Moreover, it is also easy to verify that $\mathcal{S}$ is separable w.r.t. $\|\cdot\|_{\infty}$ given $\mathcal{W}$ is separable w.r.t. $\|\cdot\|_{\infty}$. Note that 
  \begin{align*}
      P^*(\exists &w\in\mathcal{W}\text{~with~}\mathbb{P}_n(w)\le0\text{~and~}E(w)\geq\zeta_n)\\
      \le&P^*(\exists w\in\mathcal{W}\text{~with~}E(w)-\mathbb{P}_n(w)\geq\zeta_n)\\
      \le&P^n(\underset{s\in\mathcal{S}}{\sup}~\mathbb{P}_n(s)\geq\zeta_n).
  \end{align*}
  Using Theorem 5.3 from \citet{steinwart_fast_2007}, which is stated as \cref{prop:5}, with $b=2B$ and $\iota=4B^2$, we have 
  \begin{align*}
      P^n(\underset{s\in\mathcal{S}}{\sup}~\mathbb{P}_n(s)\geq\zeta_n)= P^n\bigg(\underset{s\in\mathcal{S}}{\sup}~\mathbb{P}_n(s)\geq 3E[\underset{s\in\mathcal{S}}{\sup}~\mathbb{P}_n(s)]+2\sqrt{2}B\sqrt{\frac{h}{n}}+2B\frac{h}{n}\bigg)\le e^{-h},
  \end{align*}
  which completes the first part of the proposition. To prove the second part of the proposition, we note that by the definition of $\mathcal{S}$, we have
  $$E[\underset{s\in\mathcal{S}}{\sup}~\mathbb{P}_n(s)]\le E\bigg[\underset{w\in\mathcal{W},E(w^2)\le B^2}{\sup}|E(w)-\mathbb{P}_n(w)|\bigg]=\omega_n(\mathcal{W},B^2),$$
  where $\omega_n(\mathcal{W},\xi)$ is the modulus of the continuity of $\mathcal{W}$. Define the local Rademacher complexity of $\mathcal{W}$ to be 
  $$\text{Rad}_n(\mathcal{W},\xi)=E_Z\bigg[E_{\epsilon}\bigg[\underset{w\in\mathcal{W},E(w^2)\le \xi}{\sup}\bigg|\frac{1}{n}\sum_{i=1}^n\epsilon_if(Z_i)\bigg|\bigg]\bigg],$$
  where $\{\epsilon_i\}$ are $n$ i.i.d. Rademacher random variables. According to \citet{van_der_vaart_weak_1996}, we have 
  $$\omega_n(\mathcal{W},\xi)\le 2\text{Rad}_n(\mathcal{W},\xi).$$
  Using the property that for $\forall r>0$
  $$\text{Rad}_n(r\mathcal{W},\xi)=r \text{Rad}_n(\mathcal{W},r^{-2}\xi)$$
  and applying Proposition 5.5 of \citet{steinwart_fast_2007}, which is stated as \cref{prop:4}, under the condition on the covering number of $\mathcal{W}$, we have 
  \begin{align*}E[\underset{s\in\mathcal{S}}{\sup}~\mathbb{P}_n(s)]\le&\omega_n(\mathcal{W},B^2)\le2\text{Rad}_n(\mathcal{W},B^2)
      \le2B \text{Rad}_n(B^{-1}\mathcal{W},1)
      \le2cB\bigg(\frac{l}{n}\bigg)^{1/2},
  \end{align*}
  which completes the second part of the proposition.
\end{proof}

\begin{proposition}{(Steinwart and Scovel (2007, Proposition 5.5))}
  \label[proposition]{prop:4}
  Let $\mathcal{W}$ be a class of measurable functions from $\mathcal{Z}$ to $[-1,1]$ which is separable w.r.t. $\|\cdot\|_{\infty}$ and let $P$ be a probability measure on $\mathcal{Z}$. Assume that there are constants $q>0$ and $0<p<2$ with 
  $$\underset{\mathbb{P}_n}{\sup}\log N(\mathcal{W},\epsilon,L_2(\mathbb{P}_n))\le q\epsilon^{-p}$$
  for all $\epsilon>0$. Then there exists a constant $c$ depending only on $p$ such that for all $n\geq1$ and all $\epsilon>0$ we have 
  $$\text{Rad}_n(\mathcal{W},\epsilon)\le c\max\bigg\{\epsilon^{1/2-p/4}\bigg(\frac{q}{n}\bigg)^{1/2},\bigg(\frac{q}{n}\bigg)^{2/(2+p)}\bigg\}.$$
\end{proposition}

\begin{proposition}{(Steinwart and Scovel (2007, Theorem 5.3))}
  \label[proposition]{prop:5} 
  Let P be a probability measure on $\mathcal{Z}$ and $\mathcal{W}$ be a set of bounded measurable functions from $\mathcal{Z}$ to $\mathbb{R}$ which is separable w.r.t. $\|\cdot\|_{\infty}$ and satisfies $E(w)=0$ for all $w\in\mathcal{W}$. Furthermore, let $b>0$ and $\iota\geq0$ be constants with $\|w\|_{\infty}\le b$ and $E(w^2)\le\iota$ for all $w\in\mathcal{W}$. Then for all $x\geq1$ and all $n\geq1$ we have 
  $$P^n\bigg(\underset{w\in\mathcal{W}}{\sup}~\mathbb{P}_n(w)>3E[\underset{w\in\mathcal{W}}{\sup}~\mathbb{P}_n(w)]+\sqrt{\frac{2x\iota}{n}}+\frac{bx}{n}\bigg)\le e^{-x}.$$ 
\end{proposition}

\subsection{Proof of \cref{thm:s2} for $T\geq2$}

We first prove (\ref{value_error_bound}) is \cref{thm:s2}. To this end, we define 
$$\mathfrak{L}_{\phi,t}(f_t;f_{t+1},...,f_T)=E\bigg[\frac{(Y_t+U_{t+1}(H_{t+1};f_{t+1},...,f_T))}{p(A_t|H_t)}\phi(A_tf_t(H_t))\bigg],$$
and
\begin{equation}
  \label{tilde_v}
  \widetilde{V}_t=\sup_{f_t\in\mathcal{A}_t(\tau_t)}\mathcal{V}_t(f_t,\widehat{f}_{t+1},...,\widehat{f}_{T}),
\end{equation}
where
$$\mathcal{V}_t(g_t,...,g_T)=E\bigg[\frac{(\sum_{j=t}^TY_j)\prod_{j=t}^T\mathbb{I}(A_jg_j(H_j)>0)}{\prod_{j=t}^Tp(A_j|H_j)}\bigg].$$
Note that the Fisher consistency in \cref{fisher_consistency} indicates that $\mathcal{V}_t(g_t^*,...,g_T^*)=\mathcal{V}_t(f_t^*,...,f_T^*)$, and it is sufficient to derive an upper bound for $\mathcal{V}_t(f^*_{t},...,f^*_{T})-\mathcal{V}_t(\widehat{f}_{t},...,\widehat{f}_{T})$. 

First, note that by repeating the same argument as for $T=1$, we can show $\|\widehat{f}_{t}\|_{\mathcal{G}_t}$ is bounded by $\mathcal{C}_{n,t}=c\sqrt{\frac{(T-1+t)M}{c_1\lambda_{n,t}}+\sigma_{n,t}^{d_t}}$ with probability at least $1-2\exp\bigg(-\frac{2n\delta_{0,t}^2c_1^2}{(T-t+1)^2M^2}\bigg)$ for any $t=1,..,T$. Hence, we can replace $\mathcal{A}_t(\tau_t)$ in (\ref{tilde_v}) by $\mathcal{A}_{t}(\tau_t,\mathcal{C}_{n,t})$ with high probability and obtain
\begin{align*}
  &\mathcal{V}_t(f^*_{t},...,f^*_{T})-\mathcal{V}_t(\widehat{f}_{t},...,\widehat{f}_{T})\\
  =&\mathcal{V}_t(f^*_{t},...,f^*_{T})-\widetilde{V}_t+\widetilde{V}_t-\mathcal{V}_t(\widehat{f}_{t},...,\widehat{f}_{T})\\
  \le& \underbrace{\mathcal{V}_t(f^*_{t},...,f^*_{T})-\widetilde{V}_t}_{(I)}+\underbrace{\mathfrak{L}_{\phi,t}(\widehat{f}_{t};\widehat{f}_{t+1},...,\widehat{f}_{T})-\underset{f_t\in\mathcal{A}_t(\tau_t,\mathcal{C}_{n,t})}{\inf}\mathfrak{L}_{\phi,t}(f_t;\widehat{f}_{t+1},...,\widehat{f}_{T})}_{(II)} + (T-t+1)M\eta_{n,t},
\end{align*}
where
$$\mathcal{A}_{t}(\tau_t,\mathcal{C}_{n,t})=\bigg\{f\in\mathcal{G}_t\bigg|\|f\|_{\mathcal{G}_t}\le\mathcal{C}_{n,t},E\bigg[\frac{R_t\psi(A_tf(H_t),\eta_{n,t})}{p(A_t|H_t)}\bigg]\le \tau_t\bigg\},$$
and to obtain the last inequality we have used the fact that $|Q_t|_{\infty}\le (T-t+1)M$ and the excessive risk inequality of \cref{lemma:3} to replace the difference under 0-1 loss in terms of $\mathcal{V}_t$ by the difference under hinge loss in terms of $\mathfrak{L}_{\phi,t}$. 

For $(I)$, we have 
\begin{align*}
  (I)\le&\mathcal{V}_t(f^*_{t},...,f^*_{T})-\sup_{f_t\in\mathcal{A}_t(\tau_t,\mathcal{C}_{n,t})}[\mathcal{V}_t(f_t,\widehat{f}_{t+1},...,\widehat{f}_{T})-\mathcal{V}_t(f_t,f^*_{t+1}...,f^*_{T})+\mathcal{V}_t(f_t,f^*_{t+1}...,f^*_{T})]\\
  \le & \mathcal{V}_t(f^*_{t},...,f^*_{T})-\sup_{f_t\in\mathcal{A}_t(\tau_t,\mathcal{C}_{n,t})}\mathcal{V}_t(f_t,f^*_{t+1},...,f^*_{T})+c_1^{-1}|\mathcal{V}_{t+1}(f^*_{t+1},...,f^*_{T})-\mathcal{V}_{t+1}(\widehat{f}_{t+1},...,\widehat{f}_{T})|\\
  =&c_1^{-1}|\mathcal{V}_{t+1}(f^*_{t+1},...,f^*_{T})-\mathcal{V}_{t+1}(\widehat{f}_{t+1},...,\widehat{f}_{T})|+(T-t+1)M\eta_{n,t},\\
  &+\underbrace{\inf_{f_t\in\mathcal{A}_t(\tau_t,\mathcal{C}_{n,t})}\mathfrak{L}_{\phi,t}(f_t;f^*_{t+1},...,f^*_T)-\mathfrak{L}_{\phi,t}(f_t^*;f_{t+1}^*,...,f_T^*)}_{(III)}
\end{align*}
where again we have used the fact that $|Q_t|_{\infty}\le (T-t+1)M$ and the excessive risk inequality of \cref{lemma:3}. To bound the last term in $(I)$, recall that $f^*_{t,\tau}$ denotes the solution of (\ref{f_surrogate}) at stage $t$ with risk constraint $\tau_t$ replaced by $\tau$. Then the second part of \cref{lemma:8} indicates that $V_{\sigma_{n,t}}\check{f}_{t,\tau_t-\epsilon_{n,t}'}\in\mathcal{A}_t(\tau_t-\epsilon_{n,t},\mathcal{C}_{n,t})\subseteq\mathcal{A}_t(\tau_t,\mathcal{C}_{n,t})$, where $\check{f}_{t,\tau_t-\epsilon_{n,t}'}=(\sigma_{n,t}^2/\pi)\bar{f}_{t,\tau_t-\epsilon_{n,t}'}$, and
$$|\mathfrak{L}_{\phi,t}(V_{\sigma_{n,t}}\check{f}_{t,\tau_t-\epsilon_{n,t}'};f_{t+1}^*,...,f_T^*)-\mathfrak{L}_{\phi,t}(f^*_{t,\tau_t-\epsilon_{n,t}'};f_{t+1}^*,...,f_T^*)|\le c\sigma_{n,t}^{-\alpha_t d_t}.$$
Therefore, we have 
\begin{equation*}
  \begin{split}
    (III)\le &\mathfrak{L}_{\phi,t}(V_{\sigma_{n,t}}\check{f}_{t,\tau_t-\epsilon_{n,t}'};f_{t+1}^*,...,f_T^*)-\mathfrak{L}_{\phi,t}(f_t^*;f_{t+1}^*,...,f_T^*)\\
    \le&|\mathfrak{L}_{\phi,t}(V_{\sigma_{n,t}}\check{f}_{t,\tau_t-\epsilon_{n,t}'};f_{t+1}^*,...,f_T^*)-\mathfrak{L}_{\phi,t}(f^*_{t,\tau_t-\epsilon_{n,t}'};f_{t+1}^*,...,f_T^*)|\\
    &+\mathfrak{L}_{\phi,t}(f^*_{t,\tau_t-\epsilon_{n,t}'};f_{t+1}^*,...,f_T^*)-\mathfrak{L}_{\phi,t}(f^*_{t};f_{t+1}^*,...,f_T^*)\\
    \le&c(\sigma_{n,t}^{-\alpha_td_t}+\epsilon_{n,t}')\le O(\epsilon_{n,t}')
  \end{split}
\end{equation*}
where to obtain the last inequality we used the Lipschitz continuity of the value function in \cref{lemma:5} and by definition $f^*_t=f^*_{t,\tau_t}$.

For $(II)$, we have
\begin{equation}
  \label{II_decomposition}
  \begin{split}
    (II)\le& \mathfrak{L}_{\phi,t}(\widehat{f}_{t};\widehat{f}_{t+1},...,\widehat{f}_{T}) + \lambda_{n,t}\|\widehat{f}_{t}\|_{\mathcal{G}_t}^2-\underset{f_t\in\mathcal{A}_t(\tau_t,\mathcal{C}_{n,t})}{\inf}\mathfrak{L}_{\phi,t}(f_t;\widehat{f}_{t+1},...,\widehat{f}_{T})\\
    =&\bigg[\mathfrak{L}_{\phi,t}(\widehat{f}_{t};\widehat{f}_{t+1},...,\widehat{f}_{T}) + \lambda_{n,t}\|\widehat{f}_{t}\|_{\mathcal{G}_t}^2\\
    &-\inf_{f_t\in\mathcal{A}_{t,n}(\tau_t,\mathcal{C}_{n,t})}\bigg(\mathfrak{L}_{\phi,t}(f_t;\widehat{f}_{t+1},...,\widehat{f}_{T}) + \lambda_{n,t}\|f_t\|_{\mathcal{G}_t}^2\bigg)\bigg]\\
    &+\bigg[\inf_{f_t\in\mathcal{A}_{t,n}(\tau_t,\mathcal{C}_{n,t})}\bigg(\mathfrak{L}_{\phi,t}(f_t;\widehat{f}_{t+1},...,\widehat{f}_{T}) + \lambda_{n,t}\|f_t\|_{\mathcal{G}_t}^2\bigg)\\
    &-\underset{f_t\in\mathcal{A}_t(\tau_t,\mathcal{C}_{n,t})}{\inf}\mathfrak{L}_{\phi,t}(f_t;\widehat{f}_{t+1},...,\widehat{f}_{T})\bigg],
  \end{split}
\end{equation}
where 
$$\mathcal{A}_{t,n}(\tau,\mathcal{C}_{n,t})=\bigg\{f\in\mathcal{G}_t\bigg|\|f\|_{\mathcal{G}_t}\le\mathcal{C}_{n,t},\frac{1}{n}\sum_{i=1}^n\frac{R_{it}\psi(A_{it}f(H_{it}),\eta_{n,t})}{p(A_{it}|H_{it})}\le \tau_t\bigg\}.$$
The first term on the right-hand side of the inequality (\ref{II_decomposition}) can be bounded by
\begin{align*}
  &\mathfrak{L}_{\phi,t}(\widehat{f}_{t};\widehat{f}_{t+1},...,\widehat{f}_{T}) + \lambda_{n,t}\|\widehat{f}_{t}\|_{\mathcal{G}_t}^2-\inf_{f_t\in\mathcal{A}_{t,n}(\tau_t,\mathcal{C}_{n,t})}\bigg(\mathfrak{L}_{\phi,t}(f_t;\widehat{f}_{t+1},...,\widehat{f}_{T}) + \lambda_{n,t}\|f_t\|_{\mathcal{G}_t}^2\bigg)\\
  \le&2c_1^{-1}|\mathcal{V}_{t+1}(f^*_{t+1},...,f^*_{T})-\mathcal{V}_{t+1}(\widehat{f}_{t+1},...,\widehat{f}_{T})|\\
  &+\underbrace{\bigg[\mathfrak{L}_{\phi,t}(\widehat{f}_{t};f^*_{t+1},...,f^*_{T}) + \lambda_{n,t}\|\widehat{f}_{t}\|_{\mathcal{G}_t}^2-\inf_{f_t\in\mathcal{A}_{t,n}(\tau_t,\mathcal{C}_{n,t})}\bigg(\mathfrak{L}_{\phi,t}(f_t;f^*_{t+1},...,f^*_{T}) + \lambda_{n,t}\|f_t\|_{\mathcal{G}_t}^2\bigg)\bigg]}_{(IV)},
\end{align*}
and $(IV)$ is equal to stochastic error term $(I)$ in (\ref{decomposition}) for the proof of $T=1$ with $Y$ being replaced by $Q_t$. Note that $|Q_t|\le (T-t+1)M$ and consequently $(IV)$ can be bounded using exactly the same argument as for term $(I)$ of (\ref{decomposition}), which turns out to have order $O(\xi_{n,t}^{(1)})$ with probability at least $1-\exp(-x_t)$. For the second term of (\ref{II_decomposition}), we have
\begin{align*}
  &\inf_{f_t\in\mathcal{A}_{t,n}(\tau_t,\mathcal{C}_{n,t})}\bigg(\mathfrak{L}_{\phi,t}(f_t;\widehat{f}_{t+1},...,\widehat{f}_{T}) + \lambda_{n,t}\|f_t\|_{\mathcal{G}_t}^2\bigg)
  -\underset{f_t\in\mathcal{A}_t(\tau_t,\mathcal{C}_{n,t})}{\inf}\mathfrak{L}_{\phi,t}(f_t;\widehat{f}_{t+1},...,\widehat{f}_{T})\\
  \le&2c_1^{-1}|\mathcal{V}_{t+1}(f^*_{t+1},...,f^*_{T})-\mathcal{V}_{t+1}(\widehat{f}_{t+1},...,\widehat{f}_{T})|\\
  &+\underbrace{\bigg[\inf_{f_t\in\mathcal{A}_{t,n}(\tau_t,\mathcal{C}_{n,t})}\bigg(\mathfrak{L}_{\phi,t}(f_t;f^*_{t+1},...,f^*_{T}) + \lambda_{n,t}\|f_t\|_{\mathcal{G}_t}^2\bigg)-\underset{f_t\in\mathcal{A}_t(\tau_t,\mathcal{C}_{n,t})}{\inf}\mathfrak{L}_{\phi,t}(f_t;f^*_{t+1,\tau_{t+1}},...,f^*_{T})\bigg]}_{(V)}.
\end{align*}
Note that term $(V)$ is similar to the sum of the remaining terms in (\ref{decomposition}) also with $Y$ being replaced by $Q_t$. Hence, following the same argument as for $T=1$, $(V)$ can be similarly decomposed to terms $(II)-(IV)$ in (\ref{decomposition}) and bounded separately, which turns out to have order $O(\xi_{n,t}^{(2)})+O(\epsilon_{n,t}')$ in total with probability at least $1-2\exp\big(-\frac{n\delta_t^2c_1^2}{2(T-t+1)^2M^2}\big)$.
Combing these results, we conclude that with probability at least $1-h_n(t,x_t)$,
\begin{equation}
  \label{thm:s2:inequality1}
  \begin{split}
    \mathcal{V}_t(f^*_{t},...,f^*_{T})-\mathcal{V}_t(\widehat{f}_{t},...,\widehat{f}_{T}) \le& 5c_1^{-1}|\mathcal{V}_{t+1}(f^*_{t+1},...,f^*_{T})-\mathcal{V}_{t+1}(\widehat{f}_{t+1},...,\widehat{f}_{T})|\\
    &+c\eta_{n,t}^{-1}+\underbrace{c(\epsilon_{n,t}')}_{(III)}+\underbrace{c(\xi_{n,t}^{(1)}+\xi_{n,t}^{(2)}+\epsilon_{n,t}')}_{(IV)+(V)}\\
    \le&5c_1^{-1}|\mathcal{V}_{t+1}(f^*_{t+1},...,f^*_{T})-\mathcal{V}_{t+1}(\widehat{f}_{t+1},...,\widehat{f}_{T})|\\
    &+c(\xi_{n,t}+\epsilon_{n,t}'+\eta_{n,t}^{-1})
  \end{split}
\end{equation}
for some constant $c$.

On the other hand, according to \cref{lemma:8}, similar to the prove when $T=1$ we can show that $\widehat{f}_{t}\in\mathcal{A}_{t,n}(\tau_t,\mathcal{C}_{n,t})\subseteq\mathcal{A}_{t}(\tau_t+\epsilon_{n,t}',\mathcal{C}_{n,t})$ with probability at least $1-2\exp(\frac{n\delta_{t}^2c_1^2}{2(T-t+1)^2M^2})$. Therefore, we have 
\begin{equation}
  \label{thm:s2:inequality2}
  \begin{split}
    &\mathcal{V}_t(f^*_{t},...,f^*_{T})-\mathcal{V}_t(\widehat{f}_{t},...,\widehat{f}_{T})\\
    \geq&\mathcal{V}_t(f^*_{t},...,f^*_{T})-\underset{f_t\in\mathcal{A}_{t}(\tau_t+\epsilon_{n,t}',\mathcal{C}_{n,t})}{\sup}\mathcal{V}_t(f_t,\widehat{f}_{t+1},...,\widehat{f}_{T})\\
    \geq&\mathcal{V}_t(f^*_{t},...,f^*_{T})-\mathcal{V}_t(f^*_{t, \tau_{t}+\epsilon_{n,t}'},f^*_{t+1},...,f^*_{T})\\
    &-c_1^{-1}|\mathcal{V}_t(f^*_{t+1},...,f^*_{T})-\mathcal{V}_t(\widehat{f}_{t+1},...,\widehat{f}_{T})|\\
    \geq&c\epsilon_{n,t}'-c_1^{-1}|\mathcal{V}_t(f^*_{t+1},...,f^*_{T})-\mathcal{V}_t(\widehat{f}_{t+1},...,\widehat{f}_{T})|.
  \end{split}
\end{equation}

Finally, by combining (\ref{thm:s2:inequality1}) and (\ref{thm:s2:inequality2}), we obtain that with probability at least $1-h_n(t,x_t)$, 
\begin{equation*}
  \begin{split}
    |\mathcal{V}_t(f^*_{t},...,f^*_{T})-\mathcal{V}_t(\widehat{f}_{t},...,\widehat{f}_{T})|\le 5c_1^{-1}&|\mathcal{V}_{t+1}(f^*_{t+1},...,f^*_{T})-\mathcal{V}_{t+1}(\widehat{f}_{t+1},...,\widehat{f}_{T})|\\
    &+c(\xi_{n,t}+\epsilon_{n,t}'+\eta_{n,t}^{-1}).
  \end{split}
\end{equation*}
Hence, (\ref{value_error_bound}) in \cref{thm:s2} follows by induction starting from $t=T$ to $1$. The error bound of risk (\ref{risk_error_bound}) can be established by repeating the same proof as \cref{corollary:1} for each stage. This completes the proof of \cref{thm:s2}.

%%%%%%%%%%%%%%%%%% Appendix B %%%%%%%%%%%%%%%%%%%%%%%%%%%

\subsection*{{Appendix B: DC Algorithm for Solving Single-Stage BR-DTRs}}
\label{appendix:B}

In this section, we describe the DC algorithm for solving BR-DTRs at stage $t$. The algorithm was originally proposed in \citet{wang_learning_2018} for $T=1$. Given estimated rules $(\widehat{f}_{t+1},...,\widehat{f}_{T})$, one can calculate $\{\widehat{Y}_{it}\}$ and $\{\widehat{A}_{it}\}$ from (\ref{new_y_a}). We aim to solve the optimization problem 
\begin{equation*}
  \begin{alignedat}{2}
    \underset{\boldsymbol\beta\in\mathbb{R}^d,\beta_0\in\mathbb{R}}{\min}\quad&C_{n,t}\sum_{i=1}^n\frac{\widehat{Y}_{it}}{p(A_{it}|H_{it})}\phi\big(\widehat{A}_{it}(K_{i,t}\boldsymbol\beta+\beta_0)\big)+\frac{1}{2}\boldsymbol\beta^T K_t\boldsymbol\beta\\
    \text{subject~to}\quad&\sum_{i=1}^n\frac{R_{it}}{p(A_{it}|H_{it})}\psi\big(A_{it}(K_{i,t}\boldsymbol\beta+\beta_0),\eta\big)\leq n\tau_t,\\
  \end{alignedat}
\end{equation*}
where $C_{n,t}=(2n\lambda_{n,t})^{-1}$. Here, $K_t$ is the $n$-by-$n$ kernel matrix of stage $t$ defined by $K_{ij}=K_t(H_{it},H_{jt})$ where $K_t(\cdot,\cdot)$ is the inner product equipped by RKHS $\mathcal{G}_t$ and $K_{i,t}$ is the i-th row of $K_t$.

Note that the shifted ramp loss can be decomposed as $\psi(x,\eta)=\eta^{-1}(x+\eta)_{+}-\eta^{-1}(x)_{+}.$ By applying the DC algorithm, given $\boldsymbol\beta^{(s)}$ and $\beta_{0}^{(s)}$, we update ($\boldsymbol\beta^{(s+1)},\beta_0^{(s+1)}$) by solving the optimization problem 
\begin{equation*}
  \begin{alignedat}{2}
    \underset{\boldsymbol\beta\in\mathbb{R}^{d_t},\beta_0\in\mathbb{R}}{\min}\quad&C_{n,t}\sum_{i=1}^n\frac{\widehat{Y}_{it}}{p(A_{it}|H_{it})}\phi\big(\widehat{A}_{it}(K_{i,t}\boldsymbol\beta+\beta_0)\big)+\frac{1}{2}\boldsymbol\beta^T K_t\boldsymbol\beta\\
    \text{subject~to}\quad&\sum_{i=1}^n\frac{R_{it}}{p(A_{it}|H_{it})}\bigg[\big\{A_{it}(K_{i,t}\boldsymbol\beta+\beta_0)+\eta\big\}_{+}-C_{it}^{(s)}A_{it}(K_{i,t}\boldsymbol\beta+\beta_0)\bigg]\leq n\eta\tau_t,
  \end{alignedat}
\end{equation*}
where $C_{it}^{(s)}=\mathbb{I}(A_{it}(K_{i,t}\boldsymbol\beta^{(s)}+\beta_0^{(s)})>0)$. Similar to standard SVM, we introduce slacking variables $\xi_i\geq1-\widehat{A}_{it}(K_{i,t}\boldsymbol\beta+\beta_0)$, $\xi_i\geq0$ to replace $\phi\big(\widehat{A}_{it}(K_{i,t}\boldsymbol\beta+\beta_0)\big)$ in the objective function. Moreover, we introduce additional slacking variables $\zeta_i\geq A_{it}(K_{i,t}\boldsymbol\beta+\beta_0)+\eta$, $\zeta_i\geq0$ to replace $\big\{A_{it}(K_{i,t}\boldsymbol\beta+\beta_0)+\eta\big\}_{+}$ in the risk constraint. After plugging the slacking variables, the optimization problem becomes 
\begin{equation}
  \label{dc_problem}
  \begin{alignedat}{2}
    \underset{\boldsymbol\beta\in\mathbb{R}^{d_t},\beta_0\in\mathbb{R}}{\min}\quad&C_{n,t}\sum_{i=1}^n\frac{\widehat{Y}_{it}}{p(A_{it}|H_{it})}\xi_i+C_{n,t}\sum_{i=1}^n\frac{\zeta_i}{n}+\frac{1}{2}\boldsymbol\beta^T K_t\boldsymbol\beta\\
    \text{subject~to}\quad&\sum_{i=1}^n\frac{R_{it}}{p(A_{it}|H_{it})}\bigg[\zeta_i-C_{it}^{(s)}A_{it}(K_{i,t}\boldsymbol\beta+\beta_0)\bigg]\leq n\eta\tau_t,\\
    &1-\widehat{A}_{it}(K_{i,t}\boldsymbol\beta+\beta_0)\leq\xi_i,~~0\leq\xi_i,\\
    & A_{it}(K_{i,t}\boldsymbol\beta+\beta_0)+\eta\leq\zeta_i,~~0\leq\zeta_i,~~\text{for~~}i=1,..,n.
  \end{alignedat}
\end{equation}
The additional term $C_{n,t}\sum_{i=1}^n\frac{\zeta_i}{n}$ in the objective function is to guarantee that the slacking variable $\zeta_i$ is equal to $\big\{A_{it}(K_{i,t}\boldsymbol\beta+\beta_0)+\eta\big\}_{+}$. For fixed tuning parameter $C_{n,t}$, this optimization problem will be equivalent to the original problem as the additional term will eventually vanish when the sample size $n$ increases. 

The Lagrange function of (\ref{dc_problem}) is given by 
\begin{equation*}
  \begin{split}
    L=&C_{n,t}\bigg(\sum_{i=1}^n\frac{\widehat{Y}_{it}}{p(A_{it}|H_{it})}\xi_i+\sum_{i=1}^n\frac{\zeta_i}{n}\bigg)+\frac{1}{2}\boldsymbol\beta^TK_t\boldsymbol\beta\\
    &-\pi\bigg[n\eta\tau_t-\sum_{i=1}^n\frac{R_{it}}{p(A_{it}|H_{it})}\bigg(\zeta_i-\sum_{i=1}^nC_{it}^{(s)}A_{it}(K_{i,t}\boldsymbol\beta+\beta_0)\bigg)\bigg]\\
    &-\sum_{i=1}^n\alpha_i\bigg[\xi_i-1+\widehat{A}_{it}(K_{i,t}\boldsymbol\beta+\beta_0)\bigg]-\sum_{i=1}^n\mu_i\xi_i-\sum_{i=1}^n\kappa_i\bigg[\zeta_i-\eta-A_{it}(K_{i,t}\boldsymbol\beta+\beta_0)\bigg]-\sum_{i=1}^n\rho_i\kappa_i.
  \end{split}
\end{equation*}
Taking derivatives w.r.t. $\xi_i$, $\zeta_i$, $\boldsymbol\beta$ and $\beta_0$, one can obtain that the optimal Lagrange multipliers $\boldsymbol\alpha=(\alpha_1,...,\alpha_n^T)$, $\boldsymbol\kappa=(\kappa_1,...,\kappa_n)^T$, $\boldsymbol\mu=(\mu_1,...,\mu_n)$, $\boldsymbol\rho=(\rho_1,...,\rho_n)$ and $\pi$ satisfy
\begin{equation*}
  \begin{split}
    C_{n,t}\boldsymbol{V}_{t,Y}-\boldsymbol\alpha-\boldsymbol\mu&=\boldsymbol0,\\
    C_{n,t}\boldsymbol{1}/n+\pi\boldsymbol{V}_{t,R}-\boldsymbol\kappa-\boldsymbol\rho&=\boldsymbol0,\\
    \boldsymbol\beta-\pi \boldsymbol{V}_{t,R,A,C}^{(s)}-\widehat{A}_t\boldsymbol\alpha+A_{t}\boldsymbol\kappa&=\boldsymbol0,\\
    \pi\boldsymbol{1}^T\boldsymbol{V}_{t,R,A,C}^{(s)}+\boldsymbol{1}^T\widehat{A}_t\boldsymbol\alpha-\boldsymbol{1}^TA_t\boldsymbol\kappa&=0,
  \end{split}
\end{equation*}
where $\boldsymbol{1}$ and $\boldsymbol{0}$ denote the $n$-by-$1$ vectors with all entries equal to 1 and 0 respectively,
{\everymath{\displaystyle}
\begin{equation*}
  \boldsymbol{V}_{t,Y}=
  \begin{bmatrix}
    \frac{\widehat{Y}_{1t}}{p(A_{1t}|H_{1t})}\\
    \vdots\\
    \frac{\widehat{Y}_{nt}}{p(A_{nt}|H_{nt})}
  \end{bmatrix},\quad
  \boldsymbol{V}_{t,R}=
  \begin{bmatrix}
    \displaystyle
    \frac{R_{1t}}{p(A_{1t}|H_{1t})}\\
    \vdots\\
    \frac{R_{nt}}{p(A_{nt}|H_{nt})}
  \end{bmatrix},\quad
  \boldsymbol{V}_{t,R,A,C}^{(s)}=
  \begin{bmatrix}
    \displaystyle
    \frac{R_{1t}}{p(A_{1t}|H_{1t})}A_{1t}C_{1t}^{(s)}\\
    \vdots\\
    \frac{R_{nt}}{p(A_{nt}|H_{nt})}A_{nt}C_{nt}^{(s)}
  \end{bmatrix}.
\end{equation*}
}Here, we abuse the notation and define $\widehat{A}_t=\text{diag}\{(\widehat{A}_{1t},...,\widehat{A}_{nt})\}$ and $A_t=\text{diag}\{(A_{1t},...,A_{nt})\}$. Plugging the equations back to $L$ and note that $\boldsymbol\alpha\geq0$, $\boldsymbol\kappa\geq\boldsymbol0$, $\boldsymbol\mu\geq\boldsymbol0$, $\boldsymbol\rho\geq\boldsymbol0$ and $\pi\geq0$, after some algebra one can obtain that the dual problem of (\ref{dc_problem}) w.r.t. $\boldsymbol\omega=(\pi,\boldsymbol\alpha^T,\boldsymbol\kappa^T)^T$ is given by
\begin{equation*}
  \begin{alignedat}{2}
    \underset{\boldsymbol\omega}{\min}\quad&\frac{1}{2}\boldsymbol\omega^T(H^TK_tH)\boldsymbol\omega-\boldsymbol\omega^T\boldsymbol{l}_{\eta,\tau_t}\\
    \text{subject~to}\quad&\boldsymbol{a}\leq W\boldsymbol\omega\leq \boldsymbol{b},\\
    &\boldsymbol 0_{(2n+1)\times1}\leq\boldsymbol\omega\leq \boldsymbol{u},
  \end{alignedat}
\end{equation*}
where 
{\everymath{\displaystyle}
\begin{equation*}
  H=\begin{bmatrix}
    \displaystyle
    \boldsymbol{V}_{t,R,A,C}^{(s)}&\widehat{A}_t&-A_t
  \end{bmatrix},
  \quad
  W=\begin{bmatrix}
    \displaystyle
    \boldsymbol{V}_{t,R}&\boldsymbol{0}_{n\times n}&-I_n\\
    \boldsymbol{1}^T\boldsymbol{V}_{t,R,A,C}^{(s)}&\boldsymbol{1}^T\widehat{A}_t&-\boldsymbol{1}^TA_t
  \end{bmatrix},\quad 
\end{equation*}
}$\boldsymbol{l}_{\eta,\tau_t}=(-n\eta\tau_t,\textbf{1}^T,\eta\textbf{1}^T)^T$, $\boldsymbol{a}=(-C_{n,t}\textbf{1}^T/n,0)^T$, $\boldsymbol{b}=(\infty\textbf{1}^T,0)^T$ and $\boldsymbol{u}=(\infty,C_{n,t}\boldsymbol{V}_{t,Y}^T,\infty\textbf{1}^T)^T$. Note that the optimization w.r.t. $\boldsymbol\omega$ is a standard quadratic optimization problem, which can be solved efficiently using standard software. Denote the optimal solution of the previous optimization problem by $\widehat{\boldsymbol\omega}^{(s)}$, we update $\boldsymbol\beta$ by
$$\boldsymbol\beta^{(s+1)}=\widehat{\pi}^{(s)}\boldsymbol{V}_{t,R,A,C}^{(s)}+\widehat{A}_t\widehat{\boldsymbol\alpha}^{(s)}-A_t\widehat{\boldsymbol\kappa}^{(s)}.$$ 
The new $\beta_0^{(s+1)}$ can be determined via grid search such that the original objective function is maximized among values that satisfy the constraint given $\boldsymbol\beta=\boldsymbol\beta^{(s+1)}$.
We stop the iteration when the termination condition $\max(|\boldsymbol\beta^{(s+1)}-\boldsymbol\beta^{(s)}|_{\infty},|\beta_0^{(s+1)}-\beta^{(s)}|)\leq\epsilon$ is satisfied or reaches the maximum iteration limitation. Let $\widehat{\boldsymbol\beta}=(\widehat{\beta}_1,...,\widehat{\beta}_n)^T$ and $\widehat{\beta}_0$ denote the final solution returned by DC iteration, then the final estimated decision function at stage $t$ is given by $\widehat{f}_{t}(\cdot)=\sum_{i=1}^nK_t(H_{it},\cdot)\widehat{\beta}_i+\widehat{\beta}_0.$

%%%%%%%%%%%%%%%%%% Appendix C %%%%%%%%%%%%%%%%%%%%%%%%%%%

\newpage
\appendix
\label{appendix:C}
\renewcommand{\thesection}{C.\arabic{section}}
\renewcommand{\thesubsection}{C.\arabic{section}.\arabic{subsection}} 
\renewcommand{\thesubsubsection}{C.\arabic{section}.\arabic{subsection}.\arabic{subsubsection}} 
\setcounter{table}{0}
\setcounter{figure}{0}

\section*{Appendix C: Additional Simulation Results}
\section{Simulation results for setting I with $\tau=1.5$ and setting II with $\tau=1.3$}
\begin{figure}[hbp!]
  \centering
  \textbf{Setting I $\tau=1.5$}\\
  ~\\
  \includegraphics[scale = 0.25]{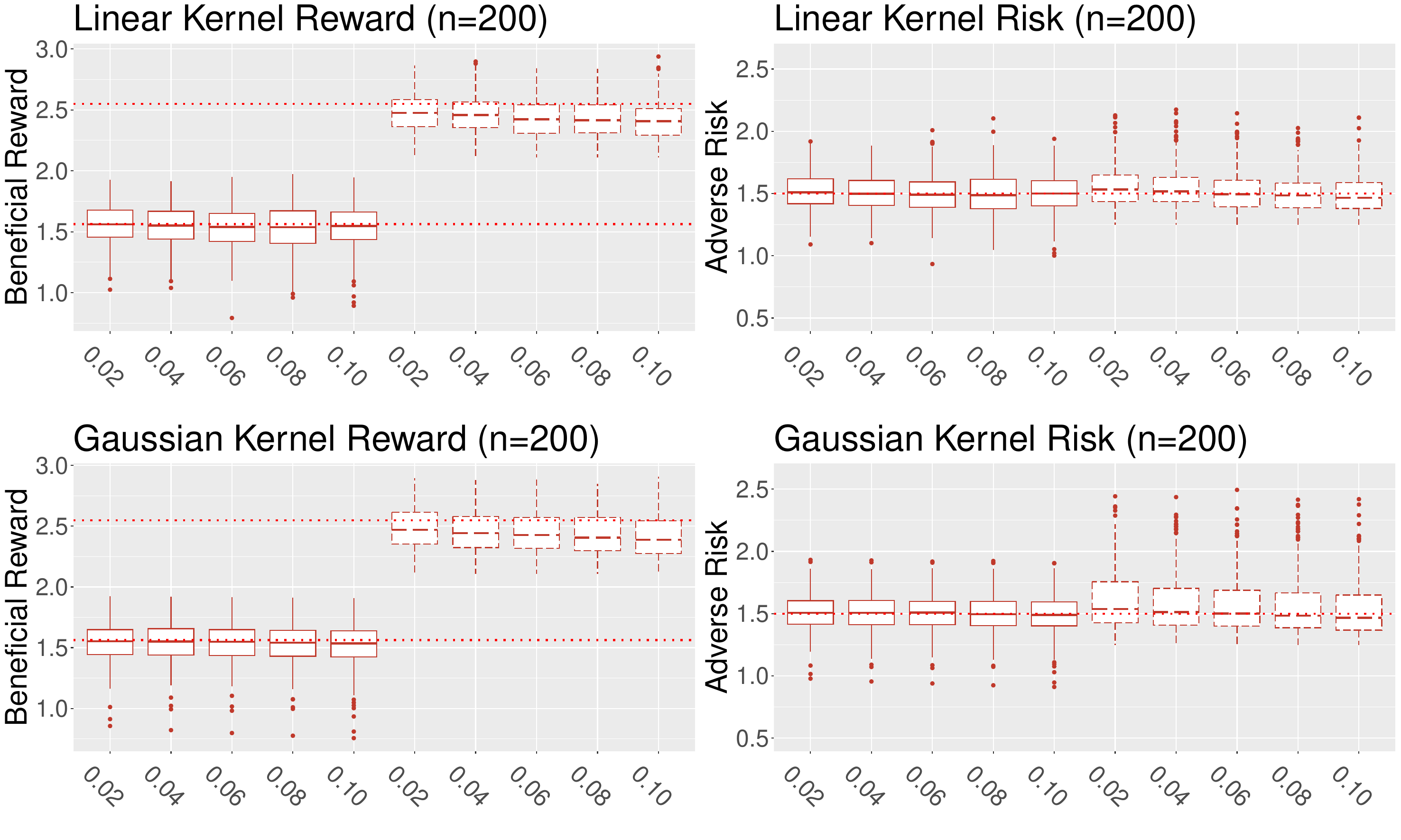}\\
  \includegraphics[scale = 0.25]{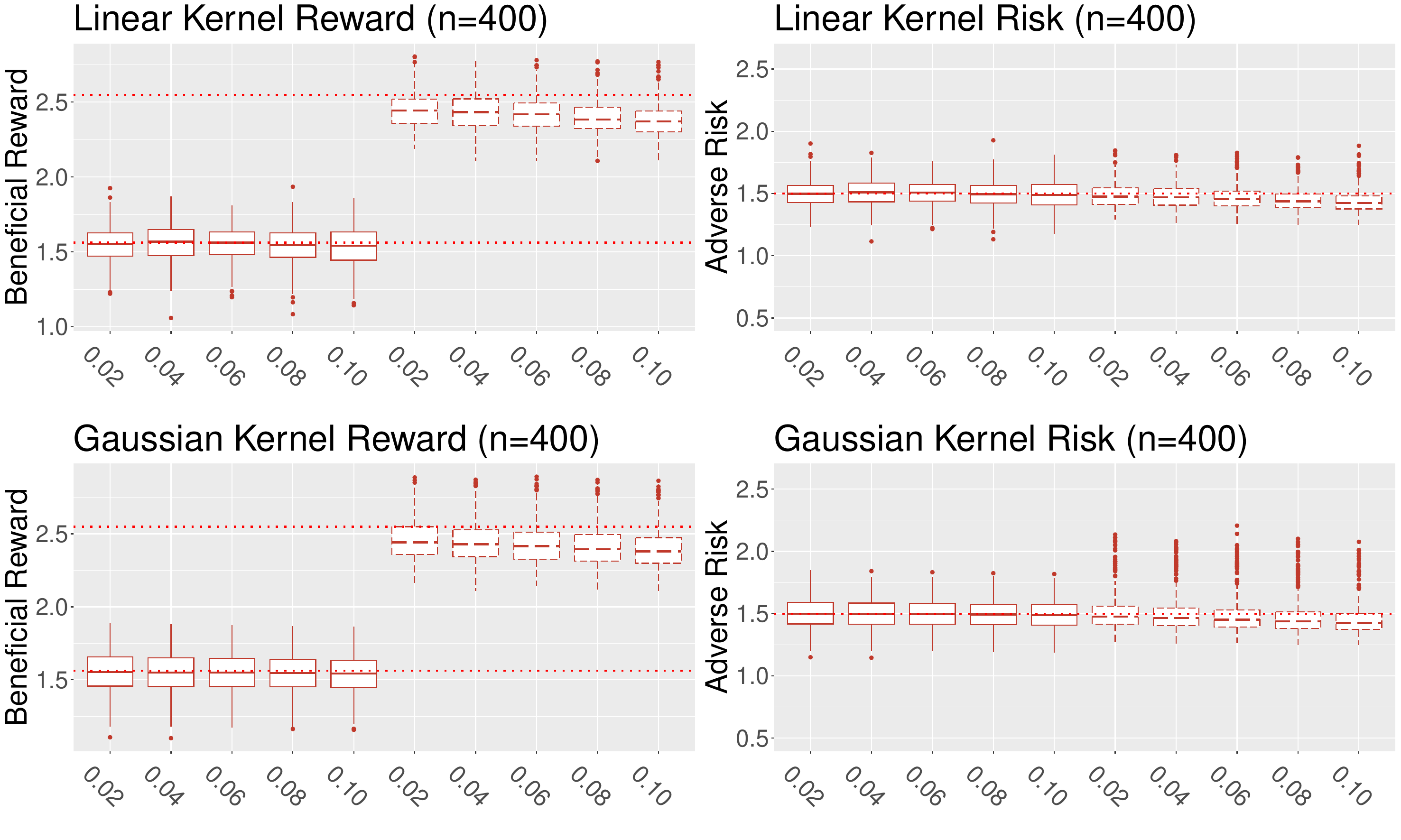}\\
  \vspace{-10pt}
  \includegraphics[scale = 0.3]{legend.pdf}
  \vspace{-10pt}
  \caption{\scriptsize Estimated reward/risk on independent testing data set for simulation setting I, training sample size $n=\{200,400\}$, $\eta=\{0.02,0.04,...,0.1\}$ (x-axis) under linear kernel or Gaussian kernel. The dashed line in reward plots refers to the theoretical optimal reward under given constraints. The dashed line in risk plots represents the risk constraint $\tau=1.5$.}
  \label{sfig:1}
\end{figure}

% Nonlinear 1.3
\newpage
\begin{figure}[hbp!]
  \centering
  \textbf{Setting II $\tau=1.3$}\\
  ~\\
  \includegraphics[scale = 0.25]{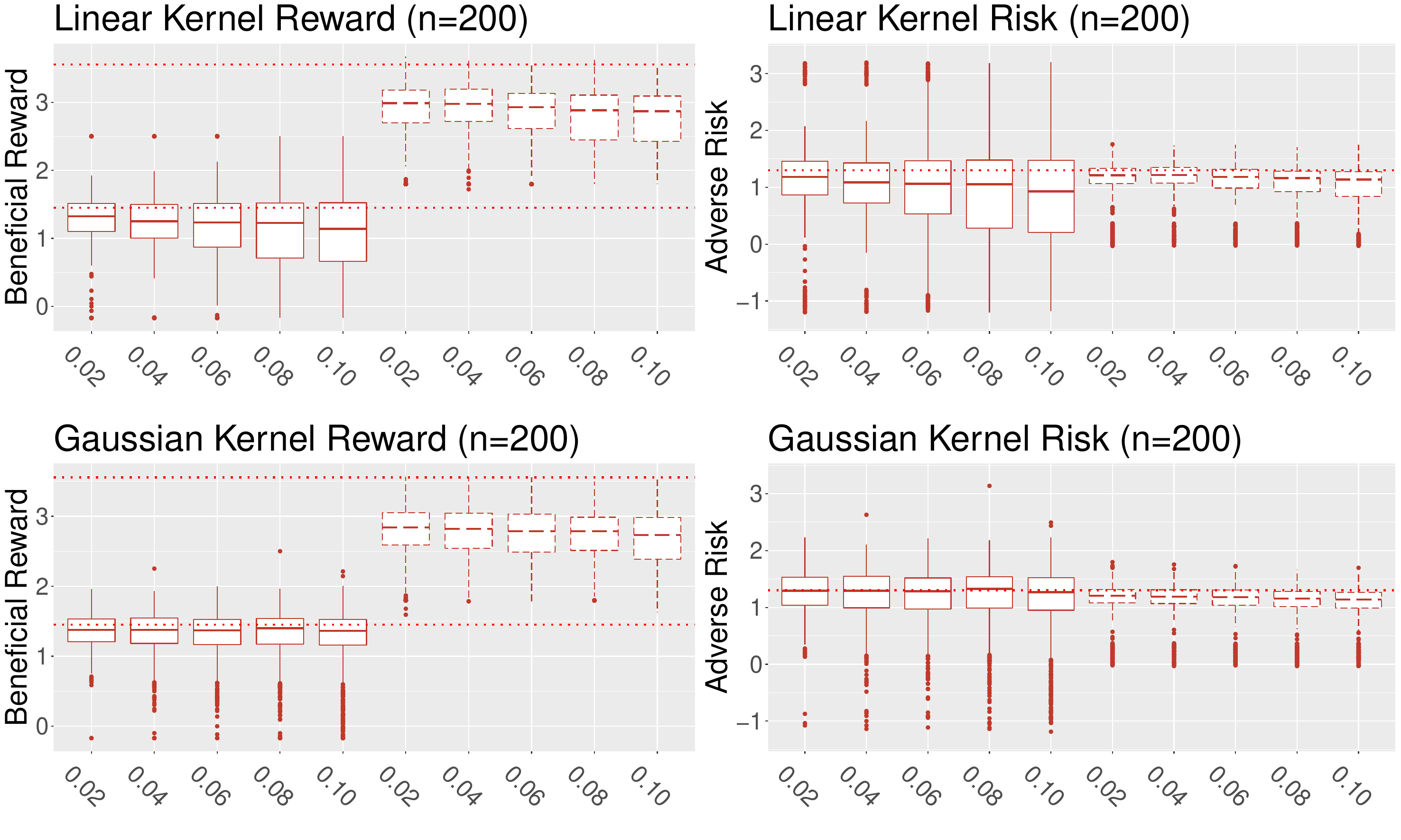}\\
  \includegraphics[scale = 0.25]{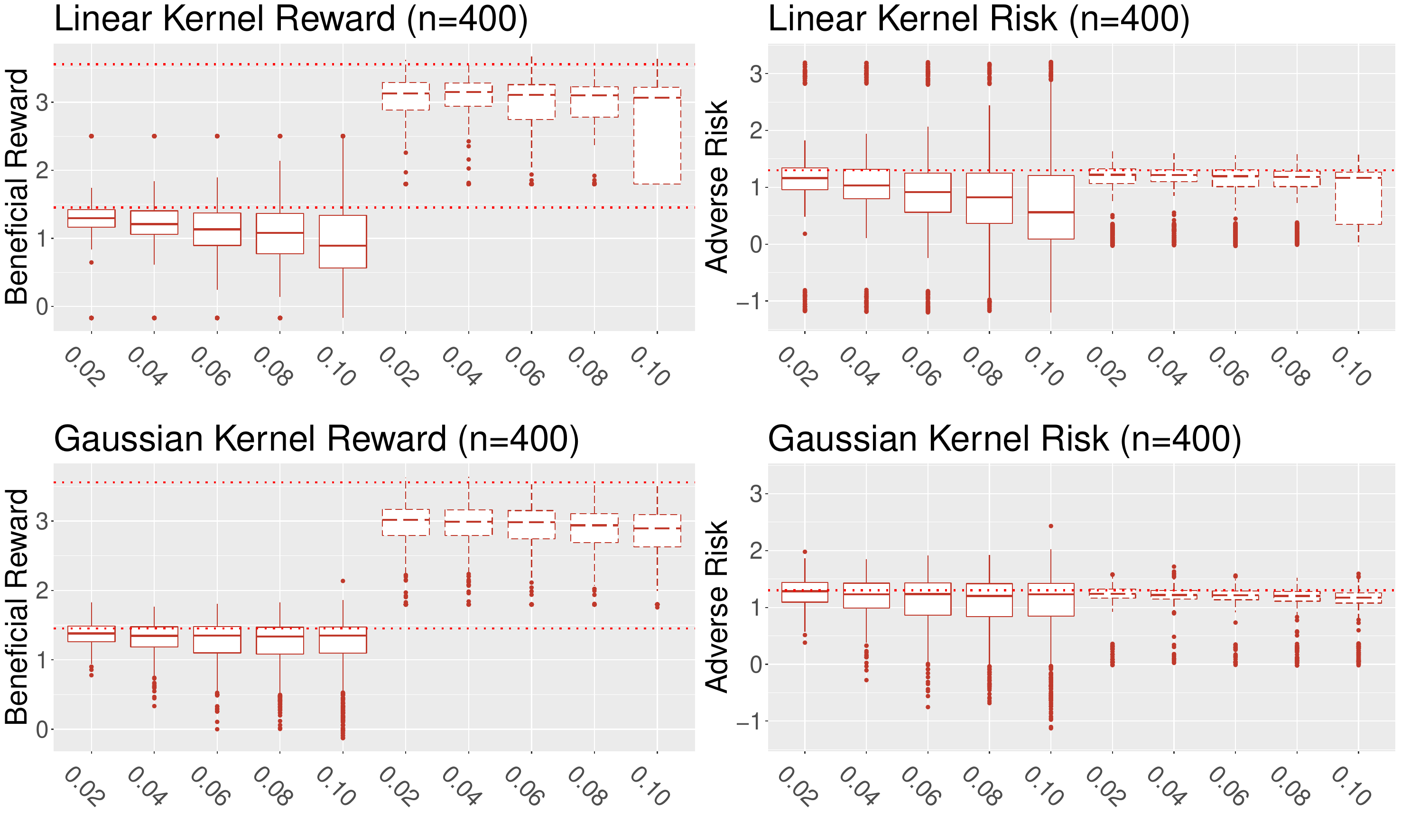}\\
  \vspace{-10pt}
  \includegraphics[scale = 0.3]{legend.pdf}
  \vspace{-10pt}
  \caption{\scriptsize Estimated reward/risk on independent testing data set for simulation setting II, training sample size $n=\{200,400\}$, $\eta=\{0.02,0.04,...,0.1\}$ (x-axis) under linear kernel or Gaussian kernel. The dashed line in reward plots refers to the theoretical optimal reward under given constraints. The dashed line in risk plots represents the risk constraint $\tau=1.3$.}
  \label{sfig:2}
\end{figure}

% Table Linear 1.5, Nonlinear 1.3
\afterpage{
\begin{landscape}
  \renewcommand{\arraystretch}{1.25}
  \begin{table}[t]
    \centering
    \scriptsize
    \caption{\label{stable1}\scriptsize Estimated reward/risk on independent testing data for setting I - $\tau_1=\tau_2=1.5$, setting II - $\tau_1=\tau_2=1.3$ and $n=400$ under 3 different methods using linear/Gaussian kernel. Results are reported in the same format as \cref{table1}.}
    \begin{tabular}{lcccccccccc}
      \hline
      &      &        & \multicolumn{4}{c}{Linear Kernel}                               & \multicolumn{4}{c}{Gaussian Kernel}                             \\\cline{4-11}
      Setting    & $\eta$  & Method & Reward - II & Risk - II & Cumulative Reward & Risk - I & Reward - II & Risk - II & Cumulative Reward & Risk - I \\
      \hline\\
      Setting I  & 0.02 & BR-DTRs & 1.557(0.080)    & 1.502(0.070)  & 2.443(0.082)   & 1.475(0.066) & 1.551(0.096)    & 1.497(0.084)  & 2.436(0.091)   & 1.475(0.068) \\
               & 0.02 & Naive  & ---             & ---           & 2.339(0.096)   & 1.460(0.083)   & ---             & ---           & 2.318(0.115)   & 1.448(0.097) \\
               & 0.04 & BR-DTRs & 1.568(0.082)    & 1.510(0.072)  & 2.429(0.087)   & 1.467(0.066)   & 1.546(0.096)    & 1.493(0.084)  & 2.425(0.089)   & 1.463(0.066) \\
               & 0.04 & Naive  & ---             & ---           & 2.319(0.098)   & 1.430(0.088)   & ---             & ---           & 2.304(0.113)   & 1.428(0.096) \\
               & 0.06 & BR-DTRs & 1.563(0.075)    & 1.507(0.066)  & 2.420(0.081)   & 1.455(0.059)   & 1.545(0.095)    & 1.492(0.084)  & 2.414(0.091)   & 1.450(0.063) \\
               & 0.06 & Naive  & ---             & ---           & 2.306(0.100)   & 1.416(0.084)   & ---             & ---           & 2.292(0.108)   & 1.411(0.093) \\
               & 0.08 & BR-DTRs & 1.546(0.082)    & 1.493(0.071)  & 2.387(0.075)   & 1.439(0.059)   & 1.541(0.096)    & 1.489(0.084)  & 2.393(0.091)   & 1.438(0.061) \\
               & 0.08 & Naive  & ---             & ---           & 2.274(0.112)   & 1.397(0.083)   & ---             & ---           & 2.277(0.110)   & 1.400(0.089) \\
               & 0.1  & BR-DTRs & 1.544(0.093)    & 1.489(0.082)  & 2.371(0.068)   & 1.421(0.054)   & 1.539(0.093)    & 1.486(0.082)  & 2.375(0.086)   & 1.424(0.060) \\
               & 0.1  & Naive  & ---             & ---           & 2.273(0.093)   & 1.383(0.080)   & ---             & ---           & 2.268(0.104)   & 1.387(0.083) \\
               &      & AOWL   & 1.983(0.010)    & 2.149(0.044)  & 3.257(0.018)   & 2.678(0.096)   & 1.914(0.030)    & 2.099(0.083)  & 3.212(0.036)   & 2.584(0.218) \\
      \hline\\
      Setting II & 0.02 & BR-DTRs & 1.297(0.132)    & 1.159(0.196)  & 3.128(0.178)   & 1.221(0.116) & 1.380(0.114)    & 1.285(0.169)  & 3.019(0.184)   & 1.240(0.079) \\
               & 0.02 & Naive  & ---             & ---           & 1.797(0.000)   & 0.166(0.000) & ---             & ---           & 1.797(0.000)   & 0.166(0.000) \\
               & 0.04 & BR-DTRs & 1.209(0.168)    & 1.034(0.252)  & 3.150(0.157)   & 1.215(0.099) & 1.347(0.139)    & 1.231(0.206)  & 2.991(0.192)   & 1.218(0.077) \\
               & 0.04 & Naive  & ---             & ---           & 1.797(0.000)   & 0.166(0.000) & ---             & ---           & 1.797(0.000)   & 0.166(0.000) \\
               & 0.06 & BR-DTRs & 1.131(0.238)    & 0.917(0.351)  & 3.109(0.179)   & 1.193(0.129) & 1.349(0.174)    & 1.233(0.253)  & 2.983(0.194)   & 1.213(0.078) \\
               & 0.06 & Naive  & ---             & ---           & 1.797(0.000)   & 0.166(0.000) & ---             & ---           & 1.797(0.000)   & 0.166(0.000) \\
               & 0.08 & BR-DTRs & 1.080(0.299)    & 0.821(0.440)  & 3.102(0.172)   & 1.182(0.115) & 1.335(0.168)    & 1.202(0.247)  & 2.937(0.199)   & 1.202(0.088) \\
               & 0.08 & Naive  & ---             & ---           & 1.797(0.000)   & 0.166(0.000) & ---             & ---           & 1.797(0.000)   & 0.166(0.000) \\
               & 0.1  & BR-DTRs & 0.893(0.394)    & 0.559(0.585)  & 3.067(0.205)   & 1.164(0.14)  & 1.350(0.155)    & 1.229(0.235)  & 2.893(0.217)   & 1.173(0.087) \\
               & 0.1  & Naive  & ---             & ---           & 1.797(0.000)   & 0.166(0.000) & ---             & ---           & 1.797(0.000)   & 0.166(0.000) \\
               &      & AOWL   & 2.440(0.064)    & 3.017(0.002)  & 5.188(0.000)   & 2.839(0.000) & 2.424(0.080)    & 3.018(0.002)  & 5.188(0.000)   & 2.839(0.000)\\
      \hline
    \end{tabular}
  \end{table}
\end{landscape}}

\section{An additional simulation with an unknown and unbalanced study design}
{In this section, we conduct an additional simulation study under setting II but now assume that the treatment assignment depends on covariates as
\begin{align*}
  \text{logit}~p(A_1|H_1)= -X_1+0.25,\quad\text{logit}~p(A_2|H_2)= -X_1+X_2-0.25,
\end{align*}
and treatment probabilities are unknown and will be estimated from data as in observational studies. 
We conduct the simulation following a similar procedure as described in \cref{sec:simulation} with $\tau_1=\tau_2=1.5$ except that in the training step, we estimate the unknown treatment assignment probability via the Lasso logistic regression using all observed features as predictors. The results are presented in \cref{sfig:3} and \cref{stable2} and the conclusions are similar to the findings in \cref{sec:simulation}.}

{\section{A simulated application to personalized promotion in E-commerce}
In this section, we simulate a toy personalized promotion problem and apply BR-DTRs to explore the performance on more general constrained decision-making problems. In this simulation, we let $T=4$ and $\{A_t\}_{t=1}^4$ denote the promotion action taken at wave $t=1,...,4$ with $A_t\in\{-1,+1\}$ denoting two available promotions. We use $\{(Y_t,C_t)\}_{t=1}^4$ to denote the instant commercial reward and the promotion cost after taking actions $\{A_t\}_{t=1}^4$. We assume that five baseline features $(X_1,...,X_5)$ are available for each customer. In this simulation, the instant commercial rewards and the promotion costs are generated according to
\begin{align*}
  Y_1&=1+X_2+(1+A_1)(X_1+X_2)+\epsilon_{Y_1},\\
  R_1&=1+2(1+A_1)(2X_1+1)+\epsilon_{R_1},\\
  Y_2&=1.5+X_2+(1+A_2)(Y_1/4+A_1/2+1/2)+\epsilon_{Y_2},\\
  R_2&=1+2(1+A_2)(X_1+1)+\epsilon_{R_2},\\
  Y_3&=1.5+X_2+(1+A_3)(Y_2/4+A_2/2+1/2)+\epsilon_{Y_3},\\
  R_3&=1+2(1+A_3)(X_1+1)+\epsilon_{R_3},\\
  Y_4&=1.5+X_2+(1+A_4)(Y_3/4+A_3/2+1/2)+\epsilon_{Y_4},\\
  R_4&=1+2(1+A_4)(X_1+1)+\epsilon_{R_4},
\end{align*}
where $\{\epsilon_{Y_t}\}_{t=1}^4$ and $\{\epsilon_{R_t}\}_{t=1}^4$ denote independent random noisy terms with distribution $\text{Unif}[-0.5,0.5]$.}

{We assume that the promotion has been delivered to $n=400$ customers in a pilot study with equal probability, i.e., $P(A_t|X_1,...,X_5)=0.5$ for $t=1,...,4$. Our goal is to learn an optimal strategy from observed data 
$$\{(A_{i1},Y_{i1},R_{i1},...,A_{i4},Y_{i4},R_{i4},X_{i1},...,X_{i5})\}_{i=1}^n$$
to determine which promotion should be delivered to customers so that the cumulative commercial reward $Y=\sum_{t=1}^4Y_t$ is maximized at a cost no larger than $\{\tau_t\}_{t=1}^4$ for each wave according to customers' baseline features, and the promotion costs and instant rewards from the previous waves of promotions. We note that BR-DTRs can be applied to handle the problem by treating $\{Y_t\}_{t=1}^4$ and $\{C_t\}_{t=1}^4$ as the reward and stagewise risks respectively, with $H_1=(X_1,...,X_5)$, $H_2=(H_1,A_1,Y_1,R_1)$, $H_3=(H_2,A_2,Y_2,R_2)$, and $H_4=(H_3,A_3,Y_3,R_3)$. }

{To evaluate the performance of BR-DTRs, we repeat the simulation 100 times and assess the estimated rules on independent testing data with $N=5,000$. We apply both the linear and the Gaussian kernel and the estimation following the same setting as \cref{sec:simulation} except that $\{C_{n,t}\}_{t=1}^4$ are selected from tuning grid $2^{\{-2,-1,0,1,2\}}$ via two-folds cross-validation. We repeat the analyses with $\tau=\tau_1=\cdots=\tau_4$ equal to 4, 5 and 6 respectively and fixing $\eta=\eta_1=\cdots=\eta_4=0.02$. {The complete results are displayed in \cref{stable3}.} The findings are similar to \cref{sec:simulation}, which indicate that BR-DTRs still yields estimated rules with stagewise risks controlled below or close to the prespecified constraints for more general decision-making problems.}

%%%%%%%%%%%%%%%%%% Appendix D %%%%%%%%%%%%%%%%%%%%%%%%%%%

\section*{Appendix D: Flowchart of the Study Design of the DURABLE Trial}
\begin{figure}[hbp!]
  \centering
  \includegraphics[scale=0.5]{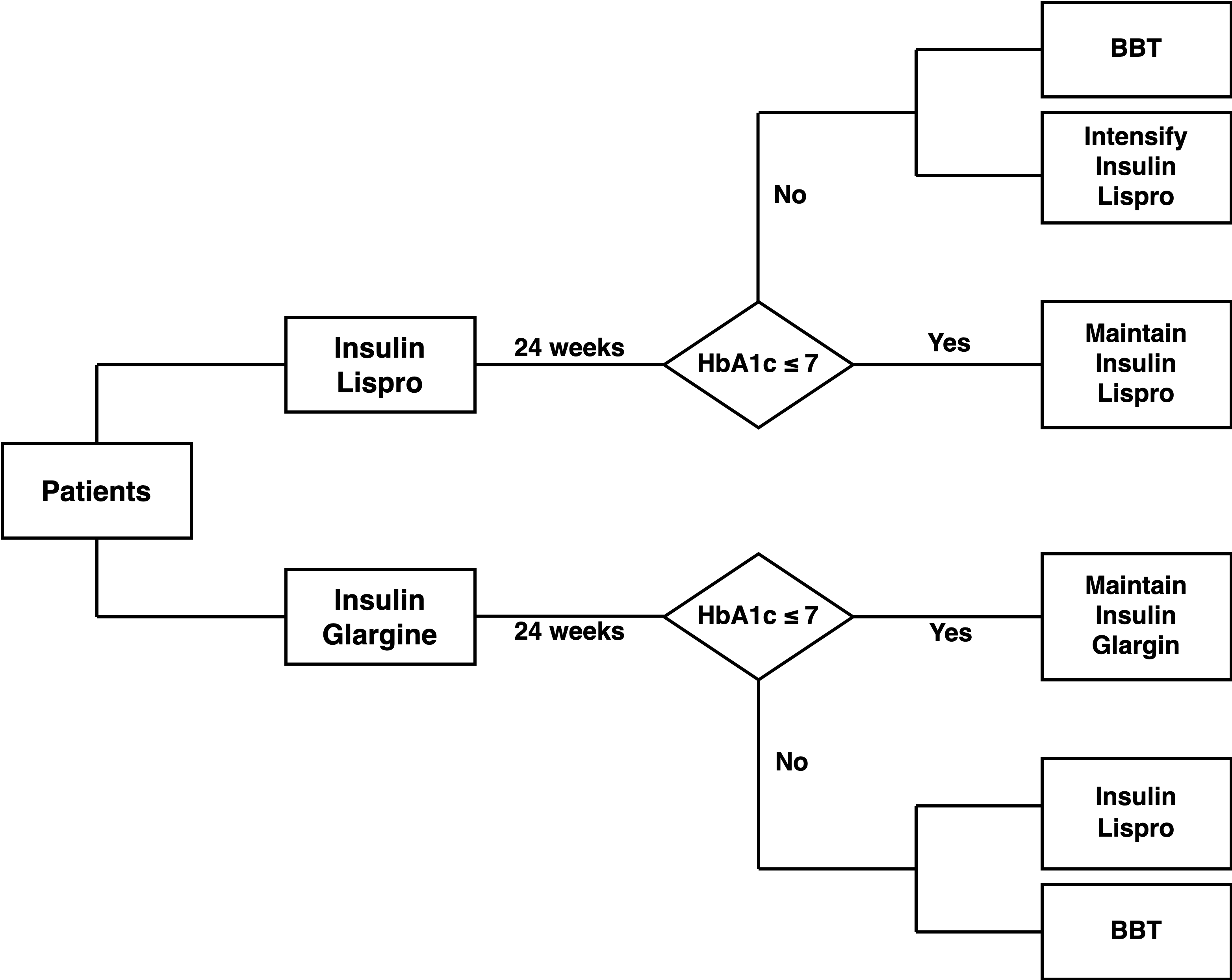}
  \caption{Study design of the DURABLE trial.}
\end{figure}

\begin{figure}[hbp!]
  \centering
  \includegraphics[scale = 0.25]{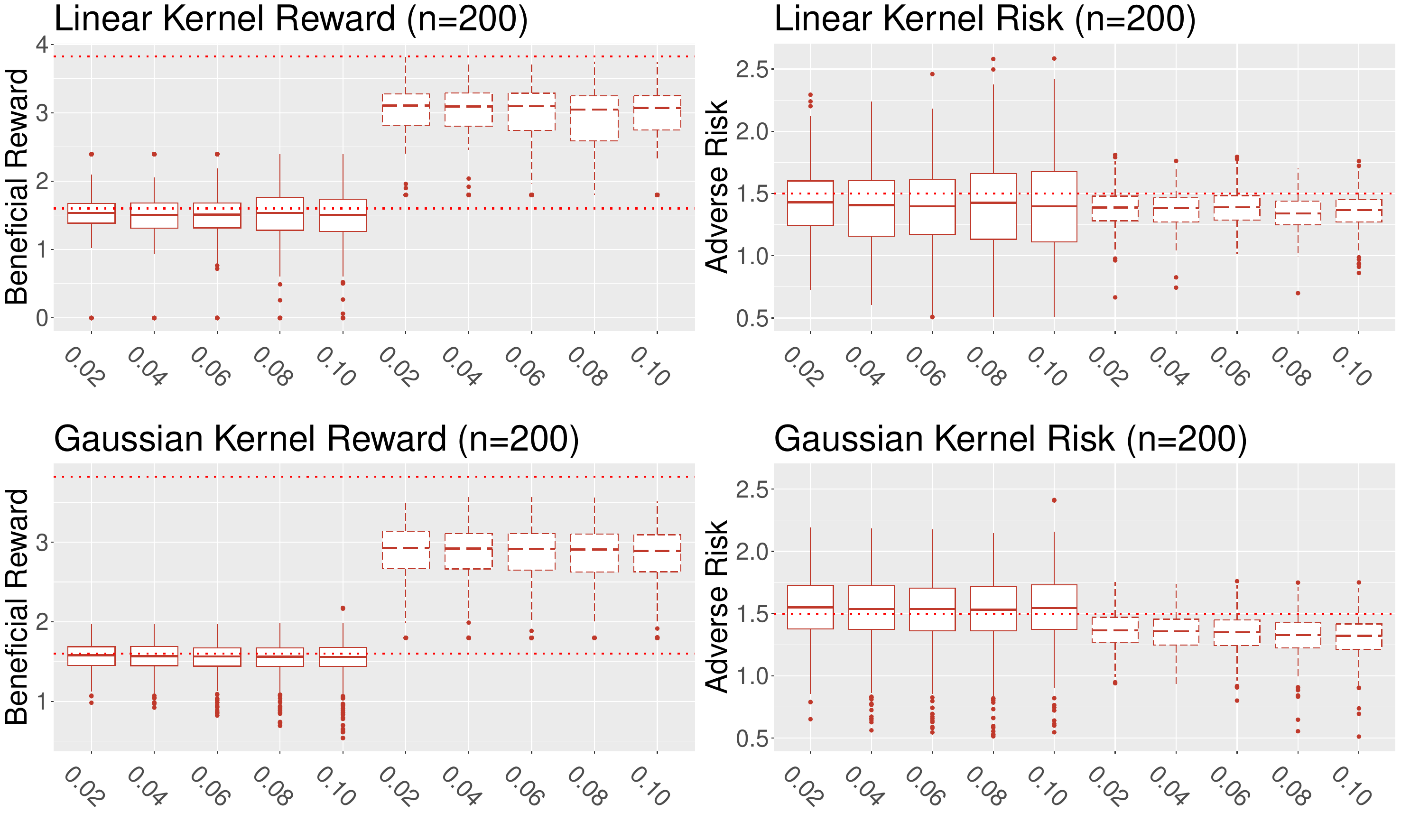}
  \includegraphics[scale = 0.25]{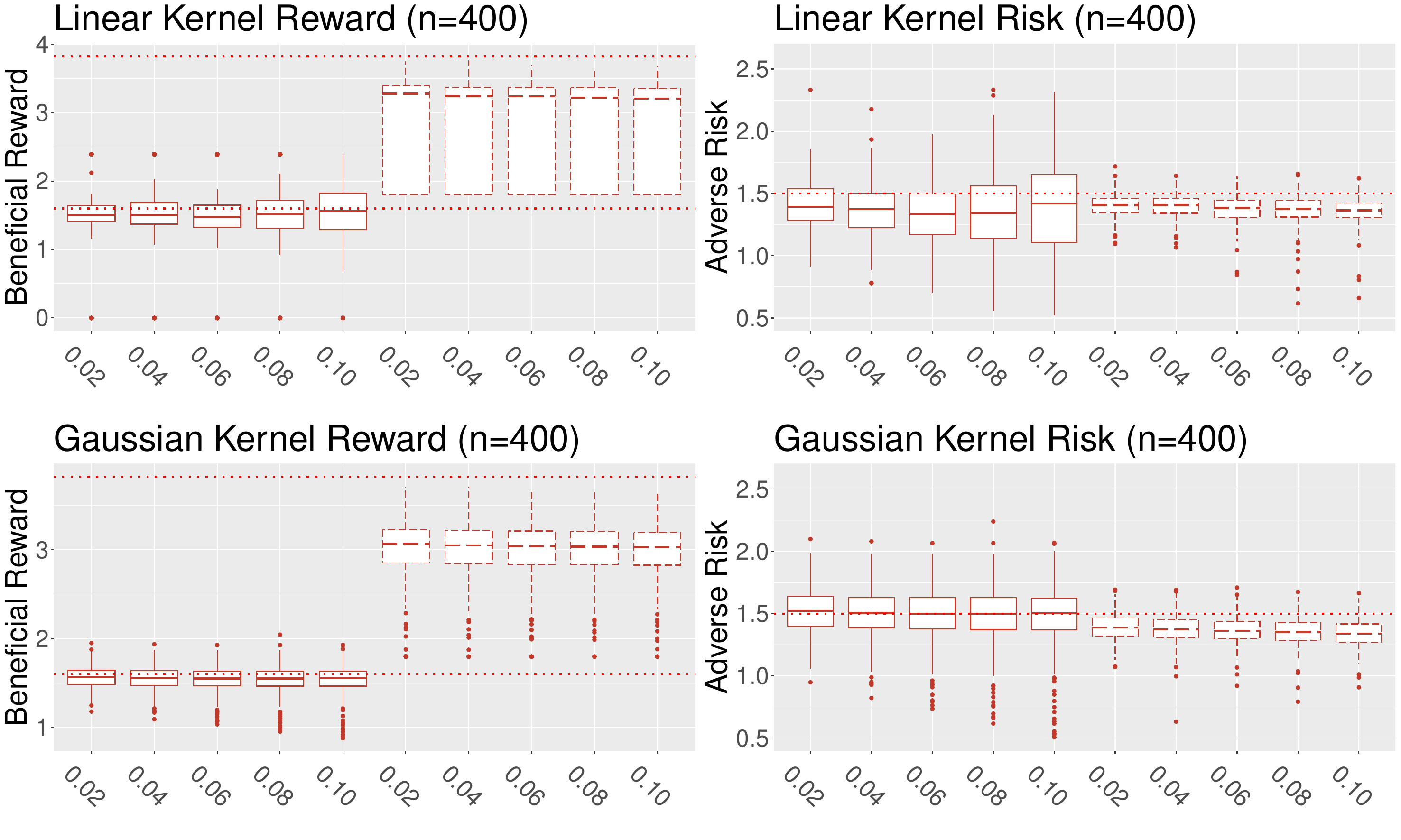}\\
  \vspace{-10pt}
  \includegraphics[scale = 0.3]{legend.pdf}
  \vspace{-10pt}
  \caption{\scriptsize Estimated reward/risk on independent testing data set for the additional simulation study with an unknown and unbalanced design. Results are reported in the way as \cref{fig:1,fig:2}.}
  \label{sfig:3}
\end{figure}

\afterpage{
\begin{landscape}
  \renewcommand{\arraystretch}{1.5}
  \begin{table}[t]
    \centering
    \scriptsize
    \caption{\label{stable2}\scriptsize Estimated reward/risk on independent testing data of the additional simulation study with an unknown and unbalanced design under 3 different methods using linear/Gaussian kernel. Results are reported in the same format as \cref{table1}.}
    \begin{tabular*}{\linewidth}{@{\extracolsep{\fill}}cccccccccc}
      \hline
      &        & \multicolumn{4}{c}{Linear Kernel}                               & \multicolumn{4}{c}{Gaussian Kernel}                             \\
      \cline{3-10}
      $\eta$  & Method & Reward - II & Risk - II & Cumulative Reward & Risk - I & Reward - II & Risk - II & Cumulative Reward & Risk - I \\
      \hline
       0.02 & BR-DTRs & 1.515(0.115) & 1.428(0.166) & 3.276(0.169) & 1.360(0.101)   & 1.563(0.077) & 1.519(0.117) & 3.057(0.180) & 1.386(0.071) \\
       0.02 & Naive  & ---          & ---          & 2.973(0.206) & 1.372(0.093)   & ---          & ---          & 2.923(0.169) & 1.378(0.080) \\
       0.04 & BR-DTRs & 1.501(0.139) & 1.409(0.206) & 3.247(0.188) & 1.354(0.112)   & 1.555(0.082) & 1.503(0.121) & 3.045(0.182) & 1.373(0.070) \\
       0.04 & Naive  & ---          & ---          & 2.932(0.227) & 1.355(0.100)   & ---          & ---          & 2.901(0.172) & 1.357(0.080) \\
       0.06 & BR-DTRs & 1.486(0.166) & 1.381(0.246) & 3.223(0.176) & 1.314(0.129)   & 1.551(0.082) & 1.496(0.123) & 3.045(0.179) & 1.360(0.068) \\
       0.06 & Naive  & ---          & ---          & 2.910(0.232) & 1.328(0.103)   & ---          & ---          & 2.886(0.174) & 1.334(0.085) \\
       0.08 & BR-DTRs & 1.504(0.206) & 1.407(0.298) & 3.220(0.202) & 1.317(0.128)   & 1.550(0.085) & 1.493(0.129) & 3.035(0.190) & 1.351(0.069) \\
       0.08 & Naive  & ---          & ---          & 2.918(0.234) & 1.321(0.110)   & ---          & ---          & 2.870(0.180) & 1.319(0.087) \\
       0.1  & BR-DTRs & 1.549(0.261) & 1.472(0.388) & 3.201(0.231) & 1.302(0.130)   & 1.549(0.084) & 1.498(0.127) & 3.028(0.177) & 1.336(0.075) \\
       0.1  & Naive  & ---          & ---          & 2.898(0.240) & 1.299(0.112)   & ---          & ---          & 2.855(0.181) & 1.304(0.088) \\
            & AOWL   & 2.382(0.030) & 2.798(0.001) & 5.309(0.000) & 2.876(0.000)   & 2.400(0.012) & 2.798(0.001) & 5.309(0.000) & 2.876(0.000) \\
      \hline
    \end{tabular*}
  \end{table}

  \begin{table}[t]
    \centering
    \scriptsize
    {\caption{\label{stable3}\scriptsize Simulation results of the additional personalized promotion example. Results are reported in the same format as \cref{table1}.}}
    \begin{tabular*}{\linewidth}{@{\extracolsep{\fill}}llcccccc}
      \hline
      $\tau$ & Kernel & Cumulative Reward & Wave 1 Cost & Wave 2 Cost & Wave 3 Cost & Wave 4 Cost \\
      \hline
      4 & Linear & 14.358(2.455) & 4.074(0.942) & 4.059(0.086) & 4.150(0.112) & 4.242(0.236) \\
      4 & Gaussian & 16.896(1.534) & 5.074(0.423) & 4.059(0.086) & 4.163(0.098) & 4.294(0.228) \\
      5 & Linear & 17.780(1.208) & 5.204(0.362) & 4.059(0.086) & 5.508(0.338) & 5.495(0.216) \\
      5 & Gaussian & 21.353(1.001) & 6.146(0.287) & 5.210(0.433) & 5.562(0.221) & 5.428(0.263) \\
      6 & Linear & 21.532(0.870) & 6.168(0.316) & 6.213(0.215) & 6.468(0.222) & 6.360(0.205) \\
      6 & Gaussian & 23.114(0.620) & 6.859(0.289) & 6.208(0.192) & 6.447(0.232) & 6.299(0.191) \\
      \hline
      $\infty$ &   &  25.394 & 8.972 & 7.034 & 7.186 & 7.051 \\
      \hline
    \end{tabular*}
  \end{table}
\end{landscape}}

% Note: in this sample, the section number is hard-coded in. Following
% proper LaTeX conventions, it should properly be coded as a reference:

%In this appendix we prove the following theorem from
%Section~\ref{sec:textree-generalization}:

% \vskip 0.2in
\newpage
\bibliography{BRDTR_jmlr}

\end{document}